\documentclass{article}

\usepackage{a4wide}
\usepackage{amsmath}
\usepackage{amsfonts}
\usepackage{mathrsfs}
\usepackage{amssymb}
\usepackage{amsthm}
\newtheorem{thm}{Theorem}
\newtheorem{prop}{Proposition}
\newtheorem{remark}{Remark}

\newtheorem{lemma}{Lemma}

\newcommand{\bbr}{{\mathbb R}}
\newcommand{\bbs}{{\mathbb S}}

\newcommand{\p}{\hat{p}}
\newcommand{\q}{\hat{q}}

\usepackage{authblk}

\begin{document}


\title{Future global existence of homogeneous solutions to the Einstein-Boltzmann system with soft potentials}

\author[1]{Ho Lee\footnote{holee@khu.ac.kr}}
\author[2]{Ernesto Nungesser\footnote{em.nungesser@upm.es}}
\affil[1]{Department of Mathematics and Research Institute for Basic Science, College of Sciences, Kyung Hee University, Seoul 02447, Republic of Korea}
\affil[2]{M2ASAI, Universidad Polit\'{e}cnica de Madrid, ETSI Navales, Avda.~de la Memoria, 4, Madrid 28040, Spain}

\maketitle

\begin{abstract}
We study the Cauchy problem for the Einstein-Boltzmann system with soft potentials in a cosmological setting. We assume the Bianchi I symmetry to describe a spatially homogeneous, but anisotropic universe and consider a cosmological constant $ \Lambda > 0 $ to describe an accelerated expansion of the universe. For the Boltzmann equation we introduce a new weight function and apply the method of Illner and Shinbrot to obtain the future global existence of spatially homogeneous, small solutions. For the Einstein equations we assume that the initial value of the Hubble variable is close to $ ( \Lambda / 3 )^{ 1 / 2 } $. We obtain the future global existence and asymptotic behavior of spatially homogeneous solutions to the Einstein-Boltzmann system with soft potentials.
\end{abstract}


%
\section{Introduction}

The standard model in cosmology is based on the Friedmann-Lema{\^i}tre-Robertson-Walker (FLRW) model of the Universe. These solutions were discovered soon after the theory of general relativity was established. However, it was not until relatively recently that the stability of these solutions was proven in the presence of a cosmological constant. It was first shown in the vacuum case and in the presence of an electromagnetic or Yang-Mills fields \cite{HF1, HF2}, and much more recently in the presence of a fluid \cite{RodSpeck, Speck} and a collisionless gas \cite{Ringstrom}. The case concerning a gas with collisions still remains open.  

Clearly the most common assumption in cosmology is to model the universe as a fluid with a linear equation of state, but there are situations where it is more appropriate to model the universe as a gas. In particular there are some indications that a kinetic picture might be a better description for cosmic large-scale structure formation \cite{Kozlikin}. In this paper, the universe will be modeled as a gas with collisions. In the present universe collisions seem to be rare, but they were very common in the early universe, for instance in the epoch of recombination, and have decreased as the universe expanded. Collisions will probably be vanishing towards the future, and this will be studied in a mathematically rigorous way by considering the Boltzmann equation. For more about the physical motivation, we refer to Chapter 6 of \cite{Huterer}, where the Boltzmann equation is used to understand baryogenesis, and to \cite{EB}, where the linear Einstein-Boltzmann equations are considered to describe the evolution of perturbations in the universe. In this paper, we will study small solutions to the Boltzmann equation, which are consistent with rare collisions in the present universe, in a spatially homogeneous setting, more specifically for Bianchi I spacetimes.

Bianchi I spacetimes are of interest in their own right, since they are a generalization of the spatially flat FLRW model.  In the vacuum case one can solve the Einstein equations with Bianchi I symmetry to obtain the so-called Kasner solutions which play an important role when considering the asymptotic behavior towards the direction of the initial singularity \cite{FRS}. Bianchi I spacetimes are also important when modeling the present universe and might solve discrepancies related to key cosmological parameters \cite{Akarsuetal}.

Concerning homogeneous spacetimes and a cosmological constant an important result was obtained by Robert Wald \cite{Wald83} by just assuming certain energy conditions. The qualitative asymptotic behavior was obtained. However in general to prove the global existence of solutions and to obtain a more detailed asymptotic behavior one will need more than energy conditions, and this was achieved for instance by Hayoung Lee in the Vlasov case \cite{HL}. The Boltzmann case is much more complicated since one needs to handle the collision integral, which is highly non-trivial even in the Newtonian case. For instance, the Boltzmann case was considered in \cite{LLN23, LN172}, but the results are only available in the case of Israel particles. The scattering cross-section for Israel particles was first introduced in \cite{I63} as a special, mathematically tractable form of cross-section in order to explicitly evaluate some thermodynamic quantities. However, in the real universe, there are many different kinds of particles and interactions between them, so it is necessary to extend the results to more general cases of scattering cross-sections. In this paper, we will obtain the global existence and asymptotic behavior of small solutions to the Boltzmann equation with scattering cross-sections in a broad range of soft potentials.

In this paper we study the Einstein-Boltzmann system. By the Boltzmann equation matter is described as a collection of particles, where the particles interact with each other through binary collisions. There is a rough classification into hard and soft potentials for the scattering cross-sections. We will consider the soft case. Our main interest is to study the Boltzmann equation in an expanding universe. The particles will be receding from each other, and the collision frequency will be decreasing as time evolves. Hence, we expect that solutions will be small in the sense that they are close to vacuum, and will apply the argument of Illner and Shinbrot \cite{IS84} to obtain the global existence of small solutions. However, the argument of \cite{IS84} does not directly apply to the Bianchi case. First, in a Bianchi cosmology the universe is assumed to be spatially homogeneous. Hence, we cannot expect the dispersion of solutions. For instance, the main idea of \cite{IS84} is to use the following symmetry:
\[
| x + t p' |^2 + | x + t q' |^2 = | x + t p |^2 + | x + t q |^2 ,
\]
but this does not apply in a spatially homogeneous case (see also \cite{Glassey} for more details, and \cite{Glassey06, Strain10} for the relativistic case). Second, we thus have to use a weight function which depends only on $ t $ and $ p $ to estimate the collision integral. We may consider the usual weight function, given by
\[
\mu = e^{ p^0 } ,
\]
to obtain the integrability on $ \bbr^3_p $. However, in a cosmological case the particle energy $ p^0 $ is in general decreasing in time, so the integrability is not uniform in time, and the estimates of the collision integral might not be uniform in time. Hence, we need to find a new weight function. In this paper we will introduce a new weight function which works nicely in the cosmological case.

Let us briefly introduce the main idea of this paper. We will follow the argument of \cite{IS84} for the Boltzmann equation. Once we have an appropriate weight function, the global existence will be obtained by assuming small initial data in the sense that the distribution function is small in a suitable weighted norm and that the Hubble variable is close to its vacuum value $ ( \Lambda / 3 )^{ 1 / 2 } $. Let $ w $ be a weight function. First, we require that $ w $ satisfy the following symmetry:
\[
w ( t , p' ) w ( t , q' ) = w ( t , p ) w ( t , q ) ,
\]
where $ p' $ and $ q' $ are the post-collision momenta for given $ p $ and $ q $. Then, we multiply the Boltzmann equation (see \eqref{Boltzmann}--\eqref{collision}) by $ w $ to obtain the following estimate:
\[
\frac{ \partial ( w f ) }{ \partial t } - \frac{ \partial w }{ \partial t } f \leq ( \det g )^{ - \frac12 } \| wf \|^2 \int_{ \bbr^3 } \cdots \, w^{ - 1 } ( q ) \, d q ,
\]
where $ \| \cdot \| $ is a suitable $ L^\infty $-norm. Secondly, we require that $ w ^{ - 1 } $ be integrable on $ \bbr^3 $. For instance, the $ \mu $ above satisfies this condition. However, $ q^0 $ is decreasing in time, so $ \mu^{ - 1 } ( q ) $ is increasing in time, and the right hand side might not be integrable in time. Indeed, we have
\[
\int_{ \bbr^3 } \mu^{ - 1 } ( q ) \, d q \approx ( \det g )^{ \frac12 } ,
\]
so we cannot make use of the decreasing factor $ ( \det g )^{ - \frac12 } $. Hence, thirdly, we require that $ w $ be increasing, or at least decreasing slowly in time. Finally, we have to control the second term on the left hand side. The easiest way to handle it is to require that $ w $ be decreasing in time. Then, the second term on the left hand side will be non-negative for $ f \geq 0 $, and $ w f $ will be estimated globally in time. Now, we have to find a weight function $ w $ satisfying these four conditions. Such a weight function $ w $ will be introduced in this paper, and this will be the main ingredient of the present work.

In this paper we will obtain the global existence of spatially homogeneous solutions to the Einstein-Boltzmann system with soft potentials by using the new weight function. In \cite{L13}, one of the authors of the present paper studied the Boltzmann equation in the FLRW case, but used a weight function which does not satisfy the first condition above, so a certain technical assumption on the scattering cross-section, which is not physically well-motivated, had to be imposed. The argument was applied to the Bianchi I case in \cite{LN171} by using a similar weight function. In \cite{LN172}, we could remove the technical assumption of \cite{L13} by using a different weight function. However, the weight function does not satisfy the third condition, and we could only obtain the existence in a special case of scattering cross-sections, called the Israel particles. The argument of \cite{LN172} has recently been extended to the Bianchi cases in \cite{LLN23}. In the work of the present paper, we have achieved significant improvements compared to our previous results in the sense that the unnatural assumptions in \cite{L13, LN171} on the scattering cross-section have been removed, and the restriction to Israel particles in \cite{LLN23, LN172} has been generalized to apply to a broad range of soft potentials. These improvements were made possible by using our new weight function. We refer to \cite{BCB73, ND06, NT06} for other results with different assumptions on the scattering cross-sections.

The paper is organized as follows. In Section \ref{sec EB}, we introduce the Einstein-Boltzmann system with Bianchi I symmetry. An explicit representation of post-collision momenta will be given in Section \ref{sec post collision}. The main result of this paper is given in Section \ref{sec M}. Our weight function will be introduced in Section \ref{sec wp}, and the main theorem will be given in Theorem \ref{main} in Section \ref{sec main}. Some preliminaries will be collected in Section \ref{sec prelim}. We first collect basic lemmas in Section \ref{sec basic}. In Section \ref{sec flrw}, we will consider the single Boltzmann equation in a given FLRW spacetime in order to show the main idea of this paper. This will be extended to the Bianchi I case in the following sections. The local existence will be obtained in Proposition \ref{prop local} in Section \ref{sec local}. In Section \ref{sec global}, we prove Theorem \ref{main} and obtain the global existence and the asymptotic behavior.

\subsection{Einstein-Boltzmann system with Bianchi I symmetry}\label{sec EB}
In this part we introduce the main equations of the paper. We assume the Bianchi I symmetry, in which case there exists a suitable frame $ {\bf W }^a $ such that the metric can be written as 
\begin{align*}
^4 g = - d t^2 + g_{ a b } { \bf W }^a { \bf W }^b ,
\end{align*}
where $ g_{ a b } ( t ) $ is a $ 3 \times 3 $ symmetric and positive definite matrix. Here and throughout the paper we assume the Einstein summation convention, where Greek indices run from $ 0 $ to $ 3 $, while Latin indices run from $ 1 $ to $ 3 $. The Einstein-Boltzmann system is given by 
\begin{align}
& \frac{ d g_{ a b } }{ d t } = 2 k_{ a b } , \label{Einstein1} \\ 
& \frac{ d k_{ a b } }{ d t } = 2 k_a^c k_{ b c } - k k_{ a b } + S_{ a b } + \frac12 ( \rho - S ) g_{ a b } + \Lambda g_{ a b } , \label{Einstein2} \\ 
& \frac{ \partial f }{ \partial t } = Q ( f , f ) , \label{Boltzmann} \\ 
& Q ( f , f ) = ( \det g )^{ - \frac12 } \int_{ \bbr^3 } \int_{ \bbs^2 } \frac{ h \sqrt{ s } }{ p^0 q^0 } \sigma ( h , \omega ) ( f ( p' ) f ( q' ) - f ( p ) f ( q ) ) \, d \omega \, d q , \label{collision} 
\end{align} 
where $ k_{ a b } ( t ) $ is the second fundamental form, and $ \Lambda > 0 $ is the cosmological constant. The matter terms $ S_{ a b } $, $ \rho $ and $ S $ are induced by the distribution function $ f ( t , p ) $ such that $ S_{ a b } = T_{ a b } $, $ \rho = T_{ 0 0 } $ and $ S = g^{ a b } S_{ a b } $, where $ T_{ \alpha \beta } $ is the stress-energy-momentum tensor defined by
\begin{align}
T_{ \alpha \beta } = ( \det g )^{ - \frac12 } \int_{ \bbr^3 } f ( t , p ) \frac{ p_\alpha p_\beta }{ p^0 } \, d p . \label{stress}
\end{align}
Here, $ \det g $ is the determinant of the $ 3 \times 3 $ matrix $ g_{ a b } $. In the Bianchi I case the constraint equation is given by 
\begin{align}
- k_{ a b } k^{ a b } + k^2 & = 2 \rho + 2 \Lambda . \label{constraint}
\end{align}
Momentum variables in the Boltzmann equation \eqref{Boltzmann} with \eqref{collision} will be given with lower indices, i.e.~
\begin{align*}
p = ( p_1 , p_2 , p_3 ) \in \bbr^3 . 
\end{align*}
Then, we assume the unit mass shell condition $ g^{ a b } p_a p_b = - 1 $, which derives $ p^0 $ as 
\begin{align}
p^0 = \sqrt{ 1 + g^{ a b } p_a p_b } . \label{p^0}
\end{align}
The quantities $ h $ and $ s $ are the relative momentum and the square of the energy in the center of momentum system, which are given by 
\begin{align}
h & =\sqrt{ - ( p^0 - q^0 )^2 + g^{ a b } ( p_a - q_a ) ( p_b - q_b ) } , \label{h} \\
s & = ( p^0 + q^0 )^2 - g^{ a b } ( p_a + q_a ) ( p_b + q_b ) . \label{s}
\end{align}
The assumption on the scattering cross-section $ \sigma ( h , \omega ) $ will be given in Section \ref{sec wp}, and the parametrization of post-collision momenta $ p' $ and $ q' $ will be introduced in Section \ref{sec post collision}.

\subsubsection{Parametrization of post-collision momenta} \label{sec post collision}
In the special relativistic case the parametrization of post-collision momenta has been investigated by Glassey and Strauss in \cite{GS91} and by Strain in \cite{Strain10}. The following is a different type of parametrization of post-collision momenta, which will be used in this paper, and can be obtained by combining the arguments of \cite{GS91} and \cite{Strain10}: for $ { \hat p } , { \hat q } \in \bbr^3 $ and $ \omega \in \bbs^2 $,
\begin{align}
p'^0 & = p^0 + 2 \bigg( - q^0 \frac{ { \hat n } \cdot \omega }{ \sqrt{ s } } + { \hat q } \cdot \omega + \frac{ ( { \hat n } \cdot \omega ) ( { \hat n } \cdot { \hat q } ) }{ \sqrt{ s } ( n^0 + \sqrt{ s } ) } \bigg) \frac{ { \hat n } \cdot \omega }{ \sqrt{ s } }, \label{phat'^0} \\
q'^0 & = q^0 - 2 \bigg( - q^0 \frac{ { \hat n } \cdot \omega }{ \sqrt{ s } } + { \hat q } \cdot \omega + \frac{ ( { \hat n } \cdot \omega ) ( { \hat n } \cdot { \hat q } ) }{ \sqrt{ s } ( n^0 + \sqrt{ s } ) } \bigg) \frac{ { \hat n } \cdot \omega }{ \sqrt{ s } }, \label{qhat'^0}
\end{align}
and
\begin{align}
{ \hat p }' & = { \hat p } + 2 \bigg( - q^0 \frac{ { \hat n } \cdot \omega }{ \sqrt{ s } } + { \hat q } \cdot \omega + \frac{ ( { \hat n } \cdot \omega ) ( { \hat n } \cdot { \hat q } ) }{ \sqrt{ s } ( n^0 + \sqrt{ s } ) } \bigg) \bigg( \omega + \frac{ ( { \hat n } \cdot \omega ) { \hat n } }{ \sqrt{ s } ( n^0 + \sqrt{ s } ) } \bigg), \label{phat'} \\
{ \hat q }' & = { \hat q } - 2 \bigg( - q^0 \frac{ { \hat n } \cdot \omega }{ \sqrt{ s } } + { \hat q } \cdot \omega + \frac{ ( { \hat n } \cdot \omega ) ( { \hat n } \cdot { \hat q } ) }{ \sqrt{ s } ( n^0 + \sqrt{ s } ) } \bigg) \bigg( \omega + \frac{ ( { \hat n } \cdot \omega ) { \hat n } }{ \sqrt{ s } ( n^0 + \sqrt{ s } ) } \bigg) . \label{qhat'}
\end{align}
Here, we assume the Minkowski metric, so that $ p^0 $ and $ q^0 $ are given by
\begin{align}
p^0 = \sqrt{ 1 + | { \hat p } |^2 }, \qquad q^0 = \sqrt{ 1 + | { \hat q } |^2 }, \label{phat^0}
\end{align}
where $ | { \hat p } | = \sqrt{ \sum_{ i = 1 }^3 ( { \hat p }_i )^2 } $, and $ { \hat n } $, $ n^0 $ and $ \sqrt{ s } $ are given by
\begin{align}
{ \hat n } = { \hat p } + { \hat q } , \qquad n^0 = p^0 + q^0 , \qquad s = ( p^0 + q^0 )^2 - | { \hat p } + { \hat q } |^2 = h^2 + 4 . \label{shat}
\end{align}
We note that $ h $ reduces to
\begin{align}
h & = \sqrt{ - ( p^0 - q^0 )^2 + | { \hat p } - { \hat q } |^2 } . \label{hhat}
\end{align}
The parametrization of \eqref{phat'^0}--\eqref{qhat'} was first introduced in \cite{LN172}. One advantage of this parametrization is that the derivative $ \partial_{ \hat p } { \hat p }' $ does not have growth in $ { \hat p } $ nor a singularity at $ { \hat p } = { \hat q } $. For instance, see Lemma 5 of \cite{LLN23} or Lemma \ref{lem dp'dp} of the present paper. We refer to the appendix of \cite{LN172} for more details.

The parametrization of \eqref{phat'^0}--\eqref{qhat'} can be extended to a general relativistic case by considering an orthonormal frame. Let $ { \bf e }^i = { e^i }_a { \bf W }^a $ be an orthonormal frame such that $ \delta^{ i j } = { e^i }_a { e^j }_b g^{ a b } $ and
\begin{align}
p_a = \p_i { e^i }_a , \label{pphat}
\end{align}
where $ p_a { \bf W }^a = { \hat p }_i { \bf e }^i $. Then, the post-collision momenta for given $ { \hat p } , { \hat q } \in \bbr^3 $ and $ \omega \in \bbs^2 $ will be given by \eqref{phat'^0}--\eqref{qhat'}. Now, using \eqref{pphat} we obtain
\begin{align}
p'^0 & = p^0 + 2 \bigg( - q^0 \frac{ n^a \omega_i { e^i }_a }{ \sqrt{ s } } + q^a \omega_i { e^i }_a + \frac{ n^a \omega_i { e^i }_a n_b q^b }{ \sqrt{ s } ( n^0 + \sqrt{ s } ) } \bigg) \frac{ n^c \omega_j { e^j }_c }{ \sqrt{ s } } , \label{p'^0} \\
q'^0 & = q^0 - 2 \bigg( - q^0 \frac{ n^a \omega_i { e^i }_a }{ \sqrt{ s } } + q^a \omega_i { e^i }_a + \frac{ n^a \omega_i { e^i }_a n_b q^b }{ \sqrt{ s } ( n^0 + \sqrt{ s } ) } \bigg) \frac{ n^c \omega_j { e^j }_c }{ \sqrt{ s } } , \label{q'^0} 
\end{align}
and
\begin{align}
p'_d & = p_d + 2 \bigg( - q^0 \frac{ n^a \omega_i { e^i }_a }{ \sqrt{ s } } + q^a \omega_i { e^i }_a + \frac{ n^a \omega_i { e^i }_a n_b q^b }{ \sqrt{ s } ( n^0 + \sqrt{ s } ) } \bigg) \bigg( \omega_j { e^j }_d + \frac{ n^c \omega_j { e^j }_c n_d }{ \sqrt{ s } ( n^0 + \sqrt{ s } ) } \bigg) , \label{p'} \\
q'_d & = q_d - 2 \bigg( - q^0 \frac{ n^a \omega_i { e^i }_a }{ \sqrt{ s } } + q^a \omega_i { e^i }_a + \frac{ n^a \omega_i { e^i }_a n_b q^b }{ \sqrt{ s } ( n^0 + \sqrt{ s } ) } \bigg) \bigg( \omega_j { e^j }_d + \frac{ n^c \omega_j { e^j }_c n_d }{ \sqrt{ s } ( n^0 + \sqrt{ s } ) } \bigg) , \label{q'} 
\end{align}
where $ \omega = ( \omega_1 , \omega_2 , \omega_3 ) \in \bbs^2 $ in the sense that $ \delta^{ i j } \omega_i \omega_j = 1 $. One can show that \eqref{p'^0}--\eqref{q'} satisfy the energy-momentum conservation:
\begin{align}
p'_\alpha + q'_\alpha = p_\alpha + q_\alpha . \label{conserved}
\end{align}
We remark that the conservation \eqref{conserved} implies, in particular
\begin{align}
\int_{ \bbr^3 } Q ( f , f ) p^0 \, d p = 0 , \label{Qconserved}
\end{align}
by using the symmetry of the collision operator.

We notice that the Euclidean inner product $ { \hat n } \cdot { \hat q } $ in \eqref{phat'^0}--\eqref{qhat'} has been rewritten as $ n_b q^b $ in \eqref{p'^0}--\eqref{q'}, while the inner products $ { \hat n } \cdot \omega $ and $ { \hat q } \cdot \omega $ have been replaced by $ n^a \omega_i { e^i }_a $ and $ q^a \omega_i { e^i }_a $, respectively, where we need an explicit formula for an orthonormal frame. In this paper, we assume that the orthonormal frame is given as follows.

\begin{remark}\label{rem ortho}
Throughout the paper we will use the following orthonormal frame: 
\begin{align*}
& ( { e^i }_a ) 
= \begin{pmatrix}
{ \bf e }^1 \\ 
{ \bf e }^2 \\ 
{ \bf e }^3 
\end{pmatrix}
= \begin{pmatrix}
\frac{ 1 }{ \sqrt{ g^{ 1 1 } } }  & 0 & 0 \\ 
\frac{ - g^{ 1 2 } }{ \sqrt{ g^{ 1 1 } ( g^{ 1 1 } g^{ 2 2 } - ( g^{ 1 2 } )^2 ) } } & \frac{ g^{ 1 1 } }{ \sqrt{ g^{ 1 1 } ( g^{ 1 1 } g^{ 2 2 } - ( g^{ 1 2 } )^2 ) } } & 0 \\
\frac{ g^{ 1 2 } g^{ 2 3 } - g^{ 1 3 } g^{ 2 2 } }{ \sqrt{ ( g^{ 1 1 } g^{ 2 2 } - ( g^{ 1 2 } )^2 ) ( \det g^{ - 1 } ) } } & \frac{ - g^{ 1 1 } g^{ 2 3 } + g^{ 1 2 } g^{ 1 3 } }{ \sqrt{ ( g^{ 1 1 } g^{ 2 2 } - ( g^{ 1 2 } )^2 ) ( \det g^{ - 1 } ) } } & \frac{ g^{ 1 1 } g^{ 2 2 } - ( g^{ 1 2 } )^2 }{ \sqrt{ ( g^{ 1 1 } g^{ 2 2 } - ( g^{ 1 2 } )^2 ) ( \det g^{ - 1 } ) } }
\end{pmatrix} , \\ 
& ( { e_i }^a ) 
= \begin{pmatrix}
{ \bf e }_1 & { \bf e }_2 & { \bf e }_3 
\end{pmatrix}
= \begin{pmatrix}
\frac{ g^{ 1 1 } }{ \sqrt{ g^{ 1 1 } } }  & 0 & 0 \\ 
\frac{ g^{ 1 2 } }{ \sqrt{ g^{ 1 1 } } } & \frac{ g^{ 1 1 } g^{ 2 2 } - ( g^{ 1 2 } )^2 }{ \sqrt{ g^{ 1 1 } ( g^{ 1 1 } g^{ 2 2 } - ( g^{ 1 2 } )^2 ) } } & 0 \\
\frac{ g^{ 1 3 } }{ \sqrt{ g^{ 1 1 } } } & \frac{ g^{ 1 1 } g^{ 2 3 } - g^{ 1 2 } g^{ 1 3 } }{ \sqrt{ g^{ 1 1 } ( g^{ 1 1 } g^{ 2 2 } - ( g^{ 1 2 } )^2 ) } } & \frac{ \det g^{ - 1 } }{ \sqrt{ ( g^{ 1 1 } g^{ 2 2 } - ( g^{ 1 2 } )^2 ) ( \det g^{ - 1 } ) } } 
\end{pmatrix} ,
\end{align*}
which satisfy $ { e^i }_a { e^j }_b g^{ a b } = \delta^{ i j } $ and $ { e_i }^a { e_j }^b g_{ a b } = \delta_{ i j } $. Here, $ \det g^{ - 1 } $ denotes
\begin{align*}
\det g^{ - 1 } = \det ( g^{ a b } )_{ a , b = 1 , 2 , 3 } .
\end{align*}
Similarly, $ e^{ - 1 } $ and $ e $ will denote the following $ 3 \times 3 $ matrices:
\begin{align}
e^{ - 1 } = ( { e^i }_a )_{ i , a = 1 , 2 , 3 } , \qquad e = ( { e_i }^a )_{ i , a = 1 , 2 , 3 } , \label{e}
\end{align}
respectively. We note that $ e^{ - 1 } $ and $ e $ are matrix inverses of each other.
\end{remark}

\subsection{Main result}\label{sec M}
In this part we state the main result of the paper. In Section \ref{sec wp} we introduce the soft potentials and the weight function. In Section \ref{sec main} we collect the notations and state the main result.

\subsubsection{Soft potentials and weight function}\label{sec wp}
In this paper we will assume that the scattering cross-section $ \sigma ( h , \omega ) $ is independent of the angular variable $ \omega \in \bbs^2 $, but depends only on the relative momentum $ h $. More precisely, we are interested in the case of soft potentials, given by
\begin{align}\label{potential}
\sigma ( h , \omega ) = h^{ - b } , \qquad 0 \leq b < 3 . 
\end{align}
Now, let us consider the weight function. We first define
\begin{align}
Z_\eta = e^{ ( 1 - \eta ) \gamma t } , \qquad \gamma = \sqrt{ \frac{ \Lambda }{ 3 } } , \label{Z}
\end{align}
where $ 0 \leq \eta < 1 $ will be determined later. For $ \eta = 0 $, we will write
\begin{align*}
Z = e^{ \gamma t } .
\end{align*}
This is basically the same as the scale factor $ R $ in the FLRW case, but it should be modified as in \eqref{Z} in order to control the anisotropy in the Bianchi case. The following is our weight function:
\begin{align}\label{weight}
w ( t , p ) = e^{ Z_\eta ( p^0 - 1 ) } . 
\end{align}
Note that the weight depends also on the metric through $ p^0 $. We observe that
\begin{align}
w ( t , p' ) w ( t , q' ) = w ( t , p ) w ( t , q ) , \label{wconserved}
\end{align}
since $ p'^0 + q'^0 = p^0 + q^0 $ for each $ t $. We also introduce $ \varphi : [ 0 , \infty ) \to \bbr $ defined by
\begin{align}
\varphi ( r ) = \frac{ \min ( r , \sqrt{ r } ) }{ 1 + \sqrt{ 2 } } . \label{varphi} 
\end{align} 
This will be useful in estimating the weight function.

\subsubsection{Main theorem}\label{sec main}
In this paper we will consider the Einstein equations for the inverse $ g^{ a b } $ of the metric $ g_{ a b } $ and the second fundamental form $ k^{ a b } $ with indices raised:
\begin{align}
& \frac{ d g^{ a b } }{ d t } = - 2 k^{ a b } , \label{E1} \\ 
& \frac{ d k^{ a b } }{ d t } = - 2 k^a_c k^{ b c } - k k^{ a b } + S^{ a b } + \frac12 ( \rho - S ) g^{ a b } + \Lambda g^{ a b } . \label{E2} 
\end{align} 
For the Boltzmann equation we will consider the equation for the weighted distribution function. Multiplying the equation \eqref{Boltzmann} by the weight \eqref{weight}, we obtain
\begin{align*}
\frac{ \partial ( w f ) }{ \partial t } - \frac{ \partial w }{ \partial t } f = ( \det g )^{ - \frac12 } \iint \frac{ h^{ 1 - b } \sqrt{ s } }{ p^0 q^0 } ( w f ( p' ) w f ( q' ) - w f ( p ) w f ( q ) ) w^{ - 1 } ( q ) \, d \omega \, d q , 
\end{align*}
where the assumption \eqref{potential} has been applied, and note that 
\begin{align*}
\frac{ \partial w }{ \partial t } & = \left( \frac{ d Z_\eta }{ d t } ( p^0 - 1 ) + Z_\eta \frac{ \partial p^0 }{ \partial t } \right) w \\
& = \left( ( 1 - \eta ) \gamma ( p^0 - 1 ) - \frac{ k^{ a b } p_a p_b }{ p^0 } \right) Z_\eta w . 
\end{align*}
Now, we obtain the weighted Boltzmann equation by replacing $ w f $ with $ f $: 
\begin{align}
\frac{ \partial f }{ \partial t } + \lambda_\eta f & = \Gamma ( f , f ) , \label{B} \\
\Gamma ( f , f ) & = ( \det g )^{ - \frac12 } \int_{ \bbr^3 } \int_{ \bbs^2 } \frac{ h^{ 1 - b } \sqrt{ s } }{ p^0 q^0 } ( f ( p' ) f ( q' ) - f ( p ) f ( q ) ) w^{ - 1 } ( q ) \, d \omega \, d q , \label{Gamma}
\end{align}
where
\begin{align}
\lambda_\eta = - \left( ( 1 - \eta ) \gamma ( p^0 - 1 ) - \frac{ k^{ a b } p_a p_b }{ p^0 } \right) Z_\eta \label{lambda} . 
\end{align}
In the rest of the paper, the distribution function $ f $ will refer to the weighted one. Matter terms are given by the stress-energy-momentum tensor \eqref{stress}, which are given in terms of the weighted distribution function as follows: 
\begin{align}
S^{ a b } & = ( \det g )^{ - \frac12 } \int_{ \bbr^3 } f ( p ) \frac{ p^a p^b }{  p^0 } w^{ - 1 } ( p ) \, d p , \label{Sab} \\
\rho & = ( \det g )^{ - \frac12 } \int_{ \bbr^3 } f ( p ) p^0 w^{ - 1 } ( p ) \, d p , \label{rho} \\
S & = ( \det g )^{ - \frac12 } \int_{ \bbr^3 } f ( p ) \frac{ g^{ a b } p_a p_b }{  p^0 } w^{ - 1 } ( p ) \, d p , \label{S}
\end{align}
where $ d p = d p_1 \, d p_2 \, d p_3 $, and the dependence on $ t $ has been omitted for simplicity.

Let us introduce the Hubble variable $ H $ defined by 
\begin{align}
H = \frac13 k^{ a b } g_{ a b } . \label{H}
\end{align}
Then, the second fundamental form $ k^{ a b } $ can be decomposed as
\begin{align}
k^{ a b } = H g^{ a b } + \sigma^{ a b } , \label{k decom}
\end{align}
where $ \sigma^{ a b } $ is the trace-free part, called the shear, which measures the anisotropy. It is convenient to introduce the (Hubble normalized) shear variable $ F $:
\begin{align}
F = \frac{ \sigma_{ a b } \sigma^{ a b } }{ 4 H^2 } , \label{F}
\end{align}
which vanishes in the FLRW case. We can use the decomposition \eqref{k decom} to write the constraint \eqref{constraint} as follows:
\begin{align}
- \sigma_{ a b } \sigma^{ a b } + 6 H^2 = 2 \rho + 2 \Lambda . \label{constraint1}
\end{align}
This implies that $ 3 H^2 \geq \Lambda $, since $ \sigma_{ a b } \sigma^{ a b } $ and $ \rho $ are non-negative, for any solutions satisfying the constraint equation. We will first prove the local existence in Section \ref{sec local}, and then will use \eqref{constraint1} to obtain the global existence in Section \ref{sec global}.

Let us collect the notations. A metric $ g^{ a b } $ will be understood as a $ 3 \times 3 $ symmetric and positive definite matrix, and $ e^{ - 1 } $ and $ e $ will denote the $ 3 \times 3 $ matrices defined in Remark \ref{rem ortho} with the notations \eqref{e}. We define
\begin{align}
\| g \| = \max_{ a , b } | g_{ a b } | , \qquad \| g^{ - 1 } \| = \max_{ a , b } | g^{ a b } | , \qquad \| e \| = \max_{ i , a } | { e_i }^a | , \qquad \| e^{ - 1 } \| = \max_{ i , a } | { e^i }_a | . \label{norm g}
\end{align}
In a similar way we define for the second fundamental form
\begin{align}
\| k \| = \max_{ a , b } | k_{ a b } | , \qquad \| k^{ ** } \| = \max_{ a , b } | k^{ a b } | . \label{norm k}
\end{align}
Note that the norm of $ k^{ a b } $ can be estimated as
\begin{align*}
\| k^{ ** } \| \leq C \| g^{ - 1 } \|^2 \| k \| . 
\end{align*}
We also remark that
\begin{align*}
\| g^{ - 1 } \| \leq C ( \det g )^{ - 1 } \| g \|^2 , \qquad \| g \| \leq C ( \det g^{ - 1 } )^{ - 1 } \| g^{ - 1 } \|^2 . 
\end{align*}
Finally, we define the norm for the distribution function as follows:
\begin{align}\label{norm f}
\| f \|_{ L^\infty_r } = \sup_{ p \in \bbr^3 } \langle p \rangle^r | f ( p ) | , \qquad \left\| \frac{ \partial f }{ \partial p } \right\|_{ L^\infty_r } = \sum_{ a = 1 }^3 \left\| \frac{ \partial f }{ \partial p_a } \right\|_{ L^\infty_r } \qquad \| f \|_{ W_r^{ 1 , \infty } } = \| f \|_{ L^\infty_r } + \left\| \frac{ \partial f }{ \partial p } \right\|_{ L^\infty_r } , 
\end{align}
where the additional weight $ \langle p \rangle $ is defined by
\begin{align*}
\langle p \rangle = \sqrt{ 1 + | p |^2 } , \qquad | p | = \sqrt{ \sum_{ i = 1 }^3 ( p_i )^2 } , 
\end{align*}
which does not depend on the metric.

Now, the following is the main theorem of this paper. Initial data $ g_0^{ a b } $, $ k_0^{ a b } $ and $ f_0 $ for the (weighted) Einstein-Boltzmann system  \eqref{E1}--\eqref{B} will be given at some $ t_0 > 0 $, and the initial value of the Hubble variable $ H $ will be denoted by $ H_0 $.

\begin{thm}\label{main}
Suppose that $ g_0^{ a b } $ is a symmetric and positive definite matrix, $ k_0^{ a b } $ is a symmetric matrix satisfying $ \gamma < H_0 < ( 6 / 5 )^{ 1 / 2 } \gamma $, and $ 0 \leq f_0 \in W^{ 1 , \infty }_{ r + 2 } $ with $ r > 5 $. Let $ g_0^{ a b } $, $ k_0^{ a b } $ and $ f_0 $ be initial data of the equations \eqref{E1}--\eqref{B} with $ \eta = \eta_0 $, defined by
\begin{align}
\eta_0 = \sqrt{ 6 \left( 1 - \frac{ \gamma^2 }{ H_0^2 } \right) } , \label{eta0}
\end{align}
satisfying the constraint equation \eqref{constraint}. Then, there exists a small $ \varepsilon > 0 $ such that if $ \| f_0 \|_{ L^\infty_r } < \varepsilon $, then a unique global solution $ g^{ a b } $, $ k^{ a b } $ and $ f \geq 0 $ exists such that $ g^{ a b } $ and $ k^{ a b } $ are differentiable, $ f $ is continuous, and $ S^{ a b } $, $ \rho $ and $ S $ are differentiable. Moreover, the solution satisfies 
\begin{align}
\gamma \leq H ( t ) & \leq \gamma + C e^{ - 2 \gamma t } , \label{thm H} \\
F ( t ) & \leq C e^{ - 2 \gamma t } , \label{thm F} \\
f ( t , p ) & \leq C \varepsilon \langle p \rangle^{ - r } , \label{thm f} \\
\rho ( t ) & \leq C \varepsilon e^{ - 3 \gamma t } , \label{thm rho} \\
S ( t ) & \leq C \varepsilon e^{ - 5 \gamma t } , \label{thm S} \\
\frac{ S ( t ) }{ \rho ( t ) } & \leq C e^{ - 2 \gamma t } , \label{thm Srho}
\end{align}
and there exists a constant matrix $ G^{ a b } $ such that
\begin{align}
| Z^2 ( t ) g^{ a b } ( t ) - G^{ a b } | \leq C e^{ - \gamma t } , \label{thm g}
\end{align}
for all $ t \geq t_0 $.
\end{thm}

In Theorem \ref{main}, we assumed that $ H_0 $ is close to $ \gamma $ and $ f_0 $ is sufficiently small, and considered the weighted equation with $ \eta = \eta_0 $. These restrictions are only required to obtain the global existence. In Section \ref{sec local} we will obtain the local existence without those restrictions, and we will use them in Section \ref{sec global} to obtain the global existence.

\section{Preliminaries}\label{sec prelim}
We first collect several basic lemmas in Section \ref{sec basic}. In Section \ref{sec flrw}, we introduce the main idea of this paper in the FLRW case. The idea will be extended to the Bianchi case in Sections \ref{sec local} and \ref{sec global}.

\subsection{Basic estimates}\label{sec basic}
Let $ g^{ a b } $ be a symmetric and positive definite matrix, and $ h $ and $ \sqrt{ s } $ are defined by \eqref{h} and \eqref{s}. Then, we have
\begin{gather}
h \leq \sqrt{ s } \leq 2 \sqrt{ p^0 q^0 }, \label{hspq} \\
h^2 \leq g^{ a b } ( p_a - q_a ) ( p_b - q_b ) \leq p^0 q^0 h^2 . \label{hp-q} 
\end{gather}
These are simple computations, and we refer to Lemma 3 of \cite{LN171} for the proof. We also have by the energy conservation \eqref{conserved}, 
\begin{align}
p^0 \leq C p'^0 q'^0 , \label{pp'q'}
\end{align}
where the constant $ C $ does not depend on $ g^{ a b } $.

We obtain by direct calculations:
\begin{align}
\frac{ \partial p^0 }{ \partial p_a } & = \frac{ p^a }{ p^0 } , \label{dpdp} \\
\frac{ \partial h }{ \partial p_a } & = \frac{ q^0 }{ h } \left( \frac{ p^a }{ p^0 } - \frac{ q^a }{ q^0 } \right) , \label{dhdp} \\
\frac{ \partial \sqrt{ s } }{ \partial p_a } & = \frac{ q^0 }{ \sqrt{ s } } \left( \frac{ p^a }{ p^0 } - \frac{ q^a }{ q^0 } \right) . \label{dsdp}
\end{align}
We use the orthonormal frame given in Remark \ref{rem ortho} to have $ p^a = g^{ a b } p_b = g^{ a b } \p_i { e^i }_b $. Then, we obtain
\begin{align}
\left| \frac{ p^a }{ p^0 } \right| \leq \left| \frac{ g^{ a b } \p_i { e^i }_b }{ p^0 } \right| \leq C \| g^{ - 1 } \| \| e^{ - 1 } \| , \qquad \left| \frac{ p_a }{ p^0 } \right| \leq \left| \frac{ \p_i { e^i }_a }{ p^0 } \right| \leq C \| e^{ - 1 } \| , \label{dpdp est}
\end{align}
which can be used to estimate \eqref{dpdp}. In a similar way, \eqref{dhdp} and \eqref{dsdp} are estimated as follows.

\begin{lemma}\label{lem dhdp}
Let $ g^{ a b } $ be a symmetric and positive definite matrix, and $ e^{ - 1 } $ be the matrix given in Remark \ref{rem ortho}. Then, we have
\begin{align}
\left| \frac{ \partial h }{ \partial p_a } \right| & \leq C q^0 \sqrt{ p^0 q^0 } \| g^{ - 1 } \| \| e^{ - 1 } \| , \\
\left| \frac{ \partial \sqrt{ s } }{ \partial p_a } \right| & \leq C q^0 \sqrt{ p^0 q^0 } \| g^{ - 1 } \| \| e^{ - 1 } \| , 
\end{align}
where $ C $ does not depend on $ g^{ a b } $. 
\end{lemma}
\begin{proof}
We use the orthonormal frame to write
\begin{align*}
\frac{ q^0 }{ h } \left( \frac{ p^a }{ p^0 } - \frac{ q^a }{ q^0 } \right) = \frac{ q^0 }{ h } g^{ a b } \left( \frac{ \p_i }{ p^0 } - \frac{ \q_i }{ q^0 } \right) { e^i }_b .
\end{align*}
Here, we apply \eqref{hp-q} to obtain
\begin{align}
\left| \frac{ \p_i }{ p^0 } - \frac{ \q_i }{ q^0 } \right| \leq | \p - \q | \leq h \sqrt{ p^0 q^0 } . \label{p-q}
\end{align}
Hence, we obtain from \eqref{dhdp}
\begin{align*}
\left| \frac{ \partial h }{ \partial p_a } \right| = \left| \frac{ q^0 }{ h } \left( \frac{ p^a }{ p^0 } - \frac{ q^a }{ q^0 } \right) \right| \leq C q^0 \sqrt{ p^0 q^0 } \| g^{ - 1 } \| \| e^{ - 1 } \| .
\end{align*}
The estimate for \eqref{dsdp} can be obtained by the same argument, and this completes the proof.
\end{proof}

%

%
\begin{lemma}\label{lem wlower}
The weight function $ w $ satisfies 
\begin{align*}
w ( t , p ) \geq \exp ( Z_\eta \varphi ( g^{ a b } p_a p_b ) ) , 
\end{align*}
where $ Z_\eta $ and $ \varphi $ are defined by \eqref{Z} and \eqref{varphi}, respectively. 
\end{lemma}
\begin{proof}
Recall that the weight function $ w $ is given by 
\begin{align*}
w ( t , p ) = e^{ Z_\eta ( p^0 - 1 ) } , 
\end{align*}
where we have 
\begin{align*}
p^0 - 1 = \frac{ g^{ a b } p_a p_b }{ p^0 + 1 } . 
\end{align*}
For $ g^{ a b } p_a p_b \leq 1 $, we have 
\begin{align*}
\frac{ g^{ a b } p_a p_b }{ p^0 + 1 } = \frac{ g^{ a b } p_a p_b }{ \sqrt{ 1 + g^{ c d } p_c p_d } + 1 } \geq \frac{ g^{ a b } p_a p_b }{ 1 + \sqrt{ 2 } } , 
\end{align*}
while, for $ g^{ a b } p_a p_b \geq 1 $ we have 
\begin{align*}
\frac{ g^{ a b } p_a p_b }{ p^0 + 1 } = \frac{ g^{ a b } p_a p_b }{ \sqrt{ 1 + g^{ c d } p_c p_d } + 1 } \geq \frac{ g^{ a b } p_a p_b }{ \sqrt{ g^{ c d } p_c p_d + g^{ c d } p_c p_d } + \sqrt{ g^{ c d } p_c p_d } } = \frac{ \sqrt{ g^{ a b } p_a p_b } }{ 1 + \sqrt{ 2 } } . 
\end{align*}
Hence, we obtain 
\begin{align*}
p^0 - 1 \geq \varphi ( g^{ a b } p_a p_b ) , 
\end{align*}
where $ \varphi $ is the function defined by \eqref{varphi}. This completes the proof of the lemma. 
\end{proof}

\begin{lemma}\label{lem wdecay}
Let $ g^{ a b } $ be a symmetric and positive definite matrix, and $ k^{ a b } $ be a symmetric matrix. Consider the decomposition \eqref{k decom} with $ H $ defined by \eqref{H} and the shear variable $ F $ defined by \eqref{F}. Suppose that $ H \geq \gamma $ and $ F \leq \eta^2 / 4 $ for some $ 0 \leq \eta < 1 $. Then, we have $ \lambda_\eta \geq 0 $ and
 \begin{align*}
\left| \frac{ \partial \lambda_\eta }{ \partial p_a } \right| \leq C ( \| g^{ - 1 } \| \| e^{ - 1 } \| + \| k^{ ** } \| \| e^{ - 1 } \| + \| g^{ - 1 } \| \| k^{ ** } \| \| e^{ - 1 } \|^3 ) Z_\eta ,
\end{align*}
for any $ a = 1 , 2 , 3 $.
\end{lemma}
\begin{proof}
Recall that $ \lambda_\eta $ is defined by
\begin{align*}
\lambda_\eta = - \left( ( 1 - \eta ) \gamma ( p^0 - 1 ) - \frac{ k^{ a b } p_a p_b }{ p^0 } \right) Z_\eta .
\end{align*}
Since $ k^{ a b } = H g^{ a b } + \sigma^{ a b } $, we have
\begin{align*}
( 1 - \eta ) \gamma ( p^0 - 1 ) - \frac{ k^{ a b } p_a p_b }{ p^0 } & = ( 1 - \eta ) \gamma ( p^0 - 1 ) - \frac{ H g^{ a b } p_a p_b }{ p^0 } - \frac{ \sigma^{ a b } p_a p_b }{ p^0 } \\
& \leq ( 1 - \eta ) \gamma ( p^0 - 1 ) - \frac{ H g^{ a b } p_a p_b }{ p^0 } + \frac{ ( \sigma^{ a b } \sigma_{ a b } )^{ \frac12 } g^{ c d } p_c p_d }{ p^0 } \\
& \leq ( 1 - \eta ) \gamma \left( p^0 - \frac{ 1 }{ p^0 } \right) - ( H - 2 H F^{ \frac12 } ) \left( p^0 - \frac{ 1 }{ p^0 } \right) \\
& \leq ( 1 - \eta ) ( \gamma - H ) \left( p^0 - \frac{ 1 }{ p^0 } \right) \\
& \leq 0 ,
\end{align*}
so that we have $ \lambda_\eta \geq 0 $. We take $ p $ derivatives of $ \lambda_\eta $ to obtain
\begin{align*}
\frac{ \partial \lambda_\eta }{ \partial p_a } & = - \left( ( 1 - \eta ) \gamma \frac{ \partial p^0 }{ \partial p_a } - 2 \frac{ k^{ a b } p_b }{ p^0 } + \frac{ k^{ b c } p_b p_c }{ ( p^0 )^2 } \frac{ \partial p^0 }{ \partial p_a } \right) Z_\eta \\
& = - \left( ( 1 - \eta ) \gamma \frac{ p^a }{ p^0 } - 2 \frac{ k^{ a b } p_b }{ p^0 } + \frac{ k^{ b c } p_b p_c p^a }{ ( p^0 )^3 } \right) Z_\eta .
\end{align*}
Then, we apply \eqref{dpdp est} to obtain
\begin{align*}
\left| \frac{ \partial \lambda_\eta }{ \partial p_a } \right| \leq C ( \| g^{ - 1 } \| \| e^{ - 1 } \| + \| k^{ ** } \| \| e^{ - 1 } \| + \| k^{ ** } \| \| g^{ - 1 } \| \| e^{ - 1 } \|^3 ) Z_\eta ,
\end{align*}
which completes the proof of the lemma.
\end{proof}

We will need to compute several derivatives of momentum variables with respect to the metric. We need to use the following formula:
\begin{align}
\frac{ \partial g^{ a b } }{ \partial g^{ u v } } = \frac{ 1 }{ 2^{ \delta_{ u v } } } \left( \delta^a_u \delta^b_v + \delta^b_u \delta^a_v \right) , \label{est dgdg}
\end{align}
with no summation over repeated indices. We first use \eqref{p^0} to obtain
\begin{align}
\frac{ \partial p^0 }{ \partial g^{ u v } } = \frac{ 1 }{ 2 p^0 } \frac{ \partial g^{ a b } }{ \partial g^{ u v } } p_a p_b , \label{est dpdg0}
\end{align}
which implies that
\begin{align}
\left| \frac{ \partial p^0 }{ \partial g^{ u v } } \right| \leq C \max_{ a , b } \left| \frac{ p_a  p_b }{ p^0 } \right| \leq C \| e^{ - 1 } \| \langle p \rangle , \label{est dpdg}
\end{align}
where we used \eqref{dpdp est}. In a similar way we obtain
\begin{align}
\left| \frac{ \partial }{ \partial g^{ u v } } \left( \frac{ 1 }{ p^0 } \right) \right| \leq C \| e^{ - 1 } \| \frac{ \langle p \rangle }{ ( p^0 )^2 } . \label{est d1/pdg}
\end{align}
We also need the following computation:
\begin{align*}
\frac{ \partial h }{ \partial g^{ u v } } & = \frac{ \partial }{ \partial g^{ u v } } \left( \sqrt{ - 2 + 2 p^0 q^0 - 2 g^{ a b } p_a q_b } \right) \nonumber \\ 
& = \frac{ 1 }{ 2 h } \left( 2 \frac{ \partial p^0 }{ \partial g^{ u v } } q^0 + 2 p^0 \frac{ \partial q^0 }{ \partial g^{ u v } } - 2 \frac{ \partial g^{ a b } }{ \partial g^{ u v } } p_a q_b \right) \nonumber \\ 
& = \frac{ 1 }{ h } \left( \frac{ 1 }{ 2 p^0 } \frac{ \partial g^{ a b } }{ \partial g^{ u v } } p_a p_b q^0 + p^0 \frac{ 1 }{ 2 q^0 } \frac{ \partial g^{ a b } }{ \partial g^{ u v } } q_a q_b - \frac{ \partial g^{ a b } }{ \partial g^{ u v } } p_a q_b \right) \nonumber \\ 
& = \frac{ p^0 q^0 }{ 2 h } \left( \frac{ p_a p_b }{ ( p^0 )^2 } + \frac{ q_a q_b }{ ( q^0 )^2 } - 2 \frac{ p_a q_b }{ p^0 q^0 } \right) \frac{ 1 }{ 2^{ \delta_{ u v } } } ( \delta^a_u \delta^b_v + \delta^b_u \delta^a_v ) \nonumber \\
& = \frac{ p^0 q^0 }{ 2 h } \frac{ 1 }{ 2^{ \delta_{ u v } } } \left( \frac{ 2 p_u p_v }{ ( p^0 )^2 } + \frac{ 2 q_u q_v }{ ( q^0 )^2 } - \frac{ 2 p_u q_v }{ p^0 q^0 } - \frac{ 2 q_u p_v }{ p^0 q^0 } \right) \nonumber \\ 
& = \frac{ p^0 q^0 }{ h } \frac{ 1 }{ 2^{ \delta_{ u v } } } \left( \frac{ p_u }{ p^0 } - \frac{ q_u }{ q^0 } \right) \left( \frac{ p_v }{ p^0 } - \frac{ q_v }{ q^0 } \right) \nonumber \\ 
& = \frac{ p^0 q^0 }{ h } \frac{ 1 }{ 2^{ \delta_{ u v } } } \left( \frac{ \p_i }{ p^0 } - \frac{ \q_i }{ q^0 } \right) \left( \frac{ \p_j }{ p^0 } - \frac{ \q_j }{ q^0 } \right) { e^i }_u { e^j }_v ,
\end{align*}
where the indices $ u $ and $ v $ are not summed. We apply \eqref{p-q} to obtain
\begin{align}
\left| \frac{ \partial h }{ \partial g^{ u v } } \right| \leq C \| e^{ - 1 } \|^2 h ( p^0 q^0 )^2 , \label{est dhdg}
\end{align}
where $ C $ does not depend on the metric. The above estimates will be extensively used to obtain the local existence in Section \ref{sec lip}.

\subsection{Boltzmann equation in an FLRW spacetime}\label{sec flrw}
In this part we consider the single Boltzmann equation in a given spatially flat FLRW spacetime. In this case the metric is given by
\begin{align*}
^4 g_{ F } = - d t^2 + e^{ 2 \gamma t } ( d x^2 + d y^2 + d z^2 ) , 
\end{align*}
where we have $ g_{ a b } = e^{ 2 \gamma t } \delta_{ a b } $ and the constant $ \gamma $ reduces to the Hubble constant, $ H = \gamma $. We further assume that the scattering cross-section is constant, i.e.~$ b = 0 $ in \eqref{potential}. Then, we have from \eqref{B}--\eqref{Gamma}
\begin{align}
\frac{ \partial f }{ \partial t } + \lambda_\eta f = e^{ -3 \gamma t } \iint \frac{ h \sqrt{ s } }{ p^0 q^0 } ( f ( p' ) f ( q' ) - f ( p ) f ( q ) ) w^{ - 1 } ( q ) \, d \omega \, d q . \label{BF}
\end{align}
In the FLRW case, the shear parameter $ F $ vanishes so that $ \lambda_\eta \geq 0 $ by Lemma \ref{lem wdecay}, and we obtain
\begin{align}
\frac{ \partial f }{ \partial t } \leq C e^{ -3 \gamma t } \| f ( t ) \|^2_{ L^\infty } \int_{ \bbr^3 } w^{ - 1 } ( q ) \, d q , \label{BFest}
\end{align}
where we used \eqref{hspq}. We note that $ g^{ a b } q_a q_b = e^{ - 2 \gamma t } | q |^2 $ in the FLRW case, where $ | q |^2 = \sum_{ a = 1 }^3 ( q_a )^2 $. Applying Lemma \ref{lem wlower} with $ \eta = 0 $ we obtain 
\begin{align*}
\int_{ \bbr^3 } w^{ - 1 } ( q ) \, d q & \leq \int_{ \bbr^3 } \exp ( - e^{ \gamma t } \varphi ( e^{ - 2 \gamma t } | q |^2 ) ) \, d q \\ 
& \leq \int_{ e^{ - 2 \gamma t } | q |^2 \leq 1 } \exp \left( - \frac{ e^{ - \gamma t } | q |^2 }{ 1 + \sqrt{ 2 } } \right) \, d q + \int_{ e^{ - 2 \gamma t } | q |^2 \geq 1 } \exp \left( - \frac{ | q | }{ 1 + \sqrt{ 2 } } \right) \, d q . 
\end{align*}
The second term is clearly bounded uniformly in $ t $. For the first term we introduce the change of variables $ q = e^{ \frac{ \gamma t }{ 2 } } u $ to estimate as
\begin{align*}
\int_{ e^{ - 2 \gamma t } | q |^2 \leq 1 } \exp \left( - \frac{ e^{ - \gamma t } | q |^2 }{ 1 + \sqrt{ 2 } } \right) \, d q & \leq e^{ \frac{ 3 \gamma t }{ 2 } } \int_{ \bbr^3 } \exp \left( - \frac{ | u |^2 }{ 1 + \sqrt{ 2 } } \right) \, d u \leq C e^{ \frac{ 3 \gamma t }{ 2 } } . 
\end{align*} 
Applying the above estimates to \eqref{BFest}, we obtain 
\begin{align*}
\frac{ \partial f }{ \partial t } \leq C e^{ - \frac{ 3 \gamma t }{ 2 } } \| f ( t ) \|^2_{ L^\infty }  . 
\end{align*}
Since the factor $ e^{ - \frac{ 3 \gamma t }{ 2 } } $ is integrable on $ [ t_0 , \infty ) $, we are now able to apply the well-known arguments for the global existence of small solutions to the Boltzmann equation. We will obtain
\begin{align*}
\sup_{ t \in [ t_0 , \infty ) } \| f ( t ) \|_{ L^\infty } \leq C \varepsilon ,
\end{align*}
where $ \varepsilon $ denotes the smallness of initial data. This argument will be applied to the Bianchi I case in Sections \ref{sec local} and \ref{sec global}.

\section{Local existence}\label{sec local}
In this section we prove the local existence for the Einstein-Boltzmann system. We first introduce our main lemma of this paper, which shows that our weight function works nicely.

\begin{lemma} \label{lem west}
For any $ N \geq 0 $ and $ c > 0 $, we have
\begin{align}
( \det g )^{ - \frac12 } \int_{ \bbr^3 } \frac{ ( q^0 )^N }{ h^b } w^{ - c } ( q ) \, d q \leq C \frac{ Z_\eta^{ - \frac32 + \frac{ b }{ 2 } } }{ ( p^0 )^{ \frac{ b }{ 2 } } } ,
\end{align}
where $ C $ depends on $ N $ and $ c $.
\end{lemma}
\begin{proof}
We apply Lemma \ref{lem wlower} and consider in an orthonormal frame to have
\begin{align*}
( \det g )^{ - \frac12 } \int_{ \bbr^3 } \frac{ ( q^0 )^N }{ h^b } w^{ - c } ( q ) \, d q \leq \int_{ \bbr^3 } \frac{ ( q^0 )^N }{ h^b } \exp ( - c Z_\eta \varphi ( | \q |^2 ) ) \, d \q ,
\end{align*}
where $ q^0 $ and $ h $ on the right hand side are understood as functions of $ \p $ and $ \q $. Note that
\begin{align*}
( q^0 )^{ \frac{ b }{ 2 } + N } \exp \left( - \frac{ c }{ 2 } Z_\eta \varphi ( | \q |^2 ) \right) \leq ( q^0 )^{ \frac{ b }{ 2 } + N } \exp \left( - \frac{ c }{ 2 } \varphi ( | \q |^2 ) \right) \leq C ,
\end{align*}
since $ Z_\eta \geq 1 $. Hence, we have
\begin{align*}
\int_{ \bbr^3 } \frac{ ( q^0 )^N }{ h^b } \exp ( - c Z_\eta \varphi ( | \q |^2 ) ) \, d \q & \leq \int_{ \bbr^3 } \frac{ ( p^0 )^{ \frac{ b }{ 2 } } ( q^0 )^{ \frac{ b }{ 2 } + N } }{ | \p - \q |^b } \exp ( - c Z_\eta \varphi ( | \q |^2 ) ) \, d \q \\
& \leq C ( p^0 )^{ \frac{ b }{ 2 } } \int_{ \bbr^3 } \frac{ 1 }{ | \p - \q |^b } \exp \left( - \frac{ c }{ 2 } Z_\eta \varphi ( | \q |^2 ) \right) d \q .
\end{align*}
Below, we estimate the last integral, separately for (1) $ | \p | \leq 2 Z_\eta^{ - \frac12 } $ and (2) $ | \p | \geq 2 Z_\eta^{ - \frac12 } $ . \medskip

\noindent (1) For $ | \p | \leq 2 Z_\eta^{ - \frac12 } $, we have
\begin{align*}
& ( p^0 )^{ \frac{ b }{ 2 } } \int_{ \bbr^3 } \frac{ 1 }{ | \p - \q |^b } \exp \left( - \frac{ c }{ 2 } Z_\eta \varphi ( | \q |^2 ) \right) d \q \\
& = ( p^0 )^{ \frac{ b }{ 2 } } \int_{ | \q | \leq 3 Z_\eta^{ - \frac12 } } \frac{ 1 }{ | \p - \q |^b } \exp \left( - \frac{ c }{ 2 } Z_\eta \varphi ( | \q |^2 ) \right) d \q + ( p^0 )^{ \frac{ b }{ 2 } } \int_{ | \q | \geq 3 Z_\eta^{ - \frac12 } } \frac{ 1 }{ | \p - \q |^b } \exp \left( - \frac{ c }{ 2 } Z_\eta \varphi ( | \q |^2 ) \right) d \q .
\end{align*}
The first integral on the right hand side is estimated as
\begin{align*}
& ( p^0 )^{ \frac{ b }{ 2 } } \int_{ | \q | \leq 3 Z_\eta^{ - \frac12 } } \frac{ 1 }{ | \p - \q |^b } \exp \left( - \frac{ c }{ 2 } Z_\eta \varphi ( | \q |^2 ) \right) d \q \\
& \leq C ( p^0 )^{ \frac{ b }{ 2 } } \int_{ | \p - \q | \leq 5 Z_\eta^{ - \frac12 } } \frac{ 1 }{ | \p - \q |^b } d \q \\
& \leq C ( p^0 )^{ \frac{ b }{ 2 } } Z_\eta^{ - \frac32 + \frac{ b }{ 2 } } ,
\end{align*}
while the second integral is estimated as
\begin{align*}
& ( p^0 )^{ \frac{ b }{ 2 } } \int_{ | \q | \geq 3 Z_\eta^{ - \frac12 } } \frac{ 1 }{ | \p - \q |^b } \exp \left( - \frac{ c }{ 2 } Z_\eta \varphi ( | \q |^2 ) \right) d \q \\
& \leq C ( p^0 )^{ \frac{ b }{ 2 } } \int_{ | \q | \geq 3 Z_\eta^{ - \frac12 } } \frac{ 1 }{ ( | \q | - | \p | )^b } \exp \left( - \frac{ c }{ 2 } Z_\eta \varphi ( | \q |^2 ) \right) d \q \\
& \leq C ( p^0 )^{ \frac{ b }{ 2 } } Z_\eta^{ \frac{ b }{ 2 } } \int_{ \bbr^3 } \exp \left( - \frac{ c }{ 2 } Z_\eta \varphi ( | \q |^2 ) \right) d \q \\
& \leq C ( p^0 )^{ \frac{ b }{ 2 } } Z_\eta^{ \frac{ b }{ 2 } } \left( \int_{ | \q | \leq 1 } \exp \left( - \frac{ c }{ 2 } Z_\eta \frac{ | \q |^2 }{ 1 + \sqrt{ 2 } } \right) d \q + \int_{ | \q | \geq 1 } \exp \left( - \frac{ c }{ 2 } Z_\eta \frac{ | \q | }{ 1 + \sqrt{ 2 } } \right) d \q \right) \\
& \leq C ( p^0 )^{ \frac{ b }{ 2 } } Z_\eta^{ \frac{ b }{ 2 } } ( Z_\eta^{ - \frac32 } + Z_\eta^{ - 3 } ) \\
& \leq C ( p^0 )^{ \frac{ b }{ 2 } } Z^{ - \frac32 + \frac{ b }{ 2 } } .
\end{align*}
Since $ | \p | \leq 2 Z_\eta^{ - \frac12 } \leq C $, we have $ p^0 \leq C $ and obtain the desired result. \medskip

\noindent (2) For $ | \p | \geq 2 Z_\eta^{ - \frac12 } $, we have
\begin{align*}
& ( p^0 )^{ \frac{ b }{ 2 } } \int_{ \bbr^3 } \frac{ 1 }{ | \p - \q |^b } \exp \left( - \frac{ c }{ 2 } Z_\eta \varphi ( | \q |^2 ) \right) d \q \\
& = ( p^0 )^{ \frac{ b }{ 2 } } \int_{ | \p - \q | \geq \frac12 | \p | } \frac{ 1 }{ | \p - \q |^b } \exp \left( - \frac{ c }{ 2 } Z_\eta \varphi ( | \q |^2 ) \right) d \q + ( p^0 )^{ \frac{ b }{ 2 } } \int_{ | \p - \q | \leq \frac12 | \p | } \frac{ 1 }{ | \p - \q |^b } \exp \left( - \frac{ c }{ 2 } Z_\eta \varphi ( | \q |^2 ) \right) d \q .
\end{align*}
The first integral on the right hand side is estimated as follows:
\begin{align*}
& ( p^0 )^{ \frac{ b }{ 2 } } \int_{ | \p - \q | \geq \frac12 | \p | } \frac{ 1 }{ | \p - \q |^b } \exp \left( - \frac{ c }{ 2 } Z_\eta \varphi ( | \q |^2 ) \right) d \q \\
& \leq C \frac{ ( p^0 )^{ \frac{ b }{ 2 } } }{ | \p |^b } \int_{ \bbr^3 } \exp \left( - \frac{ c }{ 2 } Z_\eta \varphi ( | \q |^2 ) \right) d \q \\
& \leq C \frac{ ( p^0 )^{ \frac{ b }{ 2 } } }{ | \p |^b } Z_\eta^{ - \frac32 } .
\end{align*}
For the second integral we note that if $ | \p - \q | \leq \frac12 | \p | $, then $| \q | \geq | \p | - | \p - \q | \geq \frac12 | \p | $, so that we have
\begin{align*}
\varphi ( | \q |^2 ) \geq \varphi \left( \frac14 | \p |^2 \right) ,
\end{align*}
since $ \varphi $ is increasing. Now, we have
\begin{align*}
& ( p^0 )^{ \frac{ b }{ 2 } } \int_{ | \p - \q | \leq \frac12 | \p | } \frac{ 1 }{ | \p - \q |^b } \exp \left( - \frac{ c }{ 2 } Z_\eta \varphi ( | \q |^2 ) \right) d \q \\
& \leq ( p^0 )^{ \frac{ b }{ 2 } } \exp \left( - \frac{ c }{ 2 } Z_\eta \varphi \left( \frac14 | \p |^2 \right) \right) \int_{ | \p - \q | \leq \frac12 | \p | } \frac{ 1 }{ | \p - \q |^b } d \q \\
& \leq ( p^0 )^{ \frac{ b }{ 2 } } \exp \left( - \frac{ c }{ 2 } Z_\eta \varphi \left( \frac14 | \p |^2 \right) \right) | \p |^{ 3 - b } .
\end{align*}
If $ \frac14 | \p |^2 \leq \frac12 | \p | $, then 
\begin{align*}
& ( p^0 )^{ \frac{ b }{ 2 } } \exp \left( - \frac{ c }{ 2 } Z_\eta \varphi \left( \frac14 | \p |^2 \right) \right) | \p |^{ 3 - b } \leq ( p^0 )^{ \frac{ b }{ 2 } } \exp \left( - c_1 Z_\eta | \p |^2 \right) | \p |^{ 3 - b } \leq C ( p^0 )^{ \frac{ b }{ 2 } } Z_\eta^{ - \frac32 } | \p |^{ - b } ,
\end{align*}
while if $ \frac14 | \p |^2 \geq \frac12 | \p | $, then
\begin{align*}
& ( p^0 )^{ \frac{ b }{ 2 } } \exp \left( - \frac{ c }{ 2 } Z_\eta \varphi \left( \frac14 | \p |^2 \right) \right) | \p |^{ 3 - b } \leq ( p^0 )^{ \frac{ b }{ 2 } } \exp \left( - c_2 Z_\eta | \p | \right) | \p |^{ 3 - b } \leq C ( p^0 )^{ \frac{ b }{ 2 } } Z_\eta^{ - 3 } | \p |^{ - b } ,
\end{align*}
for some $ c_1 , c_2 > 0 $. Since $ Z_\eta \geq 1 $, we obtain for the second integral
\begin{align*}
& ( p^0 )^{ \frac{ b }{ 2 } } \exp \left( - \frac{ c }{ 2 } Z_\eta \varphi \left( \frac14 | \p |^2 \right) \right) | \p |^{ 3 - b } \leq C ( p^0 )^{ \frac{ b }{ 2 } } Z_\eta^{ - \frac32 } | \p |^{ - b } ,
\end{align*}
which is the same as for the first integral. Now, since $ | \p | \geq 2 Z_\eta^{ - \frac12 } $, we have
\begin{align*}
p^0 = \sqrt{ 1 + | \p |^2 } \leq \sqrt{ \frac14 Z_\eta | \p |^2 + | \p |^2 } \leq C Z_\eta^{ \frac12 } | \p | ,
\end{align*}
which implies that
\begin{align*}
( p^0 )^{ \frac{ b }{ 2 } } Z_\eta^{ - \frac32 } | \p |^{ - b } \leq C \frac{ 1 }{ ( p^0 )^{ \frac{ b }{ 2 } } } Z_\eta^{ - \frac32 + \frac{ b }{ 2 } } .
\end{align*}
This completes the proof of the lemma.
\end{proof}

\subsection{Iteration}\label{sec it}
The local existence will be proved by using an iteration. We will first consider the iterations for the Boltzmann equation and the Einstein equations in Section \ref{sec itB} and Section \ref{sec itE}, respectively. The iteration for the coupled Einstein-Boltzmann system will be given in Section \ref{sec li}, and we will use this to obtain the local existence in Section \ref{sec lp}. The local existence result will be given in Proposition \ref{prop local}.

\subsubsection{Iteration for the Boltzmann equation}\label{sec itB}
Let us first consider the iteration for the single Boltzmann equation. We assume that there exist an interval $ [ t_0 , T ] $ and constants $ c_1 \geq 1 $ and $ 0 \leq \eta < 1 $ such that $ g^{ a b } $ and $ k^{ a b } $ exist on $ [ t_0 , T ] $, are symmetric, and satisfy
\begin{align}
\frac{ 1 }{ c_1 } | p |^2 \leq Z^2 g^{ a b } p_a p_b \leq c_1 | p |^2 , \qquad \| k^{ * * } \| \leq c_1 Z^{ - 2 } , \qquad H \geq \gamma , \qquad F \leq \frac{ \eta^2 }{ 4 } , \label{assum lb}
\end{align}
where $ H $ and $ F $ are defined by \eqref{H} and \eqref{F}, respectively. Under the assumption \eqref{assum lb} we obtain the following lemma.

\begin{lemma}\label{lem assum lb}
Suppose that there exist an interval $ [ t_0 , T ] $ and constants $ c_1 \geq 1 $ and $ 0 \leq \eta < 1 $ such that $ g^{ a b } $ and $ k^{ a b } $ exist on $ [ t_0 , T ] $, are symmetric, and satisfy \eqref{assum lb}. Then, we have
\begin{align*}
\| g^{ - 1 } \| \leq C Z^{ - 2 } , \qquad Z^6 \det g^{ - 1 } \geq \frac{ 1 }{ C } , \qquad \| e^{ - 1 } \| \leq C Z , \qquad \frac{ 1 }{ C } | p |^2 \leq Z^2 k^{ a b } p_a p_b \leq C | p |^2 , 
\end{align*}
on $ [ t_0 , T ] $. 
\end{lemma}
\begin{proof}
Note that the diagonal elements $ g^{ a a } $ are bounded by $ c_1 Z^{ - 2 } $ by the assumption \eqref{assum lb}. Note that $ g^{ a b } $ is positive definite, so that non-diagonal elements are bounded by a suitable sum of diagonal elements, and we obtain the first result:
\begin{align}
\| g^{ - 1 } \| \leq C Z^{ - 2 } . \label{lem assum lb1}
\end{align}
Notice that the determinant $ Z^6 \det g^{ - 1 } $ is the product of the eigenvalues of the matrix $ Z^2 g^{ a b } $, and the assumption \eqref{assum lb} shows that all the eigenvalues are bounded below by $ 1 / c_1 $. Hence, we obtain the second result:
\begin{align}
Z^6 \det g^{ - 1 } \geq \frac{ 1 }{ C } . \label{lem assum lb2}
\end{align}
Since $ g^{ a b } $ is positive definite, we obtain the following inequalities:
\begin{align}
\frac{ 1 }{ g^{ 1 1 } } \leq \frac{ g^{ 2 2 } g^{ 3 3 } }{ \det g^{ - 1 } } , \qquad \frac{ 1 }{ g^{ 1 1 } g^{ 2 2 } - ( g^{ 1 2 } )^2 } \leq \frac{ g^{ 3 3 } }{ \det g^{ - 1 } } . \label{lem assum lb3}
\end{align}
Now, the estimate for $ e^{ - 1 } $ is obtained by applying \eqref{lem assum lb1}--\eqref{lem assum lb3} to the explicit formula of Remark \ref{rem ortho}:
\begin{align}
\| e^{ - 1 } \| \leq C Z . \label{lem assum lb4}
\end{align}
For $ k^{ a b } $, we consider the decomposition \eqref{k decom} to obtain
\begin{align*}
k^{ a b } p_a p_b & = H g^{ a b } p_a p_b + \sigma^{ a b } p_a p_b \\
& \geq H g^{ a b } p_a p_b - ( \sigma^{ a b } \sigma_{ a b } )^{ \frac12 } g^{ c d } p_c p_d \\
& = H g^{ a b } p_a p_b - 2 H F^{ \frac12 } g^{ c d } p_c p_d \\
& \geq H ( 1 - \eta ) g^{ a b } p_a p_b . 
\end{align*}
Since $ H \geq \gamma $ and $ \eta < 1 $ by the assumption, we obtain the positive definiteness of $ k^{ a b } $, i.e.~
\begin{align*}
Z^2 k^{ a b } p_a p_b \geq \frac{ 1 }{ C } | p |^2 .
\end{align*}
Since $ Z^2 k^{ a b } $ is bounded by the assumption \eqref{assum lb}, the last inequality of the lemma is obvious. This completes the proof.
\end{proof}

\begin{remark}\label{rem assum}
Note that in the proof of Lemma \ref{lem assum lb} we used only the first assumption in \eqref{assum lb} to obtain \eqref{lem assum lb1} and \eqref{lem assum lb2}. We then used \eqref{lem assum lb1} and \eqref{lem assum lb2} to obtain \eqref{lem assum lb4}.
\end{remark}

\begin{lemma}\label{lem bracket}
Suppose that there exist an interval $ [ t_0 , T ] $ and constants $ c_1 \geq 1 $ and $ 0 \leq \eta < 1 $ such that $ g^{ a b } $ and $ k^{ a b } $ exist on $ [ t_0 , T ] $, are symmetric, and satisfy \eqref{assum lb}. Then, we have
\begin{align*}
\langle p \rangle \leq C \langle p' \rangle \langle q' \rangle , 
\end{align*}
where $ C $ depends only on $ c_1 $. 
\end{lemma}
\begin{proof}
By Lemma \ref{lem assum lb}, we have the boundedness of $ Z^2 g^{ a b } $ and the positive definiteness of $ Z^2 k^{ a b } $, so that we can follow the arguments of Lemma 7 of \cite{LLN23}. We refer to Lemma 7 of \cite{LLN23} for the proof.
\end{proof}

The following lemma shows that $ p $ derivatives of $ p' $ and $ q' $ are bounded by some quantity which is independent of $ p $. This will be necessary to estimate $ p $ derivatives of $ f $ and will be used in Lemma \ref{lem ls} of Section \ref{sec li}.

\begin{lemma}\label{lem dp'dp}
Suppose that there exist an interval $ [ t_0 , T ] $ and constants $ c_1 \geq 1 $ and $ 0 \leq \eta < 1 $ such that $ g^{ a b } $ and $ k^{ a b } $ exist on $ [ t_0 , T ] $, are symmetric, and satisfy \eqref{assum lb}. Then, the post-collision momenta $ p' $ and $ q' $ satisfy 
\begin{align*}
\left| \frac{ \partial p'_a }{ \partial p_b } \right| + \left| \frac{ \partial q'_a }{ \partial p_b } \right| \leq C ( q^0 )^9 ,
\end{align*}
for any $ a , b = 1 , 2 , 3 $.
\end{lemma}
\begin{proof}
We refer to Lemma 5 of \cite{LLN23} for the proof.
\end{proof}

Now, the iteration for the Boltzmann equation is defined as follows: 
\begin{align}\label{B iter0}
\frac{ \partial_t f_{ n + 1 } }{ \partial t } + \lambda_\eta f_{ n + 1 } = ( \det g )^{ - \frac12 } \iint \frac{ h^{ 1 - b } \sqrt{ s } }{ p^0 q^0 } ( f_n ( p' ) f_n ( q' ) - f_{ n + 1 } ( p ) f_n ( q ) ) w^{ - 1 } ( q ) \, d \omega \, d q .
\end{align}
This is the standard iteration for the Boltzmann equation, which preserves the non-negativity of $ f_n \geq 0 $. Note that $ \lambda_\eta \geq 0 $ by Lemma \ref{lem wdecay} and \eqref{assum lb}, and $ h \sqrt{ s } \leq C p^0 q^0 $ by \eqref{hspq}. We apply Lemma \ref{lem west} to obtain
\begin{align*}
\frac{ \partial ( \langle p \rangle^m f_{ n + 1 } ) }{ \partial t } & \leq C ( \det g )^{ - \frac12 } ( \| f_n \|_{ L^\infty_m }^2 + \| f_{ n + 1 } \|_{ L^\infty_m } \| f_n \|_{ L^\infty_m } ) \int_{ \bbr^3 } \frac{ 1 }{ h^{ b } } w^{ - 1 } ( q ) \, d q \\
& \leq C Z_\eta^{ - \frac32 + \frac{ b }{ 2 } } ( \| f_n \|_{ L^\infty_m }^2 + \| f_{ n + 1 } \|_{ L^\infty_m } \| f_n \|_{ L^\infty_m } ) ,
\end{align*}
where we used Lemma \ref{lem bracket} with $ m \geq 0 $, which shows that
\begin{align}
& \sup_{ t_0 \leq s \leq t } \| f_{ n + 1 } ( s ) \|_{ L^\infty_m } \leq \| f_0 \|_{ L^\infty_m } \nonumber \\
& \quad + C \left( \sup_{ t_0 \leq s \leq t } \| f_n ( s ) \|_{ L_m^{ \infty } }^2 + \sup_{ t_0 \leq s \leq t } \| f_{ n + 1 } ( s ) \|_{ L_m^{ \infty } } \sup_{ t_0 \leq s \leq t } \| f_n ( s ) \|_{ L_m^{ \infty } } \right) \int_{ t_0 }^t Z_\eta^{ - \frac32 + \frac{ b }{ 2 } } ( s ) \, d s . \label{est f_n1}
\end{align}
We now obtain the following lemma, which shows that $ f_n $ are uniformly bounded.

\begin{lemma}\label{lem lb}
Suppose that there exist an interval $ [ t_0 , T ] $ and constants $ c_1 \geq 1 $ and $ 0 \leq \eta < 1 $ such that $ g^{ a b } $ and $ k^{ a b } $ exist on $ [ t_0 , T ] $, are symmetric, and satisfy \eqref{assum lb}. Let $ f_0 $ be an initial datum of the equation \eqref{B iter0} satisfying $ \| f_0 \|_{ L_m^{ \infty } } \leq B $ with $ m \geq 0 $. Then, for any $ C_2 > 1 $, there exists $ t_0 < T_1 \leq T $ such that 
\begin{align*}
\sup_{ t_0 \leq s \leq T_1 } \| f_n ( s ) \|_{ L_m^{ \infty } } \leq C_2 B ,
\end{align*}
for all $ n $. 
\end{lemma}
\begin{proof}
We use \eqref{est f_n1} to obtain
\begin{align*}
& \sup_{ t_0 \leq s \leq t } \| f_{ n + 1 } ( s ) \|_{ L_m^{ \infty } } \leq \| f_0 \|_{ L_m^{ \infty } } \nonumber \\
& \quad + C ( t - t_0 ) \left( \sup_{ t_0 \leq s \leq t } \| f_n ( s ) \|_{ L_m^{ \infty } }^2 + \sup_{ t_0 \leq s \leq t } \| f_{ n + 1 } ( s ) \|_{ L_m^{ \infty } } \sup_{ t_0 \leq s \leq t } \| f_n ( s ) \|_{ L_m^{ \infty } } \right) .
\end{align*}
This implies that there exists $ 0 < T_1 \leq T $ such that if $ \sup_{ t_0 \leq s \leq T_1 } \| f_n ( s ) \|_{ L_m^{ \infty } } \leq C_2 B $, then $ \sup_{ t_0 \leq s \leq T_1 } \| f_{ n + 1 } ( s ) \|_{ L_m^{ \infty } } \leq C_2 B $ as well. This completes the proof. 
\end{proof}

\subsubsection{Iteration for the Einstein equations}\label{sec itE}
Next, we consider the iteration for the Einstein equations. We assume that there exist an interval $ [ t_0 , T ] $ and constants $ c_2 > 0 $ and $ m > 5 $ such that $ f $ exists on $ [ t_0 , T ] $ and satisfy 
\begin{align}
\sup_{ t_0 \leq s \leq T } \| f ( s ) \|_{ L_m^{ \infty } } \leq c_2 . \label{assum le}
\end{align}
The iteration for the Einstein equations is defined as follows. Let $ g_0^{ a b } $ and $ k_0^{ a b } $ be initial data, and define by an abuse of notation the zeroth iteration as $ { g_0 }^{ a b } ( t ) = g_0^{ a b } $ and $ { k_0 }^{ a b } ( t ) = k_0^{ a b } $. Now, for $ n \geq 0 $, we define 
\begin{align}
& \frac{ d { g_{ n + 1 } }^{ a b } }{ d t } = - 2 { k_n }^{ a b } , \label{E iter1} \\ 
& \frac{ d { k_{ n + 1 } }^{ a b } }{ d t } = - 2 { k_n }^a_c { k_n }^{ b c } - k_n { k_n }^{ a b } + { S_{ n } }^{ a b } + \frac12 ( \rho_{ n } - S_{ n } ) { g_n }^{ a b } + \Lambda { g_n }^{ a b } . \label{E iter2} 
\end{align} 
Here, $ { k_n }^a_c = { k_n }^{ a b } { g_n }_{ b c } $ and $ k_n = { k_n }^{ a b } { g_n }_{ a b } $, where $ { g_n }_{ a b } $ denotes the metric $ g_n $ with covariant indices $ a $ and $ b $, while $ { g_n }^{ a b } $ denotes the inverse $ g_n^{ - 1 } $ with contravariant indices $ a $ and $ b $. The matter terms are defined by using the metric $ { g_n }^{ a b } $, i.e.
\begin{align*}
{ S_n }_{ a b } & = ( \det g_n^{ - 1 } )^{ \frac12 } \int_{ \bbr^3 } f ( t , p ) \frac{ p_a p_b }{ { p_n }^0 } w_n^{ - 1 } ( p ) \, d p , \\
\rho_{ n } & = ( \det g_n^{ - 1 } )^{ \frac12 } \int_{ \bbr^3 } f ( t , p ) { p_n }^0 w_n^{ - 1 } ( p ) \, d p , 
\end{align*}
and $ { S_n }^{ a b } = { g_n }^{ a c } { g_n }^{ b d } { S_n }_{ c d } $ and $ S_n = { g_n }^{ a b } { S_n }_{ a b } $, where $ { p_n }^0 $ and $ w_n $ are given by
\begin{align}
{ p_n }^0 = \sqrt{ 1 + { g_n }^{ a b } p_a p_b } , \qquad w_n = e^{ Z_\eta ( { p_n }^0 - 1 ) } , \label{p_n0}
\end{align}
respectively. The iteration satisfies the following estimates.

\begin{lemma}\label{lem le}
Suppose that there exist an interval $ [ t_0 , T ] $ and constants $ c_2 > 0 $ and $ m > 5 $ such that $ f $ exists on $ [ t_0 , T ] $ and satisfies \eqref{assum le}. Let $ g_0^{ a b } $ and $ k_0^{ a b } $ be initial data of the equations \eqref{E iter1}--\eqref{E iter2} such that they are symmetric and satisfy
\begin{align*}
\frac{ 1 }{ A } | p |^2 \leq e^{ 2 \gamma t_0 } g_0^{ a b } p_a p_b \leq A | p |^2 , \qquad \| k_0^{ * * } \| \leq A e^{ - 2 \gamma t_0 } , \qquad H_0 > \gamma , \qquad F_0 < \frac14 ,
\end{align*}
where $ H_0 $ and $ F_0 $ denote the initial values of the Hubble variable \eqref{H} and the shear variable \eqref{F}, respectively. Then, for any $ C_1 > 1 $, there exist $ t_0 < T_2 \leq T $ and $ 0 \leq \eta < 1 $ such that
\begin{align}
\frac{ 1 }{ C_1 A } | p |^2 \leq Z^2 { g_n }^{ a b } p_a p_b \leq C_1 A | p |^2 , \qquad \| k_n^{ * * } \| \leq C_1 A Z^{ - 2 } , \qquad H_n \geq \gamma , \qquad F_n \leq \frac{ \eta^2 }{ 4 } , \label{lem le1}
\end{align}
for all $ n $ and $ t_0 \leq t \leq T_2 $.  
\end{lemma}
\begin{proof}
Note that the zeroth iteration functions are defined by initial data, which are constant matrices, so that they satisfy \eqref{lem le1} on an interval $ [ t_0 , T_2 ] $. Suppose that \eqref{lem le1} holds for all $ 0 \leq i \leq n $ for some $ n \geq 0 $ on $ [ t_0 , T_2 ] $, and let us consider the estimates of $ { g_{ n + 1 } }^{ a b } $, $ { k_{ n + 1 } }^{ a b } $, $ H_{ n + 1 } $ and $ F_{ n + 1 } $.

First, we use \eqref{E iter1} to obtain
\begin{align*}
{ g_{ n + 1 } }^{ a b } = g_0^{ a b } - 2 \int_{ t_0 }^t { k_n }^{ a b } \, d s , 
\end{align*}
so that we have
\begin{align*}
| { g_{ n + 1 } }^{ a b } p_a p_b - g_0^{ a b } p_a p_b | \leq C A ( t - t_0 ) | p |^2 ,
\end{align*}
where we used \eqref{lem le1}. Since $ Z = e^{ \gamma t } $ is continuous and increasing, we can find  an interval, on which we have
\begin{align*}
| Z^2 { g_{ n + 1 } }^{ a b } p_a p_b - e^{ 2 \gamma t_0 } g_0^{ a b } p_a p_b | \leq C A ( t - t_0 ) | p |^2 ,
\end{align*}
and
\begin{align}
\frac{ 1 }{ C_1 A } | p |^2 \leq \frac{ 1 }{ A } | p |^2 - C A ( t - t_0 ) | p |^2 \leq Z^2 { g_{ n + 1 } }^{ a b } p_a p_b \leq A | p |^2 + C A ( t - t_0 ) | p |^2 \leq C_1 A | p |^2 . \label{lem le p1}
\end{align}
For the boundedness of $ { k_{ n + 1 } }^{ a b } $, we use \eqref{E iter2} to obtain
\begin{align*}
| { k_{ n + 1 } }^{ a b } - k_0^{ a b } | \leq C \int_{ t_0 }^t ( \| k_n^{ * * } \|^2 \| g_n \| + | { S_n }^{ a b } | + | \rho_n | \| g_n^{ - 1 } \| + | S_n | \| g_n^{ - 1 } \| + \| g_n^{ - 1 } \| ) \, d s . 
\end{align*}
Here, $ g_n $ is the inverse of $ g_n^{ - 1 } $, so it is bounded by $ \| g_n^{ - 1 } \|^2 $ multiplied by $ ( \det g_n^{ - 1 } )^{ - 1 } $. We can apply Lemma \ref{lem assum lb} to conclude that $ \| k_n^{ * * } \|^2 \| g_n \| $ and $ \| g_n^{ - 1 } \| $ are bounded by $ C_A Z^{ - 2 } $ for some constant $ C_A $ which depends on $ A $. For the matter terms we use the assumption \eqref{assum le} to have
\begin{align*}
| { S_n }_{ a b } | & \leq C ( \det g_n^{ - 1 } )^{ \frac12 } \int_{ \bbr^3 } \langle p \rangle^{ - m + 2 } d p  \leq C ( \det g_n^{ - 1 } )^{ \frac12 } \leq C_A Z^{ - 3 } , 
\end{align*}
since $ m > 5 $. Then, we obtain
\begin{align*}
| { S_n }^{ a b } | & \leq C ( \det g_n^{ - 1 } )^{ \frac12 } \| g_n^{ - 1 } \|^2 \int_{ \bbr^3 } \langle p \rangle^{ - m + 2 } d p  \leq C ( \det g_n^{ - 1 } )^{ \frac12 } \| g_n^{ - 1 } \|^2 \leq C_A Z^{ - 7 } , \\
\rho_n & \leq C ( \det g_n^{ - 1 } )^{ \frac12 } \int_{ \bbr^3 } \langle p \rangle^{ - m } d p  \leq C ( \det g_n^{ - 1 } )^{ \frac12 } \leq C_A Z^{ - 3 } , \\
{ S_n } & \leq C ( \det g_n^{ - 1 } )^{ \frac12 } \| g_n^{ - 1 } \| \int_{ \bbr^3 } \langle p \rangle^{ - m + 2 } d p  \leq C ( \det g_n^{ - 1 } )^{ \frac12 } \| g_n^{ - 1 } \| \leq C_A Z^{ - 5 } ,
\end{align*}
which are all bounded on $ [ t_0 , T_2 ] $. Hence, we obtain 
\begin{align*}
| { k_{ n + 1 } }^{ a b } - k_0^{ a b } | \leq C_A ( t - t_0 ) .
\end{align*}
By the same argument as in \eqref{lem le p1} we obtain an interval on which we have
\begin{align}
| Z^2 { k_{ n + 1 } }^{ a b } | \leq | e^{ 2 \gamma t_0 } k_0^{ a b } | + | Z^2 { k_{ n + 1 } }^{ a b } - e^{ 2 \gamma t_0 } k_0^{ a b } | \leq A + C_A ( t - t_0 ) \leq C_1 A . \label{lem le p2}
\end{align}
Next, we recall that $ H_{ n + 1 } = \frac13 k_{ n + 1 } = \frac13 { g_{ n + 1 } }_{ a b } { k_{ n + 1 } }^{ a b } $. We need to consider the time derivative of $ { g_{ n + 1 } }_ { a b } $, which is given by
\begin{align*}
\frac{ d { g_{ n + 1 } }_{ a b } }{ d t } = - { g_{ n + 1 } }_{ a c } { g_{ n + 1 } }_{ b d } \frac{ d { g_{ n + 1 } }^{ c d } }{ d t } ,
\end{align*}
and the time derivative of $ { k_{ n + 1 } }^{ a b } $, which is given by the right hand side of \eqref{E iter2}. We observe that they are all bounded by Lemma \ref{lem assum lb} and Remark \ref{rem assum}, so that we obtain
\begin{align*}
\left| \frac{ d H_{ n + 1 } }{ d t } \right| \leq C_A .
\end{align*}
This shows that there exists an interval on which we have
\begin{align}
H_{ n + 1 } \geq H_0 - | H_{ n + 1 } - H_0 | \geq H_0 - C_A ( t - t_0 ) \geq \gamma , \label{lem le p3}
\end{align}
since $ H_0 > \gamma $. Finally, we recall that
\begin{align*}
F_{ n + 1 } = \frac{ { \sigma_{ n + 1 } }_{ a b } { \sigma_{ n + 1 } }^{ a b } }{ 4 H_{ n + 1 }^2 } , 
\end{align*}
where $ { \sigma_{ n + 1 } }^{ a b } = { k_{ n + 1 } }^{ a b } - H_{ n + 1 } { g_{ n + 1 } }^{ a b } $, and the indices are lowered with $ { g_{ n + 1 } }_{ a b } $. We can observe that $ { \sigma_{ n + 1 } }_{ a b } $, $ { \sigma_{ n + 1 } }^{ a b } $, $ H_{ n + 1 }^{ - 1 } $, $ d { \sigma_{ n + 1 } }_{ a b } / d t $, $ d { \sigma_{ n + 1 } }^{ a b } / d t $ and $ d H_{ n + 1 }^{ - 1 } / d t $ are all bounded, which shows that $ d F_{ n + 1 } / d t $ is also bounded. Since $ F_0 < 1 / 4 $, we can conclude that there exists an interval on which $ F_{ n + 1 } $ remains strictly less than $ 1 / 4 $. To be precise, we choose $ 0 \leq \eta < 1 $ such that
\begin{align}
F_0 < \frac{ \eta^2 }{ 4 } < \frac14 , \label{eta}
\end{align}
and obtain
\begin{align}
F_{ n + 1 } \leq \frac{ \eta^2 }{ 4 } . \label{lem le p4}
\end{align}
Now, we redefine $ [ t_0 , T_2 ] $ as an interval on which all the estimates \eqref{lem le p1}, \eqref{lem le p2}, \eqref{lem le p3} and \eqref{lem le p4} hold. We notice that the interval $ [ t_0 , T_2 ] $ can be chosen independently of $ n $. Therefore, we can conclude by induction that \eqref{lem le1} holds for all $ n $ and $ t_0 \leq t \leq T_2 $.
\end{proof}

\subsection{Local existence for the Einstein-Boltzmann system}\label{sec lip}

\subsubsection{Iteration for the coupled system}\label{sec li}
We introduce the iteration for the Einstein-Boltzmann system. Suppose that initial data $ g_0^{ a b } $, $ k_0^{ a b } $ and $ f_0 $ are given, and define by an abuse of notation the zeroth iteration as
\begin{align*}
{ g_0 }^{ a b } ( t ) = { g_0 }^{ a b } , \qquad { k_0 }^{ a b } ( t ) = { k_0 }^{ a b } , \qquad f_0 ( t , p ) = f_0 ( p ) . 
\end{align*}
Now, suppose that $ { g_j }^{ a b } $, $ { k_j }^{ a b } $ and $ f_j $ are given. We first define $ f_{ j + 1 } $ as the solution of the following equation: 
\begin{align}\label{EB iter0}
\frac{ \partial_t f_{ j + 1 } }{ \partial t } + \lambda_{ \eta , j } f_{ j + 1 } = ( \det g_j^{ - 1 } )^{ \frac12 } \iint \frac{ h_j^{ 1 - b } \sqrt{ s_j } }{ { p_j }^0 { q_j }^0 } ( f_j ( p_j' ) f_j ( q_j' ) - f_{ j + 1 } ( p ) f_j ( q ) ) w_j^{ - 1 } ( q ) \, d \omega \, d q , 
\end{align}
where $ \eta $ will be determined later, and $ { p_j }^0 $ and $ { q_j }^0 $ are defined by \eqref{p^0} using the metric $ { g_j }^{ a b } $:
\begin{align*}
{ p_j }^0 = \sqrt{ 1 + { g_j }^{ a b } p_a p_b } , \qquad { q_j }^0 = \sqrt{ 1 + { g_j }^{ a b } q_a q_b } ,
\end{align*}
which can be understood as $ { p_j }^0 = p^0 ( { g_j }^{ a b } ) $ with $ p^0 $ as a function of $ g^{ a b } $.  The quantities $ \lambda_{ \eta , j } $, $ h_j $, $ s_j $, $ p'_j $, $ q'_j $ and $ w_j $ are understood similarly, for instance
\begin{align*}
& \lambda_{ \eta , j } = - \left( ( 1 - \eta ) \gamma ( { p_j }^0 - 1 ) - \frac{ { k_j }^{ a b } p_a p_b }{ { p_j }^0 } \right) Z_\eta , \\
& h_j = \sqrt{ - ( { p_j }^0 - { q_j }^0 )^2 + { g_j }^{ a b } ( p_a - q_a ) ( p_b - q_b ) } , 
\end{align*}
and the post-collision momenta $ p_j' $ and $ q_j' $ are parametrized by \eqref{p'}--\eqref{q'} using the metric $ { g_j }^{ a b } $ with the corresponding $ { { e_j }^i }_a $, which are again defined by Remark \ref{rem ortho} using the metric $ { g_j }^{ a b } $. The matter terms are now defined by 
\begin{align}
{ S_{ j + 1 } }^{ a b } & = ( \det g_j^{ - 1 } )^{ \frac12 } \int_{ \bbr^3 } f_{ j + 1 } \frac{ { p_j }^a { p_j }^b }{ { p_j }^0 } w_j^{ - 1 } ( p ) \, d p , \label{EB iterm1} \\
\rho_{ j + 1 } & = ( \det g_j^{ - 1 } )^{ \frac12 } \int_{ \bbr^3 } f_{ j + 1 } { p_j }^0 w_j^{ - 1 } ( p ) \, d p , \label{EB iterm2} \\
S_{ j + 1 } & = ( \det g_j^{ - 1 } )^{ \frac12 } \int_{ \bbr^3 } f_{ j + 1 } \frac{ { g_j }^{ a b } p_a p_b }{ { p_j }^0 } w_j^{ - 1 } ( p ) \, d p . \label{EB iterm3}
\end{align}
Finally, $ { g_{ j + 1 } }^{ a b } $ and $ { k_{ j + 1 } }^{ a b } $ are defined by 
\begin{align}
& \frac{ d { g_{ j + 1 } }^{ a b } }{ d t } = - 2 { k_j }^{ a b } , \label{EB iter1} \\ 
& \frac{ d { k_{ j + 1 } }^{ a b } }{ d t } = - 2 { k_j }^a_c { k_j }^{ b c } - k_j { k_j }^{ a b } + { S_{ j + 1 } }^{ a b } + \frac12 ( \rho_{ j + 1 } - S_{ j + 1 } ) { g_j }^{ a b } + \Lambda { g_j }^{ a b } . \label{EB iter2} 
\end{align} 
Now, we combine Lemmas \ref{lem lb} and \ref{lem le} to obtain the uniform estimates of $ f_j $, $ { g_j }^{ a b } $ and $ { k_j }^{ a b } $.

\begin{lemma}\label{lem ls}
Suppose that $ g_0^{ a b } $ and $ k_0^{ a b } $ are symmetric matrices satisfying
\begin{align}
\frac{ 1 }{ A } | p |^2 \leq e^{ 2 \gamma t_0 } g_0^{ a b } p_a p_b \leq A | p |^2 , \qquad \| k_0^{ * * } \| \leq A e^{ - 2 \gamma t_0 } , \qquad H_0 > \gamma , \qquad F_0 < \frac14 , \label{lem lsi1}
\end{align}
for some $ A > 0 $, and that there exist $ B > 0 $ and $ m > 5 $ such that
\begin{align}
\| f_0 \|_{ L_m^{ \infty } } \leq B . \label{lem lsi2}
\end{align}
Let $ g_0^{ a b } $, $ k_0^{ a b } $ and $ f_0 \geq 0 $ be initial data of the equations \eqref{EB iter0}--\eqref{EB iter2} with $ \eta $ chosen by \eqref{eta}. Then, for any $ C_1 , C_2 > 1 $, there exists an interval $ [ t_0 , T ] $ on which we have
\begin{align}
\frac{ 1 }{ C_1 A } | p |^2 \leq Z^2 { g_j }^{ a b } p_a p_b \leq C_1 A | p |^2 , \qquad \| k_j^{ * * } \| \leq C_1 A Z^{ - 2 } , \qquad H_j \geq \gamma , \qquad F_j \leq \frac{ \eta^2 }{ 4 } , \label{lem ls1}
\end{align}
for all $ j $ and $ t_0 \leq t \leq T $, and
\begin{align}
\sup_{ t_0 \leq s \leq T } \| f_j ( s ) \|_{ L^{ \infty }_m } \leq C_2 B , \label{lem ls2}
\end{align}
for all $ j $. Moreover, if $ f_0 \in W^{ 1 , \infty }_m $, then $ \sup_{ t_0 \leq s \leq T } \| f_j ( s ) \|_{ W^{ 1 , \infty }_m } $ is also bounded uniformly in $ j $.
\end{lemma}
\begin{proof}
We recall that the zeroth iteration functions are defined by initial data. Hence, $ { g_0 }^{ a b } $, $ { k_0 }^{ a b } $ and $ f_0 $ satisfy \eqref{lem lsi1} and \eqref{lem lsi2}, and there exists an interval on which $ { g_0 }^{ a b } $, $ { k_0 }^{ a b } $ and $ f_0 $ satisfy \eqref{lem ls1} and \eqref{lem ls2}. Suppose that there exists an interval on which $ { g_j }^{ a b } $ and $ { k_j }^{ a b } $ satisfy \eqref{lem ls1} for some $ j \geq 0 $. We apply \eqref{lem lsi2} and \eqref{lem ls1} to Lemma \ref{lem lb} with $ c_1 = C_1 A $ to obtain $ [ t_0 , T_1 ] $ such that
\begin{align}
\sup_{ t_0 \leq s \leq T_1 } \| f_{ j + 1 } ( s ) \|_{ L^{ \infty }_m } \leq C_2 B . \label{lem ls3}
\end{align}
On the other hand, we apply \eqref{lem lsi1} and \eqref{lem ls3} to Lemma \ref{lem le} with $ c_2 = C_2 B $ to obtain $ [ t_0 , T ] $ with $ T = \min ( T_1 , T_2 ) $ on which we have
\begin{align}
\frac{ 1 }{ C_1 A } | p |^2 \leq Z^2 { g_{ j + 1 } }^{ a b } p_a p_b \leq C_1 A | p |^2 , \qquad \| k_{ j + 1 }^{ * * } \| \leq C_1 A Z^{ - 2 } , \qquad H_{ j + 1 } \geq \gamma , \qquad F_{ j + 1 } \leq \frac{ \eta^2 }{ 4 } . \label{lem ls4}
\end{align}
Now, we can conclude by induction that $ { g_j }^{ a b } $, $ { k_j }^{ a b } $ and $ f_j $ for all $ j \geq 0 $ satisfy \eqref{lem ls1} and \eqref{lem ls2} on the interval $ [ t_0 , T ] $.

Now, we further assume that $ f_0 \in W^{ 1 , \infty }_m $. We need to consider $ p $ derivatives of $ f_j $. We obtain from \eqref{EB iter0}
\begin{align}
& \frac{ \partial }{ \partial t } \left( \frac{ \partial f_{ j + 1 } }{ \partial p_a } \right) + \frac{ \partial \lambda_{ \eta , j } }{ \partial p_a } f_{ j + 1 } + \lambda_{ \eta , j } \frac{ \partial f_{ j + 1 } }{ \partial p_a } \nonumber \\
& = ( \det g_j^{ - 1 } )^{ \frac12 } \iint \frac{ \partial }{ \partial p_a } \left( \frac{ h_j^{ 1 - b } \sqrt{ s_j } }{ { p_j }^0 { q_j }^0 } \right) ( f_j ( p_j' ) f_j ( q_j' ) - f_{ j + 1 } ( p ) f_j ( q ) ) w_j^{ - 1 } ( q ) \, d \omega \, d q \nonumber \\
& \quad + ( \det g_j^{ - 1 } )^{ \frac12 } \iint \frac{ h_j^{ 1 - b } \sqrt{ s_j } }{ { p_j }^0 { q_j }^0 } \frac{ \partial ( f_j ( p'_j ) f_j ( q_j' ) ) }{ \partial p_a } w_j^{ - 1 } ( q ) \, d \omega \, d q \nonumber \\
& \quad - ( \det g_j^{ - 1 } )^{ \frac12 } \iint \frac{ h_j^{ 1 - b } \sqrt{ s_j } }{ { p_j }^0 { q_j }^0 } \frac{ \partial f_{ j + 1 } ( p ) }{ \partial p_a } f_j ( q ) w_j^{ - 1 } ( q ) \, d \omega \, d q . \label{est df_ndp0}
\end{align}
Multiplying the above by $ \mbox{sgn} ( \partial f_{ j + 1 } / \partial p_a ) $, which can be written as
\begin{align*}
\mbox{sgn} \left( \frac{ \partial f_{ j + 1 } }{ \partial p_a } \right) = \frac{ \frac{ \partial f_{ j + 1 } }{ \partial p_a } }{ \left| \frac{ \partial f_{ j + 1 } }{ \partial p_a } \right| } ,
\end{align*}
we obtain from the left hand side of \eqref{est df_ndp0}
\begin{align}
\mbox{sgn} \left( \frac{ \partial f_{ j + 1 } }{ \partial p_a } \right) (\mbox{LHS}) & = \mbox{sgn} \left( \frac{ \partial f_{ j + 1 } }{ \partial p_a } \right) \left( \frac{ \partial }{ \partial t } \left( \frac{ \partial f_{ j + 1 } }{ \partial p_a } \right) + \frac{ \partial \lambda_{ \eta , j } }{ \partial p_a } f_{ j + 1 } + \lambda_{ \eta , j } \frac{ \partial f_{ j + 1 } }{ \partial p_a } \right) \nonumber \\
& = \frac{ \partial }{ \partial t } \left| \frac{ \partial f_{ j + 1 } }{ \partial p_a } \right| + \mbox{sgn} \left( \frac{ \partial f_{ j + 1 } }{ \partial p_a } \right) \frac{ \partial \lambda_{ \eta , j } }{ \partial p_a } f_{ j + 1 } + \lambda_{ \eta, j } \left| \frac{ \partial f_{ j + 1 } }{ \partial p_a } \right| \nonumber \\
& \geq \frac{ \partial }{ \partial t } \left| \frac{ \partial f_{ j + 1 } }{ \partial p_a } \right| - \left| \frac{ \partial \lambda_{ \eta , j } }{ \partial p_a } \right| f_{ j + 1 } + \lambda_{ \eta, j } \left| \frac{ \partial f_{ j + 1 } }{ \partial p_a } \right| , \label{est df_ndp01}
\end{align}
where we used $ \partial_t | \psi | = ( \psi / | \psi | ) \partial_t \psi $ for any $ \partial_t \psi \in L^\infty $ in the sense of distribution. For more details about derivatives of the absolute value we refer to Section 6.17 of \cite{LL}. In a similar way we obtain from the right hand side of \eqref{est df_ndp0} 
\begin{align}
& \mbox{sgn} \left( \frac{ \partial f_{ j + 1 } }{ \partial p_a } \right) (\mbox{RHS}) \nonumber \\
& = \mbox{sgn} \left( \frac{ \partial f_{ j + 1 } }{ \partial p_a } \right) \left( ( \det g_j^{ - 1 } )^{ \frac12 } \iint \frac{ \partial }{ \partial p_a } \left( \frac{ h_j^{ 1 - b } \sqrt{ s_j } }{ { p_j }^0 { q_j }^0 } \right) ( f_j ( p_j' ) f_j ( q_j' ) - f_{ j + 1 } ( p ) f_j ( q ) ) w_j^{ - 1 } ( q ) \, d \omega \, d q \right) \nonumber \\
& \quad + \mbox{sgn} \left( \frac{ \partial f_{ j + 1 } }{ \partial p_a } \right) \left( ( \det g_j^{ - 1 } )^{ \frac12 } \iint \frac{ h_j^{ 1 - b } \sqrt{ s_j } }{ { p_j }^0 { q_j }^0 } \frac{ \partial ( f_j ( p'_j ) f_j ( q_j' ) ) }{ \partial p_a } w_j^{ - 1 } ( q ) \, d \omega \, d q \right) \nonumber \\
& \quad - \mbox{sgn} \left( \frac{ \partial f_{ j + 1 } }{ \partial p_a } \right) \left( ( \det g_j^{ - 1 } )^{ \frac12 } \iint \frac{ h_j^{ 1 - b } \sqrt{ s_j } }{ { p_j }^0 { q_j }^0 } \frac{ \partial f_{ j + 1 } ( p ) }{ \partial p_a } f_j ( q ) w_j^{ - 1 } ( q ) \, d \omega \, d q \right) \nonumber \allowdisplaybreaks \\
& = \mbox{sgn} \left( \frac{ \partial f_{ j + 1 } }{ \partial p_a } \right) \left( ( \det g_j^{ - 1 } )^{ \frac12 } \iint \frac{ \partial }{ \partial p_a } \left( \frac{ h_j^{ 1 - b } \sqrt{ s_j } }{ { p_j }^0 { q_j }^0 } \right) ( f_j ( p_j' ) f_j ( q_j' ) - f_{ j + 1 } ( p ) f_j ( q ) ) w_j^{ - 1 } ( q ) \, d \omega \, d q \right) \nonumber \\
& \quad + \mbox{sgn} \left( \frac{ \partial f_{ j + 1 } }{ \partial p_a } \right) \left( ( \det g_j^{ - 1 } )^{ \frac12 } \iint \frac{ h_j^{ 1 - b } \sqrt{ s_j } }{ { p_j }^0 { q_j }^0 } \frac{ \partial ( f_j ( p'_j ) f_j ( q_j' ) ) }{ \partial p_a } w_j^{ - 1 } ( q ) \, d \omega \, d q \right) \nonumber \\
& \quad - ( \det g_j^{ - 1 } )^{ \frac12 } \iint \frac{ h_j^{ 1 - b } \sqrt{ s_j } }{ { p_j }^0 { q_j }^0 } \left| \frac{ \partial f_{ j + 1 } ( p ) }{ \partial p_a } \right| f_j ( q ) w_j^{ - 1 } ( q ) \, d \omega \, d q \allowdisplaybreaks \nonumber \\
& \leq ( \det g_j^{ - 1 } )^{ \frac12 } \iint \left| \frac{ \partial }{ \partial p_a } \left( \frac{ h_j^{ 1 - b } \sqrt{ s_j } }{ { p_j }^0 { q_j }^0 } \right) \right| ( f_j ( p'_j ) f_j ( q_j' ) + f_{ j + 1 } ( p ) f_j ( q ) ) w_j^{ - 1 } ( q ) \, d \omega \, d q \nonumber \\
& \quad + ( \det g_j^{ - 1 } )^{ \frac12 } \iint \frac{ h_j^{ 1 - b } \sqrt{ s_j } }{ { p_j }^0 { q_j }^0 } \left| \frac{ \partial ( f_j ( p_j' ) f_j ( q_j' ) ) }{ \partial p_a } \right| w_j^{ - 1 } ( q ) \, d \omega \, d q \nonumber \\
& \quad - ( \det g_j^{ - 1 } )^{ \frac12 } \iint \frac{ h_j^{ 1 - b } \sqrt{ s_j } }{ { p_j }^0 { q_j }^0 } \left| \frac{ \partial f_{ j + 1 } ( p ) }{ \partial p_a } \right| f_j ( q ) w_j^{ - 1 } ( q ) \, d \omega \, d q . \label{est df_ndp02}
\end{align}
We combine the estimates \eqref{est df_ndp01} and \eqref{est df_ndp02} to obtain
\begin{align}
& \frac{ \partial }{ \partial t } \left| \frac{ \partial f_{ j + 1 } }{ \partial p_a } \right| + \lambda_{ \eta, j } \left| \frac{ \partial f_{ j + 1 } }{ \partial p_a } \right| \leq \left| \frac{ \partial \lambda_{ \eta , j } }{ \partial p_a } \right| f_{ j + 1 } \nonumber \\
& \quad + ( \det g_j^{ - 1 } )^{ \frac12 } \iint \left| \frac{ \partial }{ \partial p_a } \left( \frac{ h_j^{ 1 - b } \sqrt{ s_j } }{ { p_j }^0 { q_j }^0 } \right) \right| ( f_j ( p'_j ) f_j ( q_j' ) + f_{ j + 1 } ( p ) f_j ( q ) ) w_j^{ - 1 } ( q ) \, d \omega \, d q \nonumber \\
& \quad + ( \det g_j^{ - 1 } )^{ \frac12 } \iint \frac{ h_j^{ 1 - b } \sqrt{ s_j } }{ { p_j }^0 { q_j }^0 } \left| \frac{ \partial ( f_j ( p_j' ) f_j ( q_j' ) ) }{ \partial p_a } \right| w_j^{ - 1 } ( q ) \, d \omega \, d q \nonumber \\
& \quad - ( \det g_j^{ - 1 } )^{ \frac12 } \iint \frac{ h_j^{ 1 - b } \sqrt{ s_j } }{ { p_j }^0 { q_j }^0 } \left| \frac{ \partial f_{ j + 1 } ( p ) }{ \partial p_a } \right| f_j ( q ) w_j^{ - 1 } ( q ) \, d \omega \, d q . \label{est df_ndp1}
\end{align}
The first term on the right hand side of \eqref{est df_ndp1} is estimated by Lemmas \ref{lem wdecay} and \ref{lem assum lb}:
\begin{align*}
\left| \frac{ \partial \lambda_{ \eta , j } }{ \partial p_a } \right| \leq C Z^{ - 1 } Z_\eta \leq C .
\end{align*}
The second term on the right hand side of \eqref{est df_ndp1} is estimated by Lemma \ref{lem dhdp} with \eqref{dpdp} and \eqref{dpdp est}:
\begin{align*}
\left| \frac{ \partial }{ \partial p_a } \left( \frac{ h_j^{ 1 - b } \sqrt{ s_j } }{ { p_j }^0 { q_j }^0 } \right) \right| & \leq C \left( \frac{ h_j^{ - b } \sqrt{ s_j } }{ { p_j }^0 { q_j }^0 } \left| \frac{ \partial h_j }{ \partial p_a } \right| + \frac{ h_j^{ 1 - b } }{ { p_j }^0 { q_j }^0 } \left| \frac{ \partial \sqrt{ s_j } }{ \partial p_a } \right| + \frac{ h_j^{ 1 - b } \sqrt{ s_j } }{ ( { p_j }^0 )^2 { q_j }^0 } \left| \frac{ \partial { p_j }^0 }{ \partial p_a } \right| \right) \\
& \leq C \left( { q_j }^0 + \frac{ 1 }{ { p_j }^0 } \right) h_j^{ - b } \| g_j^{ - 1 } \| \| e_j^{ - 1 } \| \\
& \leq C Z^{ - 1 } \frac{ { q_j }^0 }{ h_j^{ b } }  .
\end{align*}
The third term on the right side of \eqref{est df_ndp1} can be estimated by the chain rule and Lemma \ref{lem dp'dp}. Now, we simply ignore the second term on the left hand side and the fourth term on the right hand side to obtain the following estimate:
\begin{align}
& \frac{ \partial }{ \partial t } \left( \langle p \rangle^m \left| \frac{ \partial f_{ j + 1 } }{ \partial p_a } \right| \right) \leq C \| f_{ j + 1 } \|_{ L^\infty_m } \nonumber \\
& \quad + C Z^{ - 1 } ( \det g_j^{ - 1 } )^{ \frac12 } ( \| f_j \|_{ L^\infty_m }^2 + \| f_{ j + 1 } \|_{ L^\infty_m } \| f_j \|_{ L^\infty_m } ) \int \frac{ { q_j }^0 }{ h_j^{ b } } w_j^{ - 1 } ( q ) \, d q \nonumber \\
& \quad + C ( \det g_j^{ - 1 } )^{ \frac12 } \left\| \frac{ \partial f_j }{ \partial p } \right\|_{ { L^\infty_m } } \| f_j \|_{ { L^\infty_m } } \int \frac{ ( { q_j }^0 )^9 }{ h_j^b } w_j^{ - 1 } ( q ) \, d q . \nonumber
\end{align}
Applying Lemma \ref{lem west}, and using \eqref{lem ls2}, we obtain
\begin{align}
\frac{ d }{ d t } \left\| \frac{ \partial f_{ j + 1 } }{ \partial p } \right\|_{ L^\infty_m } \leq C + C \left\| \frac{ \partial f_j }{ \partial p } \right\|_{ { L^\infty_m } } . \label{est df_ndp2}
\end{align}
In the rest of the proof we fix the constant $ C $ above. Let us write $ y_j = \| \partial f_j / \partial p \|_{ L^\infty_m } $. Then, we have $ y_0' = 0 $ and $ y_{ j + 1 }' \leq C + C y_j $ for $ j \geq 0 $, and let $ y $ denote 
\begin{align*}
y ( t ) = y_0 e^{ C ( t - t_0 ) } + e^{ C ( t - t_0 ) } - 1 .
\end{align*}
We notice that $ y_0 \leq y $ for all $ t_0 \leq t \leq T $. Suppose that $ y_j \leq y $ for some $ j \geq 0 $. We integrate \eqref{est df_ndp2} to obtain
\begin{align*}
y_{ j + 1 } & \leq y_0 + C ( t - t_0 ) + C \int_{ t_0 }^t y_j \, d s \\
& \leq y_0 + C ( t - t_0 ) + C \int_{ t_0 }^t y \, d s \\
& = y .
\end{align*}
Hence, we conclude by induction that $ y_j $ is bounded by $ y $ for all $ j \geq 0 $ and $ t_0 \leq t \leq T $, which implies that $ y_j $ is bounded on $ [ t_0 , T ] $, and uniformly in $ j $. This completes the proof of the lemma.
\end{proof}

\subsubsection{Proof of the local existence}\label{sec lp}
In this part we use the iteration described in Section \ref{sec li} to obtain the local existence. Suppose that $ g_0^{ a b } $, $ k_0^{ a b } $ and $ f_0 $ are initial data satisfying the conditions of Lemma \ref{lem ls}, and let $ { g_j }^{ a b } $, $ { k_j }^{ a b } $ and $ f_j $ be the corresponding solutions on $ [ t_0 , T ] $.

Now, we will show that the sequences $ { g_j }^{ a b } $, $ { k_j }^{ a b } $ and $ f_j $ converge. We use the equation \eqref{EB iter0} to obtain the equation of $ f_{ j + 1 } - f_j $:
\begin{align*}
& \frac{ \partial ( f_{ j + 1 } - f_j ) \langle p \rangle^r }{ \partial t } + ( \lambda_{ \eta , j } f_{ j + 1 } - \lambda_{ \eta , j - 1 } f_j ) \langle p \rangle^r \\
& = ( \det g_j^{ - 1 } )^{ \frac12 } \iint \frac{ h_j^{ 1 - b } \sqrt{ s_j } }{ { p_j }^0 { q_j }^0 } f_j ( p_j' ) f_j ( q_j' ) \langle p \rangle^r w_j^{ - 1 } ( q ) \, d \omega \, d q \\
& \quad - ( \det g_{ j - 1 }^{ - 1 } )^{ \frac12 } \iint \frac{ h_{ j - 1 }^{ 1 - b } \sqrt{ s_{ j - 1 } } }{ { p_{ j - 1 } }^0 { q_{ j - 1 } }^0 } f_{ j - 1 } ( p_{ j - 1 }' ) f_{ j - 1 } ( q_{ j - 1 }' ) \langle p \rangle^r w_{ j - 1 }^{ - 1 } ( q ) \, d \omega \, d q \\
& \quad - ( \det g_j^{ - 1 } )^{ \frac12 } \iint \frac{ h_j^{ 1 - b } \sqrt{ s_j } }{ { p_j }^0 { q_j }^0 } f_{ j + 1 } ( p ) f_j ( q ) \langle p \rangle^r w_j^{ - 1 } ( q ) \, d \omega \, d q \\
& \quad + ( \det g_{ j - 1 }^{ - 1 } )^{ \frac12 } \iint \frac{ h_{ j - 1 }^{ 1 - b } \sqrt{ s_{ j - 1 } } }{ { p_{ j - 1 } }^0 { q_{ j - 1 } }^0 } f_{ j } ( p ) f_{ j - 1 } ( q ) \langle p \rangle^r w_{ j - 1 }^{ - 1 } ( q ) \, d \omega \, d q , 
\end{align*}
where $ r > 5 $. Here, we note that
\begin{align*}
\lambda_{ \eta , j } f_{ j + 1 } - \lambda_{ \eta , j - 1 } f_j = \lambda_{ \eta , j } ( f_{ j + 1 } - f_j ) + ( \lambda_{ \eta , j } - \lambda_{ \eta , j - 1 } ) f_j ,
\end{align*}
where $ \lambda_{ \eta , j } \geq 0 $ by Lemmas \ref{lem wdecay} and \ref{lem ls}. Hence, we obtain
\begin{align}
\frac{ \partial | f_{ j + 1 } - f_j | \langle p \rangle^r }{ \partial t } \leq F + G + L , \label{est fj}
\end{align}
where $ F $, $ G $ and $ L $ are given by
\begin{align*}
F & = | \lambda_{ \eta , j } - \lambda_{ \eta , j - 1 } | f_j \langle p \rangle^r , \\ 
G & = \left| ( \det g_j^{ - 1 } )^{ \frac12 } - ( \det g_{ j - 1 }^{ - 1 } )^{ \frac12 } \right| \iint \frac{ h_j^{ 1 - b } \sqrt{ s_j } }{ { p_j }^0 { q_j }^0 } f_{ j } ( p_j' ) f_j ( q_j' ) \langle p \rangle^r w_j^{ - 1 } ( q ) \, d \omega \, d q \\
& \quad + ( \det g_{ j - 1 }^{ - 1 } )^{ \frac12 } \iint \left| \frac{ h_j^{ 1 - b } \sqrt{ s_j } }{ { p_j }^0 { q_j }^0 } - \frac{ h_{ j - 1 }^{ 1 - b } \sqrt{ s_{ j - 1 } } }{ { p_{ j - 1 } }^0 { q_{ j - 1 } }^0 } \right| f_{ j } ( p_j' ) f_{ j  } ( q_j' ) \langle p \rangle^r w_{ j  }^{ - 1 } ( q ) \, d \omega \, d q \\
& \quad + ( \det g_{ j - 1 }^{ - 1 } )^{ \frac12 } \iint \frac{ h_{ j - 1 }^{ 1 - b } \sqrt{ s_{ j - 1 } } }{ { p_{ j - 1 } }^0 { q_{ j - 1 } }^0 } \left| f_{ j } ( p_j' ) f_j ( q_j' ) - f_{ j - 1 } ( p_j' ) f_{ j - 1 } ( q_j' ) \right| \langle p \rangle^r w_{ j  }^{ - 1 } ( q ) \, d \omega \, d q \\
& \quad + ( \det g_{ j - 1 }^{ - 1 } )^{ \frac12 } \iint \frac{ h_{ j - 1 }^{ 1 - b } \sqrt{ s_{ j - 1 } } }{ { p_{ j - 1 } }^0 { q_{ j - 1 } }^0 } \left| f_{ j - 1 } ( p_j' ) f_{ j - 1 } ( q_j' ) - f_{ j - 1 } ( p_{ j - 1 }' ) f_{ j - 1 } ( q_{ j - 1 }' ) \right| \langle p \rangle^r w_{ j  }^{ - 1 } ( q ) \, d \omega \, d q \\
& \quad + ( \det g_{ j - 1 }^{ - 1 } )^{ \frac12 } \iint \frac{ h_{ j - 1 }^{ 1 - b } \sqrt{ s_{ j - 1 } } }{ { p_{ j - 1 } }^0 { q_{ j - 1 } }^0 } f_{ j - 1 } ( p_{ j - 1 }' ) f_{ j - 1 } ( q_{ j - 1 }' ) \langle p \rangle^r \left| w_{ j  }^{ - 1 } ( q )  - w_{ j - 1 }^{ - 1 } ( q ) \right| \, d \omega \, d q , \\ 
& =: G_1 + G_2 + G_3 + G_4 + G_5 , \allowdisplaybreaks \\ 
L & = \left| ( \det g_j^{ - 1 } )^{ \frac12 } - ( \det g_{ j - 1 }^{ - 1 } )^{ \frac12 } \right| \iint \frac{ h_j^{ 1 - b } \sqrt{ s_j } }{ { p_j }^0 { q_j }^0 } f_{ j + 1 } ( p ) f_j ( q ) \langle p \rangle^r w_j^{ - 1 } ( q ) \, d \omega \, d q \\
& \quad + ( \det g_{ j - 1 }^{ - 1 } )^{ \frac12 } \iint \left| \frac{ h_j^{ 1 - b } \sqrt{ s_j } }{ { p_j }^0 { q_j }^0 } - \frac{ h_{ j - 1 }^{ 1 - b } \sqrt{ s_{ j - 1 } } }{ { p_{ j - 1 } }^0 { q_{ j - 1 } }^0 } \right| f_{ j +1 } ( p ) f_{ j  } ( q ) \langle p \rangle^r w_{ j  }^{ - 1 } ( q ) \, d \omega \, d q \\
& \quad + ( \det g_{ j - 1 }^{ - 1 } )^{ \frac12 } \iint \frac{ h_{ j - 1 }^{ 1 - b } \sqrt{ s_{ j - 1 } } }{ { p_{ j - 1 } }^0 { q_{ j - 1 } }^0 } \left| f_{ j + 1 } ( p ) f_j ( q ) - f_{ j } ( p ) f_{ j - 1 } ( q ) \right| \langle p \rangle^r w_{ j  }^{ - 1 } ( q ) \, d \omega \, d q \\
& \quad + ( \det g_{ j - 1 }^{ - 1 } )^{ \frac12 } \iint \frac{ h_{ j - 1 }^{ 1 - b } \sqrt{ s_{ j - 1 } } }{ { p_{ j - 1 } }^0 { q_{ j - 1 } }^0 } f_{ j } ( p ) f_{ j - 1 } ( q ) \langle p \rangle^r \left| w_{ j  }^{ - 1 } ( q )  - w_{ j - 1 }^{ - 1 } ( q ) \right| \, d \omega \, d q \\
& = : L_1 + L_2 + L_3 + L_4 .
\end{align*}
Next, we use \eqref{EB iter1} to obtain the differential inequality for $ { g_{ j + 1 } }^{ a b } - { g_j }^{ a b } $: 
\begin{align}
\frac{ d | { g_{ j + 1 } }^{ a b } - { g_j }^{ a b } | }{ d t } & \leq 2 | { k_j }^{ a b } - { k_{ j - 1 } }^{ a b } | , \label{est gj} 
\end{align}
and use \eqref{EB iter2} to obtain the equation of $ { k_{ j + 1 } }^{ a b } - { k_j }^{ a b } $: 
\begin{align*}
\frac{ d ( { k_{ j + 1 } }^{ a b } - { k_j }^{ a b } ) }{ d t } & = - 2 { k_j }^a_c { k_j }^{ b c } - k_j { k_j }^{ a b } + { S_{ j + 1 } }^{ a b } + \frac12 ( \rho_{ j + 1 } - S_{ j + 1 } ) { g_j }^{ a b } + \Lambda { g_j }^{ a b } \\ 
& \quad + 2 { k_{ j - 1 } }^a_c { k_{ j - 1 } }^{ b c } + k_{ j - 1 } { k_{ j - 1 } }^{ a b } - { S_{ j } }^{ a b } - \frac12 ( \rho_{ j } - S_{ j } ) { g_{ j - 1 } }^{ a b } - \Lambda { g_{ j - 1 } }^{ a b } , 
\end{align*} 
so that we have 
\begin{align}
\frac{ d | { k_{ j + 1 } }^{ a b } - { k_j }^{ a b } | }{ d t } & \leq K_1 + K_2 + \frac12 K_3 + \Lambda | { g_j }^{ a b } - { g_{ j - 1 } }^{ a b } | , \label{est kj} 
\end{align}
where $ K_1 $, $ K_2 $ and $ K_3 $ are given by
\begin{align*}
K_1  & = | 2 { k_j }^a_c { k_j }^{ b c } + k_j { k_j }^{ a b } - 2 { k_{ j - 1 } }^a_c { k_{ j - 1 } }^{ b c } - k_{ j - 1 } { k_{ j - 1 } }^{ a b } | , \\
K_2 & = | { S_{ j + 1 } }^{ a b } - { S_{ j } }^{ a b } | , \\ 
K_3 & = | ( \rho_{ j + 1 } - S_{ j + 1 } ) { g_j }^{ a b } - ( \rho_{ j } - S_{ j } ) { g_{ j - 1 } }^{ a b } | . 
\end{align*}

Below, we will estimate $ F $, $ G_i $, $ L_i $ and $ K_i $ in terms of $ { g_j }^{ a b } - { g_{ j - 1 } }^{ a b } $, $ { k_j }^{ a b } - { k_{ j - 1 } }^{ a b } $ and $ f_j - f_{ j - 1 } $. The following lemmas will be needed.

\begin{lemma} \label{lem detg} 
Let $ A $ and $ B $ be positive definite $ n \times n $ matrices. For any $ 0 \leq \alpha \leq 1 $, we have 
\begin{align*}
\det ( \alpha A + ( 1 - \alpha ) B ) \geq ( \det A )^\alpha ( \det B )^{ 1 - \alpha } , 
\end{align*}
where the equality holds, if and only if $ A = B $. 
\end{lemma}
\begin{proof}
We refer to Corollary 7.6.8 of \cite{HJ} for the proof. 
\end{proof}

\begin{lemma}\label{lem gs}
Let $ { g_{ j - 1 } }^{ a b } $ and $ { g_j }^{ a b } $ be the solutions of Lemma \ref{lem ls}, and define $ { g_\varrho }^{ a b } $ by
\begin{align*}
{ g_\varrho }^{ a b } = ( \varrho - j + 1 ) { g_j }^{ a b } + ( j - \varrho ) { g_{ j - 1 } }^{ a b } ,
\end{align*}
for $ j - 1 \leq \varrho \leq j $. Then, we have 
\begin{align}
\| g_\varrho^{ - 1 } \| \leq C Z^{ - 2 } , \qquad Z^6 \det g_\varrho^{ - 1 } \geq \frac{ 1 }{ C } , \qquad \| e_\varrho^{ - 1 } \| \leq C Z , \label{lem gs1} 
\end{align}
where $ e_\varrho^{ - 1 } $ is the orthonormal frame associated with the metric $ { g_\varrho }^{ a b } $. Moreover, $ { p_{ j - 1 } }^0 $,  $ { p_\varrho }^0 $ and $ { p_j }^0 $ are all equivalent for $ j - 1 \leq \varrho \leq j $, i.e.~
\begin{align}
\frac{ 1 }{ C } { p_\varrho }^0 \leq { p_j }^0 \leq C { p_\varrho }^0 , \label{lem gs2}
\end{align}
where $ { p_\varrho }^0 = p^0 ( { g_\varrho }^{ a b } ) $ with $ p^0 $ as a function of $ g^{ a b } $. 
\end{lemma}
\begin{proof}
The first estimate of \eqref{lem gs1} is clear by the definition of $ { g_\varrho }^{ a b } $:
\begin{align*}
| { g_\varrho }^{ a b } | & = | ( \varrho - j + 1 ) { g_j }^{ a b } + ( j - \varrho ) { g_{ j - 1 } }^{ a b } | \\
& \leq ( \varrho - j + 1 ) | { g_j }^{ a b } | + ( j - \varrho ) | { g_{ j - 1 } }^{ a b } | \\ 
& \leq C Z^{ - 2 } . 
\end{align*}
For the second estimate, we use Lemma \ref{lem detg} with $ \alpha = \varrho - j + 1 $:
\begin{align*}
Z^6 \det g_\varrho^{ - 1 } \geq Z^6 ( \det g_j^{ - 1 } )^{ \varrho - j + 1 } ( \det g_{ j - 1 }^{ - 1 } )^{ j - \varrho } \geq \frac{ 1 }{ C } ,
\end{align*}
where we used Lemma \ref{lem assum lb}. We notice that $ { g_\varrho }^{ a b } $ satisfies 
\begin{align}
\frac{ 1 }{ C } | p |^2 \leq Z^2 { g_\varrho }^{ a b } p_a p_b \leq C | p |^2 , \label{lem gs3}
\end{align}
which is also clear by the definition of $ { g_\varrho }^{ a b } $. Hence, we obtain the estimate of $ e_\varrho^{ - 1 } $ by Lemma \ref{lem assum lb} and Remark \ref{rem assum}. The equivalence \eqref{lem gs2} is also clear by \eqref{lem gs3}, and this completes the proof of the lemma.
\end{proof}

\begin{lemma} \label{lem Eg}
Suppose that $ E = E ( g^{ a b } ) $ is differentiable with respect to $ g^{ a b } $. Then, for any $ { g_j }^{ a b } $ and $ { g_{ j - 1 } }^{ a b } $, we have 
\begin{align*}
| E ( { g_j }^{ a b } ) - E ( { g_{ j - 1 } }^{ a b } ) | \leq C \sup \left\{ \left| \frac{ \partial E }{ \partial g^{ u v } } ( { g_\varrho }^{ a b } ) \right| : j - 1 \leq \varrho \leq j \mbox{ and } u , v = 1 , 2 , 3 \right\} \| g_j^{ - 1 } - g_{ j - 1 }^{ - 1 } \| ,  
\end{align*}
where $ { g_\varrho }^{ a b } = ( \varrho - j + 1 ) { g_j }^{ a b } + ( j - \varrho ) { g_{ j - 1 } }^{ a b } $. 
\end{lemma}
\begin{proof}
The quantity $ E ( { g_j }^{ a b } ) - E ( { g_{ j - 1 } }^{ a b } ) $ can be written as 
\begin{align*}
E ( { g_j }^{ a b } ) - E ( { g_{ j - 1 } }^{ a b } ) & = \int_{ j - 1 }^j \frac{ d }{ d { \varrho } } E ( { g_\varrho }^{ a b } ) \, d \varrho \\ 
& = \int_{ j - 1 }^j \frac{ \partial E }{ \partial g^{ u v } } ( { g_\varrho }^{ a b } ) \frac{ \partial { g_\varrho }^{ u v } }{ \partial \varrho } \, d \varrho . 
\end{align*}
Since $ \partial { g_\varrho }^{ u v } / \partial \varrho = { g_j }^{ u v } - { g_{ j - 1 } }^{ u v } $, we obtain the desired result. 
\end{proof}

The above lemma is basically a mean value theorem. This will also be applied to the estimates of post-collision momenta, since $ p' $ and $ q' $ depend on $ g^{ a b } $, which are required to estimate $ G_4 $. We need the following lemma.

\begin{lemma} \label{lem dp'dg}
Suppose that a metric $ g^{ a b } $ satisfies 
\begin{align*}
\frac{ 1 }{ C } | p |^2 \leq g^{ a b } p_a p_b \leq C | p |^2 , \qquad \| g^{ - 1 } \| \leq C , \qquad \det g^{ - 1 } \geq \frac{ 1 }{ C } , 
\end{align*}
on an interval $ [ t_0 , T ] $. Then, we have 
\begin{align*}
\left| \frac{ \partial p'_a }{ \partial g^{ u v } } \right| + \left| \frac{ \partial q'_a }{ \partial g^{ u v } } \right| \leq C \langle p \rangle \langle q \rangle^4 , 
\end{align*}
where $ C $ depends on $ T $. 
\end{lemma}
\begin{proof}
We refer to Lemma 6 of \cite{LLN23} for the proof. 
\end{proof}

The quantities $ F $, $ G_i $, $ L_i $ and $ K_i $ will be estimated by using Lemma \ref{lem Eg} extensively. To avoid lengthy computations in the proof of Proposition \ref{prop local}, we collect the following basic computations. Recall \eqref{dpdp est} and apply Lemma \ref{lem gs} to obtain 
\begin{align}
\left| \frac{ { p_\varrho }^a }{ { p_\varrho }^0 } \right| \leq C Z^{ - 1 } , \qquad \left| \frac{ p_a }{ { p_\varrho }^0 } \right| \leq C Z , \label{est pp bound} 
\end{align}
where $ { p_\varrho }^a = { g_\varrho }^{ a b } p_b $ and $ { p_\varrho }^0 = p^0 ( { g_\varrho }^{ a b } ) $ with $ p^0 $ as a function of $ g^{ a b } $. Then, we recall \eqref{est dpdg} and \eqref{est d1/pdg} to have 
\begin{align}
\left| \frac{ \partial p^0 }{ \partial g^{ u v } } ( { g_\varrho }^{ a b } ) \right| & \leq C Z \langle p \rangle , \label{est dpdgs} \\
\left| \frac{ \partial }{ \partial g^{ u v } } \left( \frac{ 1 }{ p^0 } \right) ( { g_\varrho }^{ a b } ) \right| & \leq C Z \frac{ \langle p \rangle }{ ( { p_\varrho }^0 )^2 } \leq C Z^2 \frac{ 1 }{ { p_\varrho }^0 } \leq C Z^2 \frac{ 1 }{ { p_j }^0 } , \label{est d1/pdgs}
\end{align}
where we used in \eqref{est d1/pdgs} the second inequality of \eqref{est pp bound} and the equivalence \eqref{lem gs2} in Lemma \ref{lem gs}. We note that $ { p_j }^0 $ on the right end of \eqref{est d1/pdgs} can be replaced by $ { p_{ j - 1 } }^0 $. Next, we notice that $ \det g^{ - 1 } $ is a polynomial of degree $ 3 $ in $ g^{ a b } $. Hence, we use Lemma \ref{lem gs} to obtain 
\begin{align}
\left| \frac{ \partial ( \det g^{ - 1 } )^{ \frac12 } }{ \partial g^{ u v } } ( { g_\varrho }^{ a b } ) \right| \leq C \frac{ \| g_\varrho^{ - 1 } \|^2 }{ ( \det g_\varrho^{ - 1 } )^{ \frac12 } } \leq C Z^{ - 1 } . \label{est ddetgdgs}
\end{align}
We also need to estimate $ \partial h / \partial g^{ u v } $ evaluated at $ { g_\varrho }^{ a b } $. Recall \eqref{est dhdg}, and use Lemma \ref{lem gs} to obtain
\begin{align}
\left| \frac{ \partial h }{ \partial g^{ u v } } ( { g_\varrho }^{ a b } ) \right| & \leq C Z^2 h_\varrho ( { p_\varrho }^0 { q_\varrho }^0 )^2 \leq C Z^2 ( { p_\varrho }^0 { q_\varrho }^0 )^{ \frac52 } \leq C Z^2 ( { p_j }^0 { q_j }^0 )^{ \frac52 } , \label{est dhdgs}
\end{align}
where $ h_\varrho $ is the relative momentum associated with the metric $ { g_\varrho }^{ a b } $, and we applied \eqref{hspq} and \eqref{lem gs2}. By a similar computation we obtain the same estimate for $ \sqrt{ s } $: 
\begin{align}
\left| \frac{ \partial \sqrt{ s } }{ \partial g^{ u v } } ( { g_\varrho }^{ a b } ) \right| & \leq C Z^2 ( { p_j }^0 { q_j }^0 )^{ \frac52 } . \label{est dsdgs}
\end{align}
Finally, we obtain for $ h^{ 1 - b } $, 
\begin{align}
\left| \frac{ \partial ( h^{ 1 - b } ) }{ \partial g^{ u v } } ( { g_\varrho }^{ a b } ) \right| & \leq C h_\varrho^{ - b } \left| \frac{ \partial h }{ \partial g^{ u v } } ( { g_\varrho }^{ a b } ) \right| \leq C Z^2 h_\varrho^{ - b } ( { p_j }^0 { q_j }^0 )^{ \frac52 } . \label{est dhbdgs}
\end{align}
Here, we do not claim that $ h_\varrho $ and $ h_j $ are equivalent, but we can obtain the estimate \eqref{est L2hb} below.

Now, we collect the following estimates: 
\begin{align}
| { p_j }^0 - { p_{ j - 1 } }^0 | & \leq C Z \langle p \rangle \| g_j^{ - 1 } - g_{ j - 1 }^{ - 1 } \| , \label{est F1} \\ 
\left| \frac{ 1 }{ { p_j }^0 } - \frac{ 1 }{ { p_{ j - 1 } }^0 } \right| & \leq C Z^2 \frac{ 1 }{  { p_j }^0 } \| g_j^{ - 1 } - g_{ j - 1 }^{ - 1 } \| , \label{est F2} \\
\left| ( \det g_j^{ - 1 } )^{ \frac12 } - ( \det g_{ j - 1 }^{ - 1 } )^{ \frac12 } \right| & \leq C Z^{ - 1 } \| g_j^{ - 1 } - g_{ j - 1 }^{ - 1 } \| , \label{est L1} \\
| { h_j } - h_{ j - 1 } | & \leq C Z^2 ( { p_j }^0 { q_j }^0 )^{ \frac52 } \| g_j^{ - 1 } - g_{ j - 1 }^{ - 1 } \| , \label{est L2h} \\ 
| \sqrt{ s_j } - \sqrt{ s_{ j - 1 } } | & \leq C Z^2 ( { p_j }^0 { q_j }^0 )^{ \frac52 } \| g_j^{ - 1 } - g_{ j - 1 }^{ - 1 } \| , \label{est L2s} \\ 
| h_j^{ 1 - b } - h_{ j - 1 }^{ 1 - b } | & \leq C Z^2 ( { p_j }^0 { q_j }^0 )^{ \frac52 } \left( \int_{ j - 1 }^j h_\varrho^{ - b } \, d \varrho \right) \| g_j^{ - 1 } - g_{ j - 1 }^{ - 1 } \| , \label{est L2hb} 
\end{align}
which are obtained by applying Lemma \ref{lem Eg} to \eqref{est dpdgs}, \eqref{est d1/pdgs}, \eqref{est ddetgdgs}, \eqref{est dhdgs}, \eqref{est dsdgs} and \eqref{est dhbdgs}, respectively. We note that $ { p_j }^0 $ and $ { q_j }^0 $ on the right hand sides of \eqref{est F1}--\eqref{est L2hb} can be replaced by $ { p_{ j - 1 } }^0 $ and $ { q_{ j - 1 } }^0 $, respectively. The quantities $ F $, $ G_i $, $ L_i $ and $ K_i $ are now estimated as follows. \medskip

\noindent {\bf Estimate of $ F $.} Recall that $ F $ is given by 
\begin{align*}
F = | \lambda_{ \eta , j } - \lambda_{ \eta , j - 1 } | f_j \langle p \rangle^r . 
\end{align*}
We note that $ \lambda_{ \eta , j } - \lambda_{ \eta , j - 1 } $ can be written as 
\begin{align*}
& \lambda_{ \eta , j } - \lambda_{ \eta , j - 1 } \\
& = - \left( ( 1 - \eta ) \gamma ( { p_j }^0 - { p_{ j - 1 } }^0 ) - \frac{ ( { k_j }^{ a b } - { k_{ j - 1 } }^{ a b } ) p_a p_b }{ { p_j }^0 } - { k_{ j - 1 } }^{ a b } p_a p_b \left( \frac{ 1 }{ { p_j }^0 } - \frac{ 1 }{ { p_{ j - 1 } }^0 } \right) \right) Z_\eta . 
\end{align*}
We apply \eqref{est F1} and \eqref{est F2} together with the estimates \eqref{lem ls1} and \eqref{est pp bound} to obtain
\begin{align*}
| \lambda_{ \eta , j } - \lambda_{ \eta , j - 1 } | \leq C Z \langle p \rangle ( \| g_j^{ - 1 } - g_{ j - 1 }^{ - 1 } \| + \| k_j^{ ** } - k_{ j - 1 }^{ ** } \| ) Z_\eta . 
\end{align*}
This implies that 
\begin{align}
F \leq C Z Z_\eta \| f_j \|_{ L^\infty_{ r + 1 } } ( \| g_j^{ - 1 } - g_{ j - 1 }^{ - 1 } \| + \| k_j^{ ** } - k_{ j - 1 }^{ ** } \| ) . \label{Fjj} 
\end{align}
\medskip

\noindent{\bf Estimate of $ G_1 $.} Recall that $ G_1 $ is given by 
\begin{align*}
G_1 = \left| ( \det g_j^{ - 1 } )^{ \frac12 } - ( \det g_{ j - 1 }^{ - 1 } )^{ \frac12 } \right| \iint \frac{ h_j^{ 1 - b } \sqrt{ s_j } }{ { p_j }^0 { q_j }^0 } f_{ j } ( p_j' ) f_j ( q_j' ) \langle p \rangle^r w_j^{ - 1 } ( q ) \, d \omega \, d q , 
\end{align*}
and apply \eqref{est L1} and Lemma \ref{lem bracket} to obtain
\begin{align}
G_1 & \leq C Z^{ - 1 } \| g_j^{ - 1 } - g_{ j - 1 }^{ - 1 } \| \| f_{ j } \|_{ L^\infty_r }^2 \int \frac{ 1 }{ h_j^b } w_j^{ - 1 } \, d q \nonumber \\ 
& \leq C Z^{ - 1 } \| g_j^{ - 1 } - g_{ j - 1 }^{ - 1 } \| f_j \|_{ L^\infty_r }^2 \frac{ 1 }{ ( \det g_j^{ - 1 } )^{ \frac12 } } \frac{ Z_\eta^{ - \frac32 + \frac{ b }{ 2 } } }{ ( { p_j }^0 )^{ \frac{ b }{ 2 } } } \nonumber \\ 
& \leq C Z^2 Z_\eta^{ - \frac32 + \frac{ b }{ 2 } } \| g_j^{ - 1 } - g_{ j - 1 }^{ - 1 } \| \| f_j \|_{ L^\infty_r }^2 , \label{G1jj}
\end{align}
where we used Lemmas \ref{lem west}, \ref{lem assum lb} and \ref{lem ls}. \medskip

\noindent {\bf Estimate of $ L_1 $.} Recall that $ L_1 $ is 
\begin{align*}
L_1 = \left| ( \det g_j^{ - 1 } )^{ \frac12 } - ( \det g_{ j - 1 }^{ - 1 } )^{ \frac12 } \right| \iint \frac{ h_j^{ 1 - b } \sqrt{ s_j } }{ { p_j }^0 { q_j }^0 } f_{ j + 1 } ( p ) f_j ( q ) \langle p \rangle^r w_j^{ - 1 } ( q ) \, d \omega \, d q . 
\end{align*}
The estimate is almost the same as for $ G_1 $, and we obtain 
\begin{align}
L_1 & \leq C Z^2 Z_\eta^{ - \frac32 + \frac{ b }{ 2 } } \| g_j^{ - 1 } - g_{ j - 1 }^{ - 1 } \| \| f_{ j + 1 } \|_{ L^\infty_r } \| f_j \|_{ L^\infty_r } . \label{L1jj} 
\end{align}
\medskip

\noindent {\bf Estimate of $ G_2 $.} We write $ G_2 $ as  
\begin{align*}
G_2 = ( \det g_{ j - 1 }^{ - 1 } )^{ \frac12 } \iint \left| B ( { g_j }^{ a b } ) - B ( { g_{ j - 1 } }^{ a b } ) \right| f_{ j } ( p_j' ) f_{ j  } ( q_j' ) \langle p \rangle^r w_{ j  }^{ - 1 } ( q ) \, d \omega \, d q , 
\end{align*}
where $ B = B ( g^{ a b } ) $ is given by 
\begin{align*}
B = \frac{ h^{ 1 - b } \sqrt{ s } }{ p^0 q^0 } , 
\end{align*}
which will be understood as a function of $ g^{ a b } $ for given $ p $ and $ q $. To apply Lemma \ref{lem Eg}, we write $ B ( { g_j }^{ a b } ) - B ( { g_{ j - 1 } }^{ a b } ) $ as 
\begin{align*}
B ( { g_j }^{ a b } ) - B ( { g_{ j - 1 } }^{ a b } ) 
& = \left( \frac{ 1 }{ { p_j }^0 } - \frac{ 1 }{ { p_{ j - 1 } }^0 } \right) \frac{ h_{ j }^{ 1 - b } \sqrt{ s_{ j } } }{ { q_{ j } }^0 } + \left( \frac{ 1 }{ { q_j }^0 } - \frac{ 1 }{ { q_{ j - 1 } }^0 } \right) \frac{ h_{ j }^{ 1 - b } \sqrt{ s_{ j } } }{ { p_{ j - 1 } }^0 } \\
& \quad + ( \sqrt{ s_j } - \sqrt{ s_{ j - 1 } } ) \frac{ h_{ j }^{ 1 - b } }{ { p_{ j - 1 } }^0 { q_{ j - 1 } }^0 } + ( h_{ j }^{ 1 - b } - h_{ j - 1 }^{ 1 - b } ) \frac{ \sqrt{ s_{ j - 1 } } }{ { p_{ j - 1 } }^0 { q_{ j - 1 } }^0 } . 
\end{align*}
We apply \eqref{est F2}, \eqref{est L2s} and \eqref{est L2hb} to obtain 
\begin{align*}
& \begin{aligned}
\left| \left( \frac{ 1 }{ { p_j }^0 } - \frac{ 1 }{ { p_{ j - 1 } }^0 } \right) \frac{ h_{ j }^{ 1 - b } \sqrt{ s_{ j } } }{ { q_{ j } }^0 } \right| & \leq C Z^2 \frac{ 1 }{ { p_j }^0 } \| g_j^{ - 1 } - g_{ j - 1 }^{ - 1 } \| \frac{ h_{ j }^{ 1 - b } \sqrt{ s_{ j } } }{ { q_{ j } }^0 } \\ 
& \leq C Z^2 \frac{ 1 }{ h_j^b } \| g_j^{ - 1 } - g_{ j - 1 }^{ - 1 } \| , 
\end{aligned} \allowdisplaybreaks \\
& \begin{aligned}
\left| \left( \frac{ 1 }{ { q_j }^0 } - \frac{ 1 }{ { q_{ j - 1 } }^0 } \right) \frac{ h_{ j }^{ 1 - b } \sqrt{ s_{ j } } }{ { p_{ j - 1 } }^0 } \right| & \leq C Z^2 \frac{ 1 }{ { q_j }^0 } \| g_j^{ - 1 } - g_{ j - 1 }^{ - 1 } \| \frac{ h_{ j }^{ 1 - b } \sqrt{ s_{ j } } }{ { p_{ j - 1 } }^0 } \\
& \leq C Z^2 \frac{ 1 }{ h_j^b } \| g_j^{ - 1 } - g_{ j - 1 }^{ - 1 } \| ,  
\end{aligned} \allowdisplaybreaks \\ 
& \begin{aligned}
\left| ( \sqrt{ s_j } - \sqrt{ s_{ j - 1 } } ) \frac{ h_{ j }^{ 1 - b } }{ { p_{ j - 1 } }^0 { q_{ j - 1 } }^0 } \right| & \leq C Z^2 ( { p_j }^0 { q_j }^0 )^{ \frac52 } \| g_j^{ - 1 } - g_{ j - 1 }^{ - 1 } \| \frac{ h_{ j }^{ 1 - b } }{ { p_{ j - 1 } }^0 { q_{ j - 1 } }^0 } \\ 
& \leq C Z^2 ( { p_j }^0 { q_j }^0 )^2 \frac{ 1 }{ h_j^b } \| g_j^{ - 1 } - g_{ j - 1 }^{ - 1 } \| ,  
\end{aligned} \allowdisplaybreaks \\ 
& \begin{aligned}
\left| ( h_{ j }^{ 1 - b } - h_{ j - 1 }^{ 1 - b } ) \frac{ \sqrt{ s_{ j - 1 } } }{ { p_{ j - 1 } }^0 { q_{ j - 1 } }^0 } \right| & \leq C Z^2 ( { p_j }^0 { q_j }^0 )^{ \frac52 } \left( \int_{ j - 1 }^j h_\varrho^{ - b } \, d \varrho \right) \| g_j^{ - 1 } - g_{ j - 1 }^{ - 1 } \| \frac{ \sqrt{ s_{ j - 1 } } }{ { p_{ j - 1 } }^0 { q_{ j - 1 } }^0 } \\ 
& \leq C Z^2 ( { p_j }^0 { q_j }^0 )^2 \left( \int_{ j - 1 }^j h_\varrho^{ - b } \, d \varrho \right) \| g_j^{ - 1 } - g_{ j - 1 }^{ - 1 } \| , 
\end{aligned} 
\end{align*}
where we used the equivalence \eqref{lem gs2} appropriately. The quantity $ ( { p_j }^0 )^2 $ is bounded by 
\begin{align*}
( { p_j }^0 )^2 = 1 + { g_j }^{ a b } p_a p_b \leq 1 + C Z^{ - 2 } | p |^2 \leq C \langle p \rangle^2 \leq C \langle p_j' \rangle^2 \langle q_j' \rangle^2 , 
\end{align*}
where we used Lemma \ref{lem bracket} in the last inequality. We combine the above estimates to obtain 
\begin{align*}
G_2 & \leq C Z^2 \| f_{ j } \|_{ L^\infty_{ r + 2 } }^2 \left( ( \det g_{ j - 1 }^{ - 1 } )^{ \frac12 } \int_{ \bbr^3 } \frac{ ( { q_j }^0 )^2 }{ h_j^b } w_{ j  }^{ - 1 } ( q ) \, d q \right) \| g_j^{ - 1 } - g_{ j - 1 }^{ - 1 } \| \\ 
& \quad + C Z^2 \| f_{ j } \|_{ L^\infty_{ r + 2 } }^2 \left( ( \det g_{ j - 1 }^{ - 1 } )^{ \frac12 } \int_{ \bbr^3 } \int_{ j - 1 }^j \frac{ ( { q_j }^0 )^2 }{ h_\varrho^b } w_{ j  }^{ - 1 } ( q ) \, d \varrho \, d q \right) \| g_j^{ - 1 } - g_{ j - 1 }^{ - 1 } \| . 
\end{align*}
Here, the integration of $ ( { q_j }^0 )^2 h_j^{ - b } w_j^{ - 1 } $ can be estimated by Lemma \ref{lem west}:
\begin{align*}
( \det g_{ j - 1 }^{ - 1 } )^{ \frac12 } \int_{ \bbr^3 } \frac{ ( { q_j }^0 )^2 }{ h_j^b } w_{ j  }^{ - 1 } ( q ) \, d q & \leq C ( \det g_{ j }^{ - 1 } )^{ \frac12 } \int_{ \bbr^3 } \frac{ ( { q_j }^0 )^2 }{ h_j^b } w_{ j  }^{ - 1 } ( q ) \, d q \\ 
& \leq C \frac{ Z_\eta^{ - \frac32 + \frac{ b }{ 2 } } }{ ( { p_j }^0 )^{ \frac{ b }{ 2 } } } . 
\end{align*}
On the other hand, for the integration of $ ( { q_j }^0 )^2 h_\varrho^{ - b } w_j^{ - 1 } $ we notice that 
\begin{align*}
{ q_j }^0 - 1 = \frac{ ( { q_j }^0 )^2 - 1 }{ { q_j }^0 + 1 } = \frac{ { g_j }^{ a b } q_a q_b }{ { q_j }^0 + 1 } \geq \frac{ 1 }{ C } \frac{ { g_\varrho }^{ a b } q_a q_b }{ { q_\varrho }^0 + 1 } \geq \frac{ 1 }{ C } ( { q_\varrho }^0 - 1 ) , 
\end{align*}
and this shows that 
\begin{align}
w_j ( q ) = e^{ Z_\eta ( { q_j }^0 - 1 ) } \geq e^{ \frac{ 1 }{ C } Z_\eta ( { q_\varrho }^0 - 1 ) } = w_\varrho^{ \frac{ 1 }{ C } } ( q ) . \label{wjws}
\end{align}
Now, we obtain by Lemma \ref{lem west}, 
\begin{align*}
( \det g_{ j - 1 }^{ - 1 } )^{ \frac12 } \int_{ \bbr^3 } \int_{ j - 1 }^j \frac{ ( { q_j }^0 )^2 }{ h_\varrho^b } w_{ j  }^{ - 1 } ( q ) \, d \varrho \, d q & \leq C \int_{ j - 1 }^j ( \det g_{ \varrho }^{ - 1 } )^{ \frac12 } \int_{ \bbr^3 } \frac{ ( { q_\varrho }^0 )^2 }{ h_\varrho^b } w_{ \varrho }^{ - \frac{ 1 }{ C } } ( q ) \, d q \, d \varrho \\ 
& \leq C \frac{ Z_\eta^{ - \frac32 + \frac{ b }{ 2 } } }{ ( { p_j }^0 )^{ \frac{ b }{ 2 } } } . 
\end{align*}
We conclude that
\begin{align}
G_2 \leq C Z^2 Z_\eta^{ - \frac32 + \frac{ b }{ 2 } } \| f_j \|_{ L^\infty_{ r + 2 } }^2 \| g_j^{ - 1 } - g_{ j - 1 }^{ - 1 } \| . \label{G2jj}
\end{align}
\medskip

\noindent {\bf Estimate of $ L_2 $.} We write $ L_2 $ as  
\begin{align*}
L_2 = ( \det g_{ j - 1 }^{ - 1 } )^{ \frac12 } \iint \left| B ( { g_j }^{ a b } ) - B ( { g_{ j - 1 } }^{ a b } ) \right| f_{ j +1 } ( p ) f_{ j  } ( q ) \langle p \rangle^r w_{ j  }^{ - 1 } ( q ) \, d \omega \, d q , 
\end{align*}
where $ B = B ( g^{ a b } ) $ is given by 
\begin{align*}
B = \frac{ h^{ 1 - b } \sqrt{ s } }{ p^0 q^0 } . 
\end{align*}
The estimate of $ B ( { g_j }^{ a b } ) - B ( { g_{ j - 1 } }^{ a b } ) $ is almost the same as in $ G_2 $, but we may only consider 
\begin{align*}
{ p_j }^0 \leq C \langle p \rangle .
\end{align*}
Then, we will obtain 
\begin{align}
L_2 \leq C Z^2 Z_\eta^{ - \frac32 + \frac{ b }{ 2 } } \| f_{ j + 1 } \|_{ L^\infty_{ r + 2 } } \| f_j \|_{ L^\infty_{ r } } \| g_j^{ - 1 } - g_{ j - 1 }^{ - 1 } \| . \label{L2jj} 
\end{align}
\medskip

\noindent {\bf Estimate of $ G_3 $.} We may write $ G_3 $ as  
\begin{align*}
G_3 & = ( \det g_{ j - 1 }^{ - 1 } )^{ \frac12 } \iint \frac{ h_{ j - 1 }^{ 1 - b } \sqrt{ s_{ j - 1 } } }{ { p_{ j - 1 } }^0 { q_{ j - 1 } }^0 } \left| f_{ j } ( p_j' ) f_j ( q_j' ) - f_{ j - 1 } ( p_j' ) f_{ j - 1 } ( q_j' ) \right| \langle p \rangle^r \ w_{ j  }^{ - 1 } ( q ) \, d \omega \, d q \\ 
& = ( \det g_{ j - 1 }^{ - 1 } )^{ \frac12 } \iint \frac{ h_{ j - 1 }^{ 1 - b } \sqrt{ s_{ j - 1 } } }{ { p_{ j - 1 } }^0 { q_{ j - 1 } }^0 } \left| ( f_{ j } - f_{ j - 1 } ) ( p_j' ) f_j ( q_j' ) + f_{ j - 1 } ( p_j' ) ( f_j - f_{ j - 1 } ) ( q_j' ) \right| \langle p \rangle^r w_{ j  }^{ - 1 } ( q ) \, d \omega \, d q . 
\end{align*}
Then, we estimate it as follows: 
\begin{align}
G_3 & \leq C \| f_{ j } - f_{ j - 1 } \|_{ L^\infty_r } ( \| f_j \|_{ L^\infty_r } + \| f_{ j - 1 } \|_{ L^\infty_r } ) ( \det g_{ j - 1 }^{ - 1 } )^{ \frac12 } \int \frac{ 1 }{ h_{ j - 1 }^b } w_{ j  }^{ - 1 } ( q ) \, d q \nonumber \\ 
& \leq C \| f_{ j } - f_{ j - 1 } \|_{ L^\infty_r } ( \| f_j \|_{ L^\infty_r } + \| f_{ j - 1 } \|_{ L^\infty_r } ) ( \det g_{ j - 1 }^{ - 1 } )^{ \frac12 } \int \frac{ 1 }{ h_{ j - 1 }^b } w_{ j - 1 }^{ - \frac{ 1 }{ C } } ( q ) \, d q \nonumber \\ 
& \leq C Z_\eta^{ - \frac32 + \frac{ b }{ 2 } } \| f_{ j } - f_{ j - 1 } \|_{ L^\infty_r } ( \| f_j \|_{ L^\infty_r } + \| f_{ j - 1 } \|_{ L^\infty_r } ) , \label{G3jj} 
\end{align}
by using \eqref{wjws} and Lemma \ref{lem west}. \medskip

\noindent {\bf Estimate of $ L_3 $.} We may write $ L_3 $ as  
\begin{align*}
L_3 & = ( \det g_{ j - 1 }^{ - 1 } )^{ \frac12 } \iint \frac{ h_{ j - 1 }^{ 1 - b } \sqrt{ s_{ j - 1 } } }{ { p_{ j - 1 } }^0 { q_{ j - 1 } }^0 } \left| f_{ j + 1 } ( p ) f_j ( q ) - f_{ j } ( p ) f_{ j - 1 } ( q ) \right| \langle p \rangle^r w_{ j  }^{ - 1 } ( q ) \, d \omega \, d q \\ 
& = ( \det g_{ j - 1 }^{ - 1 } )^{ \frac12 } \iint \frac{ h_{ j - 1 }^{ 1 - b } \sqrt{ s_{ j - 1 } } }{ { p_{ j - 1 } }^0 { q_{ j - 1 } }^0 } \left| ( f_{ j + 1 } - f_j ) ( p ) f_j ( q ) + f_{ j } ( p ) ( f_j - f_{ j - 1 } ) ( q ) \right| \langle p \rangle^r w_{ j  }^{ - 1 } ( q ) \, d \omega \, d q . 
\end{align*}
Then, by the same computations as in $ G_3 $, we obtain 
\begin{align}
L_3 & \leq C Z_\eta^{ - \frac32 + \frac{ b }{ 2 } } ( \| f_{ j + 1 } - f_j \|_{ L^\infty_r } + \| f_j - f_{ j - 1 } \|_{ L^\infty_r } ) \| f_j \|_{ L^\infty_r } . \label{L3jj} 
\end{align}
\medskip

\noindent{\bf Estimate of $ G_4 $.} We notice that $ G_4 $ can be written as
\begin{align*}
G_4 = ( \det g_{ j - 1 }^{ - 1 } )^{ \frac12 } \iint \frac{ h_{ j - 1 }^{ 1 - b } \sqrt{ s_{ j - 1 } } }{ { p_{ j - 1 } }^0 { q_{ j - 1 } }^0 } \left| F ( { g_j }^{ a b } ) - F ( { g_{ j - 1 } }^{ a b } ) \right| w_{ j  }^{ - 1 } ( q ) \, d \omega \, d q , 
\end{align*}
where $ F = F ( g^{ a b } ) $ is given by 
\begin{align*}
F ( { g }^{ a b } ) = f_{ j - 1 } ( p' ) f_{ j - 1 } ( q' ) \langle p \rangle^r , 
\end{align*}
which will be understood as a function of $ g^{ a b } $ for given $ f_{ j - 1 } $, $ p $, $ q $ and $ \omega $. To apply Lemma \ref{lem Eg}, we compute 
\begin{align*}
\frac{ \partial F }{ \partial g^{ u v } } = \frac{ \partial f_{ j - 1 } }{ \partial p_a } ( p' ) \frac{ \partial p'_a }{ \partial g^{ u v } } f_{ j - 1 } ( q' ) \langle p \rangle^r + f_{ j - 1 } ( p' ) \frac{ \partial f_{ j - 1 } }{ \partial p_a } ( q' ) \frac{ \partial q'_a }{ \partial g^{ u v } } \langle p \rangle^r . 
\end{align*}
Here, we can use Lemma \ref{lem dp'dg} to estimate $ \partial p'_a / \partial g^{ u v } $ and $ \partial q'_a / \partial g^{ u v } $ as $ \langle p \rangle \langle q \rangle^4 $. But, on a finite time, the quantities $ \langle q \rangle^4 $ and $ ( { q_{ j - 1 } }^0 )^4 $ are equivalent, so that we can estimate as follows:
\begin{align*}
& \left| \frac{ \partial F }{ \partial g^{ u v } } ( { g_\varrho }^{ a b } ) \right| \\
& \leq C \left( \max_{ a } \left| \frac{ \partial f_{ j - 1 } }{ \partial p_a } ( p_\varrho' ) \right| \langle p \rangle \langle q \rangle^4 f_{ j - 1 } ( q_\varrho' ) \langle p \rangle^r + f_{ j - 1 } ( p_\varrho' ) \max_a \left| \frac{ \partial f_{ j - 1 } }{ \partial p_a } ( q_\varrho' ) \right| \langle p \rangle \langle q \rangle^4 \langle p \rangle^r \right) \\ 
& \leq C \left( \max_{ a } \left| \frac{ \partial f_{ j - 1 } }{ \partial p_a } ( p_\varrho' ) \right| \langle p_\varrho' \rangle^{ r + 1 } \langle q_\varrho' \rangle^{ r + 1 } f_{ j - 1 } ( q_\varrho' ) + f_{ j - 1 } ( p_\varrho' ) \max_a \left| \frac{ \partial f_{ j - 1 } }{ \partial p_a } ( q_\varrho' ) \right| \langle p_\varrho' \rangle^{ r + 1 } \langle q_\varrho' \rangle^{ r + 1 } \right) ( { q_{ j - 1 } }^0 )^4 \\
& \leq C \left\| \frac{ \partial f_{ j - 1 } }{ \partial p } \right\|_{ L^\infty_{ r + 1 } } \| f_{ j - 1 } \|_{ L^\infty_{ r + 1 } } ( { q_{ j - 1 } }^0 )^4 , 
\end{align*}
where we used Lemma \ref{lem bracket}. Hence, we have 
\begin{align*}
| F ( { g_j }^{ a b } ) - F ( { g_{ j - 1 } }^{ a b } ) | & \leq C \| f_{ j - 1 } \|^2_{ W^{ 1 , \infty }_{ r + 1 } } \| g_j^{ - 1 } - g_{ j - 1 }^{ - 1 } \| ( { q_{ j - 1 } }^0 )^4 , 
\end{align*}
which shows that 
\begin{align}
G_4 & \leq C \| f_{ j - 1 } \|^2_{ W^{ 1 , \infty }_{ r + 1 } } \left( ( \det g_{ j - 1 }^{ - 1 } )^{ \frac12 } \int \frac{ ( { q_{ j - 1 } }^0 )^4 }{ h_{ j - 1 }^{ b } } w_{ j  }^{ - 1 } ( q ) \, d q \right) \| g_j^{ - 1 } - g_{ j - 1 }^{ - 1 } \| \nonumber \\ 
& \leq C \| f_{ j - 1 } \|^2_{ W^{ 1 , \infty }_{ r + 1 } } \left( ( \det g_{ j - 1 }^{ - 1 } )^{ \frac12 } \int \frac{ ( { q_{ j - 1 } }^0 )^4 }{ h_{ j - 1 }^{ b } } w_{ j - 1 }^{ - \frac{ 1 }{ C } } ( q ) \, d q \right) \| g_j^{ - 1 } - g_{ j - 1 }^{ - 1 } \| \nonumber \\ 
& \leq C Z_\eta^{ - \frac32 + \frac{ b }{ 2 } } \| f_{ j - 1 } \|^2_{ W^{ 1 , \infty }_{ r + 1 } } \| g_j^{ - 1 } - g_{ j - 1 }^{ - 1 } \| , \label{G4jj} 
\end{align} 
where we used \eqref{wjws} and Lemma \ref{lem west}. \medskip

\noindent{\bf Estimate of $ G_5 $.} Recall that $ G_5 $ is 
\begin{align*}
G_5 = ( \det g_{ j - 1 }^{ - 1 } )^{ \frac12 } \iint \frac{ h_{ j - 1 }^{ 1 - b } \sqrt{ s_{ j - 1 } } }{ { p_{ j - 1 } }^0 { q_{ j - 1 } }^0 } f_{ j - 1 } ( p_{ j - 1 }' ) f_{ j - 1 } ( q_{ j - 1 }' ) \langle p \rangle^r | w_{ j  }^{ - 1 } ( q )  - w_{ j - 1 }^{ - 1 } ( q ) | \, d \omega \, d q . 
\end{align*}
In order to apply Lemma \ref{lem Eg}, we compute
\begin{align*}
\frac{ \partial w^{ - 1 } ( q ) }{ \partial g^{ u v } } & = \frac{ \partial }{ \partial g^{ u v } } \left( e^{ - Z_\eta ( q^0 - 1 ) } \right) = - Z_\eta \frac{ \partial q^0 }{ \partial g^{ u v } } w^{ - 1 } ( q ) , 
\end{align*}
and use \eqref{est dpdgs} to obtain 
\begin{align*}
\left| \frac{ \partial w^{ - 1 } }{ \partial g^{ u v } } ( { g_\varrho }^{ a b } ) \right| \leq C Z Z_\eta \langle q \rangle w_\varrho^{ - 1 } ( q ) \leq C Z Z_\eta \langle q \rangle w_{ j - 1 }^{ - \frac{ 1 }{ C } } ( q ) , 
\end{align*}
by an argument similar to \eqref{wjws}. Hence, we have 
\begin{align}
\left| w_{ j  }^{ - 1 } ( q )  - w_{ j - 1 }^{ - 1 } ( q ) \right| & \leq C Z Z_\eta \langle q \rangle w_{ j - 1 }^{ - \frac{ 1 }{ C } } ( q ) \| g_j^{ - 1 } - g_{ j - 1 }^{ - 1 } \| , \label{wjj}
\end{align}
and this shows that $ G_5 $ is estimated as follows: 
\begin{align}
G_5 & \leq C Z Z_\eta \| f_{ j - 1 } \|_{ L^\infty_{ r + 1 } }^2 \left( ( \det g_{ j - 1 }^{ - 1 } )^{ \frac12 } \int \frac{ 1 }{ h_{ j - 1 }^b } w_{ j - 1 }^{ - \frac{ 1 }{ C } } ( q ) \, d q \right) \| g_j^{ - 1 } - g_{ j - 1 }^{ - 1 } \| \nonumber \\
& \leq C Z Z_\eta^{ - \frac12 + \frac{ b }{ 2 } } \| f_{ j - 1 } \|_{ L^\infty_{ r + 1 } }^2 \| g_j^{ - 1 } - g_{ j - 1 }^{ - 1 } \| , \label{G5jj} 
\end{align}
where we used $ \langle q \rangle \leq C \langle p_{ j - 1 }' \rangle \langle q_{ j - 1 }' \rangle $. \medskip

\noindent {\bf Estimate of $ L_4 $.} Recall that $ L_4 $ is 
\begin{align*}
L_4 = ( \det g_{ j - 1 }^{ - 1 } )^{ \frac12 } \iint \frac{ h_{ j - 1 }^{ 1 - b } \sqrt{ s_{ j - 1 } } }{ { p_{ j - 1 } }^0 { q_{ j - 1 } }^0 } f_{ j } ( p ) f_{ j - 1 } ( q ) \langle p \rangle^r | w_{ j  }^{ - 1 } ( q )  - w_{ j - 1 }^{ - 1 } ( q ) | \, d \omega \, d q . 
\end{align*}
The estimate of $ L_4 $ is almost the same as for $ G_5 $. We will obtain 
\begin{align}
L_4 & \leq C Z Z_\eta^{ - \frac12 + \frac{ b }{ 2 } } \| f_j \|_{ L^\infty_{ r + 1 } } \| f_{ j - 1 } \|_{ L^\infty_{ r + 1 } } \| g_j^{ - 1 } - g_{ j - 1 }^{ - 1 } \| . \label{L4jj} 
\end{align}
\medskip

\noindent{\bf Estimate of $ K_1 $.} Recall that $ K_1 $ is
\begin{align*}
K_1 = | 2 { k_j }^a_c { k_j }^{ b c } + k_j { k_j }^{ a b } - 2 { k_{ j - 1 } }^a_c { k_{ j - 1 } }^{ b c } - k_{ j - 1 } { k_{ j - 1 } }^{ a b } | . 
\end{align*}
Here, $ { k_j }^a_c = { k_j }^{ a d } { g_j }_{ c d } $ and $ k_j = { k_j }^{ c d } { g_j }_{ c d } $, where $ { g_j }_{ c d } $ is the inverse of $ { g_j }^{ c d } $, so that it is given by a polynomial of degree $ 2 $ in $ { g_j }^{ a b } $ multiplied by $ ( \det g_j^{ - 1 } )^{ - 1 } $. Note that the quantities $ { g_j }^{ a b } $, $ { k_j }^{ a b } $, $ ( \det g_j^{ - 1 } )^{ - 1 } $, $ { g_{ j - 1 } }^{ a b } $, $ { k_{ j - 1 } }^{ a b } $ and $ ( \det g_{ j - 1 }^{ - 1 } )^{ - 1 } $ are all bounded on $ [ t_0 , T ] $. Hence, $ K_1 $ is estimated in terms of $ { g_j }^{ a b } - { g_{ j - 1 } }^{ a b } $, $ { k_j }^{ a b } - { k_{ j - 1 } }^{ a b } $ and $ ( \det g_{ j }^{ - 1 } )^{ - 1 } - ( \det g_{ j - 1 }^{ - 1 } )^{ - 1 } $. We use \eqref{est L1} to obtain 
\begin{align}
K_1 \leq C ( \| g_j^{ - 1 } - g_{ j - 1 }^{ - 1 } \| + \| k_j^{ ** } - k_{ j - 1 }^{ ** } \| ) . \label{K1jj} 
\end{align} 
\medskip

\noindent{\bf Estimates of $ K_2 $ and $ K_3 $.} Recall that $ K_2 = | { S_{ j + 1 } }^{ a b } - { S_j }^{ a b } | $, where $ { S_{ j + 1 } }^{ a b } $ is given by
\begin{align*}
{ S_{ j + 1 } }^{ a b } & = ( \det g_j^{ - 1 } )^{ \frac12 } \int_{ \bbr^3 } f_{ j + 1 } \frac{ { p_j }^a { p_j }^b }{ { p_j }^0 } w_j^{ - 1 } ( p ) \, d p . 
\end{align*}
By the same arguments as in the proof of Lemma \ref{lem le}, we obtain 
\begin{align*}
| { S_{ j + 1 } }^{ a b } | \leq C \| f_{ j + 1 } \|_{ L^\infty_m } , \qquad | \rho_{ j + 1 } | \leq C \| f_{ j + 1 } \|_{ L^\infty_m } , \qquad | S_{ j + 1 } | \leq C \| f_{ j + 1 } \|_{ L^\infty_m } , 
\end{align*}
where the factor $ Z $ has been ignored, since it is bounded on $ [ t_0 , T ] $. Now, we compute as 
\begin{align}
K_2 & \leq | ( \det g_j^{ - 1 } )^{ \frac12 } - ( \det g_{ j - 1 }^{ - 1 } )^{ \frac12 } | \int_{ \bbr^3 } f_{ j + 1 } \frac{ | { p_j }^a { p_j }^b | }{ { p_j }^0 } w_j^{ - 1 } ( p ) \, d p \nonumber \\
& \quad + ( \det g_{ j - 1 }^{ - 1 } )^{ \frac12 } \int_{ \bbr^3 } | f_{ j + 1 } - f_j | \frac{ | { p_j }^a { p_j }^b | }{ { p_j }^0 } w_j^{ - 1 } ( p ) \, d p \nonumber \\ 
& \quad + ( \det g_{ j - 1 }^{ - 1 } )^{ \frac12 } \int_{ \bbr^3 } f_j \left| \frac{ { p_j }^a { p_j }^b }{ { p_j }^0 } - \frac{ { p_{ j - 1 } }^a { p_{ j - 1 } }^b }{ { p_{ j - 1 } }^0 } \right| w_j^{ - 1 } ( p ) \, d p \nonumber \\ 
& \quad + ( \det g_{ j - 1 }^{ - 1 } )^{ \frac12 } \int_{ \bbr^3 } f_j \frac{ | { p_{ j - 1 } }^a { p_{ j - 1 } }^b | }{ { p_{ j - 1 } }^0 } | w_j^{ - 1 } ( p ) - w_{ j - 1 }^{ - 1 } ( p ) | \, d p \nonumber .
\end{align}
For the second term, we notice that 
\begin{align*}
\int_{ \bbr^3 } | f_{ j + 1 } - f_j | \frac{ | { p_j }^a { p_j }^b | }{ { p_j }^0 } w_j^{ - 1 } ( p ) \, d p & \leq C \| f_{ j + 1 } - f_j \|_{ L^\infty_r } \int_{ \bbr^3 } \langle p \rangle^{ - r + 2 } \, d p \\ 
& \leq C \| f_{ j + 1 } - f_j \|_{ L^\infty_r } , 
\end{align*}
where we need to assume $ r > 5 $. For the fourth term, we use \eqref{wjj} to obtain 
\begin{align*}
\int_{ \bbr^3 } f_j \frac{ | { p_{ j - 1 } }^a { p_{ j - 1 } }^b | }{ { p_{ j - 1 } }^0 } | w_j^{ - 1 } ( p ) - w_{ j - 1 }^{ - 1 } ( p ) | \, d p & \leq C \| g_j^{ - 1 } - g_{ j - 1 }^{ - 1 } \| \int_{ \bbr^3 } f_j \frac{ \langle p \rangle^3 }{ { p_{ j - 1 } }^0 } w_{ j - 1 }^{ - \frac{ 1 }{ C } } ( p ) \, d p \\ 
& \leq C \| g_j^{ - 1 } - g_{ j - 1 }^{ - 1 } \| \int_{ \bbr^3 } f_j \langle p \rangle^{ 2 } \, d p \\ 
& \leq C \| g_j^{ - 1 } - g_{ j - 1 }^{ - 1 } \| , 
\end{align*}
where we used \eqref{est pp bound}. The other terms are easily estimated by using \eqref{est F2} and \eqref{est L1}, and $ K_2 $ is then estimated as follows:
\begin{align}
K_2 \leq C ( \| g_j^{ - 1 } - g_{ j - 1 }^{ - 1 } \| + \| f_{ j + 1 } - f_j \|_{ L^\infty_r } ) . \label{K2jj} 
\end{align}
Finally, $ K_3 $ is estimated as follows: 
\begin{align}
K_3 \leq C ( \| g_j^{ - 1 } - g_{ j - 1 }^{ - 1 } \| + \| f_{ j + 1 } - f_j \|_{ L^\infty_r } ) , \label{K3jj} 
\end{align}
which can be proved by the same arguments as for $ K_2 $, and we skip the details. \medskip

We combine the above estimates to obtain the following result. 

\begin{prop}\label{prop local}
Suppose that $ g_0^{ a b } $ is a symmetric and positive definite matrix, $ k_0^{ a b } $ is a symmetric matrix satisfying $ H_0 > \gamma $ and $ F_0 < 1 / 4 $, and $ 0 \leq f_0 \in W^{ 1 , \infty }_{ r + 2 } $ with $ r > 5 $. Let $ g_0^{ a b } $, $ k_0^{ a b } $ and $ f_0 $ be initial data of the equations \eqref{E1}--\eqref{B} with $ \eta $ chosen by \eqref{eta}, satisfying the constraint equation \eqref{constraint}. Then, there exist an interval $ [ t_0 , T ] $ and a unique solution $ g^{ a b } $, $ k^{ a b } $ and $ f \geq 0 $ such that $ g^{ a b } $ and $ k^{ a b } $ are differentiable on $ [ t_0 , T ] $, $ f $ is continuous on $ [ t_0 , T ] \times \bbr^3 $, $ S^{ a b } $, $ \rho $ and $ S $ are differentiable on $ [ t_0 , T ] $. Moreover, the solution satisfies 
\begin{align}
\frac{ 1 }{ C } | p |^2 \leq Z^2 g^{ a b } p_a p_b \leq C | p |^2 , \qquad \| k^{ * * } \| \leq C Z^{ - 2 } , \qquad H \geq \gamma , \qquad F \leq \frac{ \eta^2 }{ 4 } , \label{prop est1}
\end{align}
and
\begin{align}
\sup_{ t_0 \leq s \leq T } \| f ( s ) \|_{ L^\infty_r } \leq C , \label{prop est2}
\end{align}
for all $ t_0 \leq t \leq T $.
\end{prop}
\begin{proof}
Since $ g_0^{ a b } $ is positive definite and $ k_0^{ a b } $ is a constant matrix, we can choose $ A > 0 $ satisfying the first two conditions of \eqref{lem lsi1} of Lemma \ref{lem ls}. The third and fourth conditions of \eqref{lem lsi1} are already assumed, so that we can use Lemma \ref{lem ls}. We combine all the estimates \eqref{Fjj}--\eqref{L1jj}, \eqref{G2jj}--\eqref{G4jj} and \eqref{G5jj}--\eqref{K3jj}, and apply them to \eqref{est fj}--\eqref{est kj} to obtain 
\begin{align*}
& \frac{ d }{ d t } \| f_{ j + 1 } - f_j \|_{ L^\infty_r } \\ 
& \leq C ( 1 + \sup_i \| f_i \|_{ W^{ 1 , \infty }_{ r + 2 } } )^2 ( \| g_j^{ - 1 } - g_{ j - 1 }^{ - 1 } \| + \| k_j^{ ** } - k_{ j - 1 }^{ ** } \| + \| f_j - f_{ j - 1 } \|_{ L^\infty_r } + \| f_{ j + 1 } - f_j \|_{ L^\infty_r } ) , 
\end{align*}
and
\begin{align*} 
& \frac{ d }{ d t } \| g_{ j + 1 }^{ - 1 } - g_j^{ - 1 } \| \leq C \| k_j^{ ** } - k_{ j - 1 }^{ ** } \| , \\ 
& \frac{ d }{ d t } \| k_{ j + 1 }^{ ** } - k_j^{ ** } \| \leq C ( \| g_j^{ - 1 } - g_{ j - 1 }^{ - 1 } \| + \| k_j^{ ** } - k_{ j - 1 }^{ ** } \| + \| f_{ j + 1 } - f_j \|_{ L^\infty_r } ) . 
\end{align*}
By Lemma \ref{lem ls} with $ m = r + 2 $, we have $ \sup_i \| f_i \|_{ W^{ 1 , \infty }_{ r + 2 } } \leq C $ on $ [ t_0 , T ] $. Let us write 
\begin{align*}
X_{ j + 1 } ( t ) = \sup_{ t_0 \leq s \leq t } \| g_{ j + 1 }^{ - 1 } ( s ) - g_j^{ - 1 } ( s ) \| + \sup_{ t_0 \leq s \leq t } \| k_{ j + 1 }^{ ** } ( s ) - k_j^{ ** } ( s ) \| + \sup_{ t_0 \leq s \leq t } \| f_{ j + 1 } ( s ) - f_j ( s ) \|_{ L^\infty_r } .
\end{align*}
Then, we obtain by the Gr{\"o}nwall lemma, 
\begin{align*}
X_{ j + 1 } ( t ) \leq C \int_{ t_0 }^t X_j ( s ) \, d s , 
\end{align*}
on $ [ t_0 , T ] $, where $ C $ depends on $ T $. Iterating the above we obtain
\begin{align*}
X_{ n + 1 } ( t ) \leq C \frac{ C^n ( t - t_0 )^n }{ n ! } \leq \frac{ C^n }{ n ! } , 
\end{align*}
which implies that $ { g_j }^{ a b } $, $ { k_j }^{ a b } $ and $ f_j $ converge to some $ g^{ a b } $, $ k^{ a b } $ and $ f $, which are continuous. Now, $ g^{ a b } $ and $ k^{ a b } $ are differentiable on $ [ t_0 , T ] $ by the equations \eqref{E1}--\eqref{E2}, and the estimates \eqref{prop est1}--\eqref{prop est2} are obtained by the same arguments as in Lemma \ref{lem ls}.

Differentiability of $ S^{ a b } $, $ \rho $ and $ S $ can be verified as follows. Recall that $ S_{ a b } $ is given by
\begin{align*}
( \det g^{ - 1 } )^{ \frac12 } \int_{ \bbr^3 } f ( t , p ) \frac{ p_a p_b }{  p^0 } w^{ - 1 } ( p ) \, d p . 
\end{align*}
Note that 
\begin{align*}
\left| \frac{ p_a p_b }{  p^0 } w^{ - 1 } ( p ) \right| \leq C , 
\end{align*}
by \eqref{dpdp est} with $ C $ depending on $ T $, and 
\begin{align*}
\frac{ \partial }{ \partial t } \left( \frac{ p_a p_b }{  p^0 } w^{ - 1 } \right) & = - \frac{ p_a p_b }{ ( p^0 )^2 } \frac{ \partial p^0 }{ \partial t } w^{ - 1 } + \frac{ p_a p_b }{ p^0 } \frac{ \partial w^{ - 1 } }{ \partial t } \\ 
& = \frac{ p_a p_b k^{ c d } p_c p_d }{ ( p^0 )^3 } w^{ - 1 } - \frac{ p_a p_b }{ p^0 } \left( ( 1 - \eta ) \gamma ( p^0 - 1 ) - \frac{ k^{ a b } p_a p_b }{ p^0 } \right) Z_\eta w^{ - 1 } , 
\end{align*}
so that we have
\begin{align*}
\left| \frac{ \partial }{ \partial t } \left( \frac{ p_a p_b }{  p^0 } w^{ - 1 } \right) \right| \leq C . 
\end{align*}
Moreover, we have $ f \leq C \langle p \rangle^{ - r } $, and
\begin{align*}
\left| \frac{ \partial f }{ \partial t } \right| & \leq | \lambda_\eta | f + ( \det g )^{ - \frac12 } \iint \frac{ h^{ 1 - b } \sqrt{ s } }{ p^0 q^0 } ( f ( p' ) f ( q' ) + f ( p ) f ( q ) ) w^{ - 1 } ( q ) \, d \omega \, d q \\ 
& \leq C \langle p \rangle f + C \iint \frac{ 1 }{ h^b } ( \langle p' \rangle^{ - r } \langle q' \rangle^{ - r } + \langle p \rangle^{ - r } \langle q \rangle^{ - r } ) w^{ - 1 } ( q ) \, d \omega \, d q \\
& \leq C \langle p \rangle^{ - r + 1 } + C \langle p \rangle^{ - r } \int \frac{ 1 }{ h^b } w^{ - 1 } ( q ) \, d q \\ 
& \leq C \langle p \rangle^{ - r + 1 } , 
\end{align*}
by Lemmas \ref{lem west} and \ref{lem bracket}, so that $ f $ and $ \partial f / \partial t $ are integrable. The above computation shows that $ S_{ a b } $ is differentiable with respect to $ t $. Similarly, $ S^{ a b } $, $ \rho $ and $ S $ are differentiable, and this completes the proof of the proposition. 
\end{proof}

\section{Global existence and asymptotic behavior}\label{sec global}
In this section we prove Theorem \ref{main}. We will obtain the global existence and asymptotic behavior for the Einstein-Boltzmann system \eqref{E1}--\eqref{B}. In order to obtain the global existence we will first show that the solutions satisfy the constraint equation \eqref{constraint} on the whole interval of existence. Then, we will derive several basic estimates of the Hubble variable, and we will use them to obtain the global existence.

Let us first consider the differential equations of $ H $ and $ \rho $. The Hubble variable $ H $ is defined by \eqref{H}, and we can use \eqref{E1} and \eqref{E2} to obtain
\begin{align}
\frac{ d H }{ d t } & = \frac13 \dot{ g }_{ a b } k^{ a b } + \frac13 g_{ a b } \dot{ k }^{ a b } \nonumber \\
& = \frac23 k_{ a b } k^{ a b } + \frac13 g_{ a b } ( - 2 k^a_c k^{ b c } - k k^{ a b } + S^{ a b } + \frac12 ( \rho - S ) g^{ a b } + \Lambda g^{ a b } ) \nonumber \\
& = \frac23 k_{ a b } k^{ a b } - \frac23 k_{ b c } k^{ b c } - \frac13 k^2 + \frac13 S + \frac12 ( \rho - S ) + \Lambda \nonumber \\
& = - 3 H^2 + \frac12 \rho - \frac16 S + \Lambda . \label{dHdt}
\end{align}
We also need the differential equation of $ \rho $, defined by \eqref{rho}. We can use 
\begin{align*}
\frac{ d \det g }{ d t } = 2 k \det g ,
\end{align*}
to obtain
\begin{align}
\frac{ d \rho }{ d t } & = - \frac12 ( \det g )^{ - \frac32 } \frac{ d \det g }{ d t } \int f ( p ) p^0 w^{ - 1 } ( p ) \, d p \nonumber \\
& \quad + ( \det g )^{ - \frac12 } \int \frac{ \partial ( f ( p ) w^{ - 1 } ( p ) ) }{ \partial t } p^0 + f ( p ) w^{ - 1 } ( p ) \frac{ \partial p^0 }{ \partial t } \, d p \nonumber \\
& = - k ( \det g )^{ - \frac12 } \int f ( p ) p^0 w^{ - 1 } ( p ) \, d p - ( \det g )^{ - \frac12 } \int f ( p ) w^{ - 1 } ( p ) \frac{ k^{ a b } p_a p_b }{ p^0 } \, d p \nonumber \\
& = - k \rho - k^{ a b } S_{ a b } . \label{drhodt}
\end{align}
Here, we used the fact that $ f w^{ - 1 } $ is the solution of the original (unweighted) Boltzmann equation \eqref{Boltzmann}, so that it satisfies \eqref{Qconserved}. Now, we use \eqref{dHdt} and \eqref{drhodt} to show that the initial constraint \eqref{constraint} propagates as time evolves. We write the constraint equation as
\begin{align*}
- k_{ a b } k^{ a b } + k^2 - 2 \rho - 2 \Lambda = 0 ,
\end{align*}
and let $ \xi $ denote the quantity on the left hand side. By using \eqref{Einstein2}, \eqref{E2}, \eqref{dHdt} and \eqref{drhodt}, we obtain
\begin{align*}
\frac{ d \xi }{ d t } 
& = - \left( 2 k_a^c k_{ b c } - k k_{ a b } + S_{ a b } + \frac12 ( \rho - S ) g_{ a b } + \Lambda g_{ a b } \right) k^{ a b } \\
& \quad - k_{ a b } \left( - 2 k^a_c k^{ b c } - k k^{ a b } + S^{ a b } + \frac12 ( \rho - S ) g^{ a b } + \Lambda g^{ a b } \right) \\
& \quad + 18 H \left( - 3 H^2 + \frac12 \rho - \frac16 S + \Lambda \right) + 2 k \rho + 2 k^{ a b } S_{ a b } \\
& = - 2 k_a^c k_{ b c } k^{ a b } + k k_{ a b } k^{ a b } - S_{ a b } k^{ a b } - \frac12 ( \rho - S ) k - \Lambda k \\
& \quad + 2 k_{ a b } k^a_c k^{ b c } + k k_{ a b } k^{ a b } - S_{ a b } k^{ a b } - \frac12 ( \rho - S ) k - \Lambda k \\
& \quad - 2 k^3 + 3 k \rho - k S + 6 k \Lambda + 2 k \rho + 2 k^{ a b } S_{ a b } \\
& = 2 k k_{ a b } k^{ a b } + 4 k \rho + 4 k \Lambda - 2 k^3 \\
& = - 2 k \xi .
\end{align*}
Therefore, if $ \xi ( t_0 ) = 0 $, then $ \xi ( t ) = 0 $ for any $ t \geq t_0 $, which implies that the constraint equation holds on the whole interval of existence. We obtain the following lemma.

\begin{lemma}\label{lem constraint}
Let $ g^{ a b } $, $ k^{ a b } $ and $ f $ be a solution of Proposition \ref{prop local}. Then, the solution satisfies
\begin{align}
6 H^2 = \sigma_{ a b } \sigma^{ a b } + 2 \rho + 2 \Lambda , \label{lem constraint1}
\end{align}
on the whole interval of existence.
\end{lemma}
\begin{proof}
Recall that the constraint \eqref{constraint} is equivalent to \eqref{constraint1}. This proves the lemma by the preceding arguments.
\end{proof}

\begin{lemma}\label{lem H}
Let $ g^{ a b } $, $ k^{ a b } $ and $ f $ be a solution of Proposition \ref{prop local}. Then, the solution satisfies
\begin{align}
\frac{ d H }{ d t } & \leq 0 , \label{lem H1} \\
H & \geq \gamma , \label{lem H2} \\
\frac{ d H }{ d t } & \leq - ( H + \gamma ) ( H - \gamma ) , \label{lem H3}
\end{align}
on the whole interval of existence.
\end{lemma}
\begin{proof}
By Lemma \ref{lem constraint}, the Hubble variable $ H $ satisfies \eqref{lem constraint1}. We write it as
\begin{align*}
3 H^2 = \frac{ \sigma_{ a b } \sigma^{ a b } }{ 2 } + \rho + \Lambda ,
\end{align*}
and apply this to \eqref{dHdt} to obtain
\begin{align}
\frac{ d H }{ d t } & = - \left( \frac{ \sigma_{ a b } \sigma^{ a b } }{ 2 } + \rho + \Lambda \right) + \frac12 \rho - \frac16 S + \Lambda \nonumber \\
& = - \frac{ \sigma_{ a b } \sigma^{ a b } }{ 2 } - \frac12 \rho - \frac16 S . \label{lem H4}
\end{align}
Here, the right hand side is non-positive, so that we obtain \eqref{lem H1}. Next, we note that \eqref{lem constraint1} implies
\begin{align*}
H^2 \geq \gamma^2 .
\end{align*}
Since $ H_0 > \gamma $ and $ H $ is continuous, we obtain \eqref{lem H2}. Finally, we write \eqref{lem constraint1} as
\begin{align*}
2 H^2 = \frac{ \sigma_{ a b } \sigma^{ a b } }{ 3 } + \frac23 \rho + \frac23 \Lambda ,
\end{align*}
and apply this to \eqref{dHdt} to obtain
\begin{align*}
\frac{ d H }{ d t } & = - H^2 - \left( \frac{ \sigma_{ a b } \sigma^{ a b } }{ 3 } + \frac23 \rho + \frac23 \Lambda \right) + \frac12 \rho - \frac16 S + \Lambda \nonumber \\
& = - H^2 - \frac{ \sigma_{ a b } \sigma^{ a b } }{ 3 } - \frac16 \rho - \frac16 S + \frac13 \Lambda \\
& \leq - H^2 + \frac13 \Lambda .
\end{align*}
Here, $ \Lambda / 3 = \gamma^2 $, so that we obtain \eqref{lem H3}.
\end{proof}

\begin{lemma}\label{lem Ha}
Let $ g^{ a b } $, $ k^{ a b } $ and $ f $ be a solution of Proposition \ref{prop local}. Then, the solution satisfies
\begin{align}
H - \gamma & \leq ( H_0 - \gamma ) e^{ - 2 \gamma ( t - t_0 ) } , \label{lem Ha1} \\
\frac{ 1 }{ H - \gamma } & \leq \frac{ 1 }{ H_0 - \gamma } e^{ 3 ( H_0 + \gamma ) ( t - t_0 ) } , \label{lem Ha2}
\end{align}
on the whole interval of existence.
\end{lemma}
\begin{proof}
We first observe from \eqref{lem H3} that
\begin{align*}
\frac{ d }{ d t } \ln \left( \frac{ H - \gamma }{ H + \gamma } \right) \leq - 2 \gamma .
\end{align*}
Hence, we have
\begin{align*}
\ln \left( \frac{ H - \gamma }{ H_0 - \gamma } \right) \leq - 2 \gamma ( t - t_0 ) ,
\end{align*}
where we used the fact that $ H $ is decreasing, and this proves \eqref{lem Ha1}. Next, we use \eqref{lem H4} to obtain
\begin{align*}
\frac{ d }{ d t } \left( \frac{ 1 }{ H - \gamma } \right) & = - \frac{ 1 }{ ( H - \gamma )^2 } \left( - \frac{ \sigma_{ a b } \sigma^{ a b } }{ 2 } - \frac12 \rho - \frac16 S \right) \\
& \leq \frac{ 1 }{ ( H - \gamma )^2 } \left( \frac{ \sigma_{ a b } \sigma^{ a b } }{ 2 } + \frac23 \rho \right) ,
\end{align*}
since $ S \leq \rho $. Then, by using the constraint \eqref{lem constraint1}, we obtain
\begin{align*}
\frac{ d }{ d t } \left( \frac{ 1 }{ H - \gamma } \right) & \leq \frac{ 3 ( H^2 - \gamma^2 ) }{ ( H - \gamma )^2 } \leq \frac{ 3 ( H_0 + \gamma ) }{ H - \gamma } ,
\end{align*}
where we used again the fact that $ H $ is decreasing, and this proves \eqref{lem Ha2}.
\end{proof}

\subsection{Proof of the main theorem}\label{sec proof}
We are now ready to prove Theorem \ref{main}. We will prove the global existence in Section \ref{sec proof1} and will obtain the asymptotic behavior in Section \ref{sec proof2}.

\subsubsection{Global existence}\label{sec proof1}
Let us first consider the global existence. In Theorem \ref{main}, we assume that the initial data satisfy the constraint equation \eqref{constraint}, or equivalently \eqref{constraint1}, and that $ H_0 $ and $ \eta_0 $ are given by
\begin{align*}
\gamma < H_0 < \sqrt{ \frac65 } \gamma , \qquad \eta_0 = \sqrt{ 6 \left( 1 - \frac{ \gamma^2 }{ H_0^2 } \right) } .
\end{align*}
The initial value $ F_0 $ of the shear variable can be estimated by using \eqref{constraint1} as
\begin{align*}
F_0 = \frac{ 6 H_0^2 - 2 \rho_0 - 2 \Lambda }{ 4 H_0^2 } \leq \frac32 \left( 1 - \frac{ \gamma^2 }{ H_0^2 } \right) < \frac14 ,
\end{align*}
and we notice here that
\begin{align*}
\frac32 \left( 1 - \frac{ \gamma^2 }{ H_0^2 } \right) = \frac{ \eta_0^2 }{ 4 } .
\end{align*}
Hence, the assumptions of Proposition \ref{prop local} are satisfied, so that we obtain the local existence from Proposition \ref{prop local}.

Now, it will be enough to show that the estimates \eqref{prop est1} and \eqref{prop est2} of Proposition \ref{prop local} hold on $ [ t_0 , \infty ) $. For the first estimate of \eqref{prop est1} we consider the differential equation of $ Z^2 g^{ a b } p_a p_b $. By direct computations we obtain
\begin{align*}
\frac{ d }{ d t } Z^2 g^{ a b } p_a p_b & = 2 Z \dot{ Z } g^{ a b } p_a p_b - 2 Z^2 k^{ a b } p_a p_b \\
& = 2 Z^2 ( \gamma g^{ a b } p_a p_b - H g^{ a b } p_a p_b - \sigma^{ a b } p_a p_b ) \\
& \leq 2 \left( H - \gamma + \sqrt{ \sigma^{ a b } \sigma_{ a b } } \right) Z^2 g^{ a b } p_a p_b ,
\end{align*}
where we used \eqref{lem H2} in the last inequality. Then, we have
\begin{align}
\sigma^{ a b } \sigma_{ a b } \leq 6 H^2 - 2 \Lambda = 6 ( H + \gamma ) ( H - \gamma ) \leq 6 ( H_0 + \gamma ) ( H_0 - \gamma ) e^{ - 2 \gamma ( t - t_0 ) } , \label{sigma decay}
\end{align}
where we used Lemma \ref{lem constraint} in the first inequality, and \eqref{lem H1} and \eqref{lem Ha1} in the last inequality. This implies that
\begin{align*}
\int_{ t_0 }^t 2 \left( H - \gamma + \sqrt{ \sigma^{ a b } \sigma_{ a b } } \right) \, d s & \leq 2 \int_{ t_0 }^t ( H_0 - \gamma ) e^{ - 2 \gamma ( s - t_0 ) } + \sqrt{ 6 ( H_0^2 - \gamma^2 ) } e^{ - \gamma ( s - t_0 ) } \, d s \\
& \leq \frac{ 1 }{ \gamma } ( H_0 - \gamma ) + \frac{ 2 }{ \gamma } \sqrt{ 6 ( H_0^2 - \gamma^2 ) } .
\end{align*}
Hence, we obtain
\begin{align}
\frac{ 1 }{ C } e^{ 2 \gamma t_0 } g_0^{ a b } p_a p_b \leq Z^2 g^{ a b } p_a p_b \leq C e^{ 2 \gamma t_0 } g_0^{ a b } p_a p_b , \label{boots1}
\end{align}
where the constant $ C $ can be given explicitly by
\begin{align*}
C = \exp \left( \frac{ 1 }{ \gamma } ( H_0 - \gamma ) + \frac{ 2 }{ \gamma } \sqrt{ 6 ( H_0^2 - \gamma^2 ) } \right) .
\end{align*}
Next, for the second estimate of \eqref{prop est1}, we need Lemma 26.1 of \cite{Ringstrom}, which can be written in our notation as
\begin{align*}
| \sigma^{ a b } | \leq C \sqrt{ \sum_{ a , b = 1 }^3 ( g^{ a b } )^2 } \sqrt{ \sigma^{ c d } \sigma_{ c d } } \leq C \sqrt{ \sigma^{ c d } \sigma_{ c d } } \| g^{ - 1 } \| .
\end{align*}
Then, we apply \eqref{sigma decay} and Lemma \ref{lem assum lb} to obtain
\begin{align}
| \sigma^{ a b } | \leq C Z^{ - 3 } , \label{sigma decay1}
\end{align}
and use the decomposition \eqref{k decom} of $ k^{ a b } $ to have
\begin{align}
| k^{ a b } | & \leq | H g^{ a b } | + | \sigma^{ a b } | \leq C Z^{ - 2 } . \label{boots2}
\end{align}
The third estimate of \eqref{prop est1} is obvious by \eqref{lem H2}:
\begin{align}
H \geq \gamma . \label{boots3}
\end{align}
For the fourth estimate of \eqref{prop est1}, we again use Lemma \ref{lem constraint} to obtain
\begin{align*}
F = \frac{ \sigma_{ a b } \sigma^{ a b } }{ 4 H^2 } \leq \frac{ 6 H^2 - 2 \Lambda }{ 4 H^2 } = \frac32 \left( 1 - \frac{ \gamma^2 }{ H^2 } \right) ,
\end{align*}
so that we have
\begin{align}
F \leq \frac{ \eta_0^2 }{ 4 } , \label{boots4}
\end{align}
since $ H $ is decreasing. Finally, for \eqref{prop est2} we obtain from \eqref{est f_n1}
\begin{align*}
\sup_{ t_0 \leq s \leq t } \| f ( s ) \|_{ L^\infty_r } 
& \leq \| f_0 \|_{ L^\infty_r } + C \sup_{ t_0 \leq s \leq t } \| f ( s ) \|_{ L_r^{ \infty } }^2 ,
\end{align*}
since $ Z_\eta^{ - \frac32 + \frac{ b }{ 2 } } $ is integrable. Hence, there exists a small $ \varepsilon > 0 $ such that if $ \| f_0 \|_{ L^\infty_r } < \varepsilon $, then we obtain
\begin{align}
\sup_{ t_0 \leq t < \infty } \| f ( t ) \|_{ L^\infty_r } \leq C \varepsilon . \label{boots5}
\end{align}
Now, we combine \eqref{boots1}, \eqref{boots2}, \eqref{boots3}, \eqref{boots4} and \eqref{boots5} to obtain the global existence.

\subsubsection{Asymptotic behavior}\label{sec proof2}
It remains to show that the solutions satisfy the estimates \eqref{thm H}--\eqref{thm g}. We note that the estimates \eqref{thm H}, \eqref{thm F} and \eqref{thm f} are direct consequences of \eqref{lem H2}, \eqref{lem Ha1}, \eqref{sigma decay} and \eqref{boots5}.

Let us consider the estimates \eqref{thm rho}, \eqref{thm S} and \eqref{thm Srho}. By the definition of $ \rho $, and using \eqref{boots5} with Lemma \ref{lem assum lb}, we have
\begin{align*}
\rho & = ( \det g )^{ - \frac12 } \int f ( p ) p^0 w^{ - 1 } ( p ) \, d p \\
& \leq C \varepsilon Z^{ - 3 } \int \langle p \rangle^{ - r } \, d p \\
& \leq C \varepsilon e^{ - 3 \gamma t } .
\end{align*}
In a similar way, we obtain 
\begin{align*}
S & = ( \det g )^{ - \frac12 } \int f ( p ) \frac{ g^{ a b } p_a p_b }{ p^0 } w^{ - 1 } ( p ) \, d p \\
& \leq C \varepsilon Z^{ - 3 } \int \langle p \rangle^{ - r } g^{ a b } p_a p_b \, d p \\
& \leq C \varepsilon Z^{ - 5 } \int \langle p \rangle^{ - r + 2 } \, d p \\
& \leq C \varepsilon e^{ - 5 \gamma t } .
\end{align*}
In order to obtain \eqref{thm Srho}, we need to introduce
\begin{align*}
N^0 = ( \det g )^{ - \frac12 } \int f ( p ) w^{ - 1 } ( p ) \, d p ,
\end{align*}
which is the zeroth component of the first moment of $ f $. By direct computations we obtain
\begin{align*}
\frac{ d N^0 }{ d t } = - 3 H N^0 ,
\end{align*}
so that we have
\begin{align*}
N^0 = N^0 ( t_0 ) \exp \left( - 3 \int_{ t_0 }^t H ( s ) \, d s \right) ,
\end{align*}
where the integral can be estimated by using \eqref{lem Ha1} as follows:
\begin{align*}
\int_{ t_0 }^t H ( s ) \, d s \leq \gamma ( t - t_0 ) + \int_{ t_0 }^t ( H ( s ) - \gamma ) \, d s \leq \gamma ( t - t_0 ) + C .
\end{align*}
Hence, we obtain
\begin{align*}
N^0 ( t ) \geq C N^0 ( t_0 ) e^{ - 3 \gamma t } .
\end{align*}
Since $ N^0 \leq \rho $, we obtain the estimate \eqref{thm Srho}.

Finally, we consider the asymptotic behavior \eqref{thm g} of $ g^{ a b } $. In Section \ref{sec proof1}, we obtained the differential equation of $ Z^2 g^{ a b } $:
\begin{align}
\frac{ d }{ d t } Z^2 g^{ a b } = - 2 ( H - \gamma ) Z^2 g^{ a b } - 2 Z^2 \sigma^{ a b } . \label{dZgdt}
\end{align}
Then, $ Z^2 g^{ a b } $ is explicitly given by
\begin{align*}
Z^2 g^{ a b } = Z^2 ( t_0 ) g_0^{ a b } e^{ - 2 \int_{ t_0 }^t ( H - \gamma ) \, d s } - 2 \int_{ t_0 }^t Z^2 ( s ) \sigma^{ a b } ( s ) e^{ - 2 \int_s^t ( H - \gamma ) \, d \tau } \, d s ,
\end{align*}
and notice that $ \lim_{ t \to \infty } Z^2 g^{ a b } $ is well defined by applying \eqref{lem Ha1} and \eqref{sigma decay1} to the right hand side. Now, we define a constant matrix $ G^{ a b } $ by
\begin{align*}
G^{ a b } = Z^2 ( t_0 ) g_0^{ a b } e^{ - 2 \int_{ t_0 }^\infty ( H - \gamma ) \, d s } - 2 \int_{ t_0 }^\infty Z^2 ( s ) \sigma^{ a b } ( s ) e^{ - 2 \int_s^\infty ( H - \gamma ) \, d \tau } \, d s ,
\end{align*}
and integrate \eqref{dZgdt} from $ t $ to $ \infty $ to obtain \eqref{thm g}: applying Lemma \ref{lem assum lb} together with \eqref{lem Ha1} and \eqref{sigma decay1}, we finally obtain
\begin{align*}
| G^{ a b } - Z^2 ( t ) g^{ a b } ( t ) | & \leq C \int_t^\infty ( H ( s ) - \gamma + Z^{ - 1 } ( s ) ) \, d s \\
& \leq C e^{ - \gamma t } .
\end{align*}
This completes the proof of Theorem \ref{main}.


\begin{thebibliography}{00}

\bibitem{Akarsuetal}
Akarsu, O., Kumar, S., Sharma, S., Tedesco, L.:
\newblock{Constraints on a Bianchi type I spacetime extension of the standard $\Lambda$CDM model.}
\newblock \emph{Phys. Rev. D} 100 (2019), 023532.

\bibitem{BCB73}
Bancel, D., Choquet-Bruhat, Y.:
Existence, uniqueness, and local stability for the Einstein-Maxwell-Boltzman system.
{\it Comm. Math. Phys.} 33 (1973), 83--96.

\bibitem{FRS}
Fournodavlos, G., Rodnianski, I., Speck, J.:
Stable Big Bang formation for Einstein's equations: The complete sub-critical regime.
{\it J. Amer. Math. Soc.} 36 (2023), 827--916.

\bibitem{HF1}
Friedrich, H.: On the existence of $ n $-geodesically complete or future complete solutions of Einstein's field equations with smooth asymptotic structure. 
{\it Comm. Math. Phys.} 107 (1986), 587--609.
  
\bibitem{HF2}
Friedrich, H.: On the global existence and the asymptotic behavior of solutions to the Einstein-Maxwell-Yang-Mills equations.  {\it J. Differential Geom.} 34 (1991), 275--345.

\bibitem{Glassey}
Glassey, R. T.: The Cauchy problem in kinetic theory.
Society for Industrial and Applied Mathematics(SIAM), 1996.

\bibitem{Glassey06}
Glassey, R.~T.: Global solutions to the Cauchy problem for the relativistic Boltzmann equation with near-vacuum data.
{\it Comm. Math. Phys.} 264 (2006), no. 3, 705--724.


\bibitem{GS91}
Glassey, R. T., Strauss, W. A.: On the derivatives of the collision map of relativistic particles. 
{\it Transport Theory Statist. Phys.} 20 (1991), no. 1, 55--68.

\bibitem{HJ}
Horn, R. A., Johnson, C. R.: Matrix Analysis. 
Cambridge University Press, Cambridge, 2013.

\bibitem{Huterer}
Huterer, D.: A Course in Cosmology: From Theory to Practice.
Cambridge University Press, 2023.

\bibitem{IS84}
Illlner, R., Shinbrot, M.: The Boltzmann equation: global existence for a rare gas in an infinite vacuum.
{\it Comm. Math. Phys.} 95 (1984), no. 2, 217--226.

\bibitem{I63}
Israel, W.: Relativistic kinetic theory of a simple gas.
{\it J. Math. Phys.} 4 (1963), no. 9, 1163--1181.

\bibitem{Kozlikin}
Kozlikin, E., Lilow, R., Fabis, F., Bartelmann, M.: A first comparison of Kinetic Field Theory with Eulerian Standard Perturbation Theory.
{\it J. Cosmol. Astropart. Phys.} 6 (2021), 035.

\bibitem{HL}
Lee, H.:  Asymptotic behaviour of the Einstein-Vlasov system with a positive cosmological constant.
\newblock{Math. Proc. Cambridge Phil. Soc.} 137 (2004), 495--509.

\bibitem{L13}
Lee, H.: Asymptotic behaviour of the relativistic Boltzmann equation in the Robertson-Walker spacetime. 
{\it J. Differential Equations} 255 (2013), no. 11, 4267--4288.

\bibitem{LLN23}
Lee, H., Lee, J., Nungesser, E.: Small solutions of the Einstein-Boltzmann-scalar field system with Bianchi symmetry.
{\it J. Math. Phys.} 64 (2023), no. 1, Paper No. 011507.


\bibitem{LN171}
Lee, H., Nungesser, E.: Future global existence and asymptotic behaviour of solutions to the Einstein-Boltzmann system with Bianchi I symmetry.
{\it J. Differential Equations} 262 (2017), 5425--5467.

\bibitem{LN172}
Lee, H., Nungesser, E.: Late-time behaviour of Israel particles in a FLRW spacetime with $ \Lambda > 0 $.
{\it J. Differential Equations} 263 (2017), 841--862.

\bibitem{LL}
Lieb, E. H., Loss, M.: Analysis. Second edition. Graduate Studies in Mathematics, 14. American Mathematical Society, Providence, RI, 2001.

\bibitem{EB}
Nadkarni-Ghosh, S., Refregier, A.: The Einstein-Boltzmann equations revisited 
{\it Mon. Not. R. Astron. Soc.} 471, 2, 2391--2430.

\bibitem{ND06}
Noutchegueme, N., Dongo, D.: Global existence of solutions for the Einstein-Boltzmann system in a Bianchi type I spacetime for arbitrarily large initial data.
{\it Classical Quantum Gravity} 23 (2006), no. 9, 2979--3003.


\bibitem{NT06}
Noutchegueme, N., Takou, E.: Global existence of solutions for the Einstein-Boltzmann system with cosmological constant in the Robertson-Walker space-time.
{\it Commun. Math. Sci.} 4 (2006), no. 2, 291--314.

\bibitem{Ringstrom}
Ringstr{\"{o}}m, H.: On the topology and future stability of the universe.
Oxford University Press, Oxford, 2013.

\bibitem{RodSpeck}
Rodnianski, I., Speck, J.: The Stability of the Irrotational Euler-Einstein System with a Positive Cosmological Constant.
{\it J. Eur. Math. Soc.} 15, (2013) 2369--2462.

\bibitem{Speck}
Speck, J.: The nonlinear future stability of the FLRW Family of solutions to the Euler-Einstein system with a Positive Cosmological Constant. 
{\it Sel. Math. New Ser.} 18, (2012) 633--715.


\bibitem{Strain10}
Strain, R.~M.: Global Newtonian limit for the relativistic Boltzmann equation near vacuum.
{\it SIAM J. Math. Anal.} 42 (2010), no.~4, 1568--1601.

\bibitem{Wald83}
Wald, R. M.: Asymptotic behavior of homogeneous cosmological models in the presence of a positive cosmological constant. 
{\it Phys. Rev. D (3)} 28 (1983), no.~8, 2118--2120. 



\end{thebibliography}
\end{document}